\documentclass[11pt,a4paper]{article}
\usepackage{tikz}
\usepackage[compat=1.1.0]{tikz-feynman}
\usepackage{silence} 
\usepackage{jheppub}
\usepackage{graphicx}
\usepackage{orcidlink} 
\usepackage[makeroom]{cancel}
\usepackage{enumitem}
\usepackage{amssymb,amsmath,mathrsfs,bm,bbm,mathtools} 
\usepackage{amsthm}
\definecolor{ao(english)}{rgb}{0.0, 0.5, 0.0}
\hypersetup{colorlinks=true,linkcolor=blue,urlcolor=blue,citecolor=ao(english)}
\usepackage{stackrel}
\usepackage{multirow}

\newcommand\nn{\nonumber\\}

\newcommand{\bma}{\left(\begin{array}}
	\newcommand{\ema}{\end{array}\right)}
\newcommand{\be}{\begin{equation}}
	\newcommand{\ee}{\end{equation}}
\newcommand{\ben}{\begin{equation*}}
	\newcommand{\een}{\end{equation*}}
\newcommand{\ba}{\begin{eqnarray}}
	\newcommand{\ea}{\end{eqnarray}}
\newcommand{\ban}{\begin{eqnarray*}}
	\newcommand{\ean}{\end{eqnarray*}}
\newcommand{\bs}{\begin{subequations}}
	\newcommand{\es}{\end{subequations}}
\newcommand{\bc}{\begin{center}}
	\newcommand{\ec}{\end{center}}
\newcommand{\ve}{\varepsilon}

\def\Arg{{\rm Arg}}

\newcommand{\Pl}{{\text{\tiny Pl}}}

\newcommand{\Mpl}{M_\Pl}

\newcommand\bb{\bibitem}

\newcommand{\au}[2]{#1.~#2}
\newcommand{\arX}[1]{\href{http://arxiv.org/abs/#1}{{\cob arXiv:#1}}}
\newcommand{\oarX}[1]{\href{http://arxiv.org/abs/#1}{{\cob #1}}}

\newcommand{\book}[5]{\emph{#1}, #2, #3, #4 (#5)}
\newcommand{\books}[4]{\emph{#1}, #2, #3 (#4)} 
\newcommand{\doin}[6]{\href{http://dx.doi.org/#1}{{\cob {\it #2 #3} {\bf #4} (#6) #5}}}
\newcommand{\doinn}[5]{\href{http://dx.doi.org/#1}{{\cob {\it #2} {\bf #3} (#5) #4}}}
\newcommand{\doij}[5]{\href{http://dx.doi.org/#1}{{\cob {\it #2} {\bf #3} (#5) #4}}}
\newcommand{\ndoin}[6]{\href{#1}{{\cob {\it #2 #3} {\bf #4} (#6) #5}}}

\newcommand{\procsinm}[5]{in \emph{#1}, #2 eds., #3, #4 (#5)}

\newcommand{\procm}[6]{in \emph{#1}, #2 eds., #3, #4, #5 (#6)}
\newcommand{\tia}[1]{\textit{#1},}

\newcommand{\boxd}[1]{\boxed{\phantom{\Biggl(}#1\phantom{\Biggl)}}}

\renewcommand{\leq}{\leqslant}
\renewcommand{\geq}{\geqslant}
\newcommand{\Eq}[1]{(\ref{#1})}
\newcommand{\Eqq}[1]{eq.~(\ref{#1})}
\newcommand{\Eqqs}[1]{eqs.~(\ref{#1})}
\def\rme{e}
\def\rmd{d}
\def\rmi{i}

\def\Re{\text{Re}}
\def\Im{\text{Im}}

\def\a{\alpha}
\def\b{\beta}
\def\de{\delta}

\def\g{\gamma}
\def\la{\lambda}

\def\e{\epsilon}
\def\ve{\varepsilon}
\def\Om{\Omega}
\def\om{\omega}
\def\G{\Gamma}
\def\t{\tau}
\def\s{\sigma}

\def\vp{\varphi}

\def\B{\Box}
\def\H{{\rm H}}
\def\mst{M_*}

\def\lst{\ell_*}

\def\cC{\mathcal{C}}
\def\cD{\mathcal{D}}

\def\cF{\mathcal{F}}
\def\cG{\mathcal{G}}

\def\cM{\mathcal{M}}

\def\cO{\mathcal{O}}
\def\cP{\mathcal{P}}

\def\ups{u}
\def\p{\partial}
\def\pe{p_\textsc{e}}
\def\ke{k_\textsc{e}}
\def\xe{x_\textsc{e}}
\def\bpe{\bar p_\textsc{e}}
\def\bke{\bar k_\textsc{e}}

\def\cob{\color{blue}}

\newtheorem{theo}{Theorem}

\begin{document}
	
\renewcommand{\thefootnote}{\fnsymbol{footnote}}
	
\title{Perturbative unitarity of\\ fractional field theories and gravity} 

\author[a,*]{Gianluca Calcagni\,\orcidlink{0000-0003-2631-4588}\note{Corresponding author.}}
\emailAdd{g.calcagni@csic.es}
\affiliation[a]{Instituto de Estructura de la Materia, CSIC, Serrano 121, 28006 Madrid, Spain}

\author[b,c,d]{and Fabio Briscese\,\orcidlink{0000-0002-9519-5896}}
\emailAdd{fabio.briscese@uniroma3.it}
\affiliation[b]{Dipartimento di Architettura, Università Roma Tre, Via Aldo Manuzio 68L, 00153 Rome, Italy}
\affiliation[c]{Istituto Nazionale di Alta Matematica Francesco Severi, Gruppo Nazionale di Fisica Matematica, Piazzale Aldo Moro 5, 00185 Rome, Italy}
\affiliation[d]{Istituto Nazionale di Fisica Nucleare, Sezione di Roma 3, Via della Vasca Navale 84, 00146 Rome, Italy}
	
\abstract{Motivated by quantum gravity on spacetimes with multi-scale geometry, we analyze quantum field theories with a self-adjoint fractional power $(\Box^2)^{\gamma/2}$ of the d'Alem\-bert\-ian in the kinetic term, for any real $\gamma>0$. Selecting a particularly simple version of the kinetic term which we call hermitian polynomial, we study the spectral decomposition of the propagator and, when $\gamma>1$, obtain the standard mass singularity $-k^2=m^2$. This is the only mode in the perturbative spectrum of asymptotic states, since the only other content of the theory is a cloud of purely virtual particles with complex masses. We also show that other versions of the self-adjoint fractional kinetic term lead to a different distribution of the virtual complex modes but to the same physical spectrum for $0<\gamma<3$, thus addressing the issue of uniqueness in this class of nonlocal theories. The non-hermitian version of the theory has the $-k^2=m^2$ particle plus a continuum of standard massive modes. We prove that unitarity of scalar, gauge and gravity models is respected at all perturbative orders if, in the hermitian cases, one adopts the fakeon prescription on scattering amplitudes or, in the non-hermitian case, $0<\gamma<1$ or $2<\gamma<3$ with the standard Feynman prescription. These results drastically simplify previous characterizations of fractional quantum gravity, which is super-renormalizable for $\gamma>2$. Finally, we discuss violations of macro-causality in the massless sector of fractional QFT and a way to avoid them.}
	

\maketitle
\renewcommand{\thefootnote}{\arabic{footnote}}


\section{Introduction}\label{intro}

If we tried to enclose in a nutshell the main reasons behind the efforts to find a theory of quantum gravity \cite{Ori09,Fousp,Calcagni:2013hv,Bam24}, we could trace one to the failure in quantizing Einstein gravity as a perturbative quantum field theory (QFT) \cite{tHooft:1974toh,Deser:1974cz,Deser:1974cy,Deser:1974xq,Goroff:1985sz,Goroff:1985th,vandeVen:1991gw,Bern:2015xsa,Bern:2017puu}. The perceived impossibility to combine directly the two governing paradigms of modern theoretical physics, and thus to have a unified description of the four fundamental interactions, spurred a number of heterogeneous attempts to obtain a consistent theory. In some cases such as asymptotic safety or loop quantum gravity, Einstein gravity is kept as the classical basis and perturbative quantization is replaced by non-perturbative techniques. In other instances, the classical theory of gravity is modified but the perturbative QFT framework is kept unaltered or with minimal adaptations. This is the case of fractional quantum gravity (FQG) and the general class of fractional QFTs, the topic of this paper.

FQG is a type of perturbative quantum gravity with nonlocal operators of non-integer order \cite{Calcagni:2016azd,Calcagni:2021ipd,Calcagni:2021aap,Calcagni:2022shb,Calcagni:2025wnn}. Its action in four spacetime dimensions is
\be\label{gravac}
S=\frac{\Mpl^2}{2}\int\rmd^4x\,\sqrt{|g|}\,\left[R-c_2\lst^2 G_{\mu\nu}(\lst^4\B^2)^{\frac{\g}{2}-1}\,R^{\mu\nu}-c_0\lst^2 R(\lst^4\B^2)^{\frac{\g}{2}-1}\,R\right]+\dots,
\ee
where we work in spacetime signature $(-,+,+,+)$, $\Mpl=1/\sqrt{8\pi G}$ is the reduced Planck mass, $G$ is Newton's constant, $g=\det g_{\mu\nu}$, $R_{\mu\nu}$ and $R=g_{\mu\nu}R^{\mu\nu}$ are, respectively, the Ricci tensor and scalar defined in the usual way, $G_{\mu\nu}=R_{\mu\nu}-(1/2)g_{\mu\nu}R$ is the Einstein tensor, $\lst$ is a fundamental length scale marking the transition between a dynamics with second-order derivatives and one with fractional derivatives, $(\B^2)^{\g/2}$ with $\g>1$ is the non-integer  power of the square of the Laplace--Beltrami operator or d'Alembertian $\B$ \cite{Calcagni:2025wnn} and, finally, the ellipsis stands for the contribution of matter fields plus higher-curvature nonlocal terms (called killer operators \cite{Modesto:2014lga}) that do not modify the propagator but can make the theory finite \cite{Calcagni:2022shb}. $c_0$ and $c_2$ are dimensionless constants: $c_0$ is free, while $c_2=\pm 1$ is just a sign choice. In previous works \cite{Calcagni:2021aap,Calcagni:2022shb}, we chose $c_2=-1$ and $c_0=0$ but we will keep these couplings general until later and make the opposite choice of sign for $c_2$ for convenience. 

The action \Eq{gravac} arises from a top-down construction based on the following assumptions:
\begin{enumerate}
\item The dynamics and the parameters of the theory must be chosen so that to obtain a perturbatively unitary and renormalizable QFT of gravity and matter fields. This choice is a matter of taste. We want to use the conservative framework of perturbative QFT because we know not only how to pose a great number of physical questions within it, but also how to answer them. This does not exclude that the theory admit also a non-perturbative formulation \cite{Calcagni:2024xku}.
\item Spacetime and the field dynamics undergo dimensional flow, which is the change of spacetime dimensionality with the probed scale. This is part of a long-standing program \cite{Calcagni:2016azd,Calcagni:2021ipd,Calcagni:2009kc} aiming at implementing dimensional flow as a fundamental rather than as an emergent property. The reason to do this is that dimensional flow appears in all quantum-gravity scenarios \cite{tHooft:1993dmi,Carlip:2009kf,Calcagni:2016edi,Carlip:2017eud,Mielczarek:2017cdp,Carlip:2019onx} and is typically related to the good ultraviolet (UV) properties of these theories, although it is insufficient \emph{per se} to secure such properties \cite{Calcagni:2013qqa}.
\item The theory is covariant, diffeomorphism invariant and Lorentz invariant in any local inertial frame. These conditions are dictated by experience with theories with dimensional flow breaking Lorentz invariance. Their phenomenology is interesting but it is difficult to handle perturbative QFTs with non-trivial spacetime measures \cite{Calcagni:2016azd,Calcagni:2013qqa,Calcagni:2017jtf} or with nonlocal derivative operators breaking standard diffeomorphism symmetry \cite{Calcagni:2021aap,Calcagni:2018dhp}.
\item The classical action is self-adjoint, hence the theory admits a well-defined classical limit. This requirement dictates the use of the self-adjoint operator $(\lst^4\B^2)^{\g/2-1}$ in \Eq{gravac} \cite{Calcagni:2025wnn} instead of the operator $(-\lst^2\B)^{\g-2}$ employed in previous papers \cite{Calcagni:2021ipd,Calcagni:2021aap,Calcagni:2022shb}.
\end{enumerate}
In particular, assumption 3 forces the spacetime measure to be the ordinary Lebesgue measure, so that dimensional flow is realized only dynamically, i.e., through the kinetic terms. It also reduces the choice of derivative operators with non-trivial dimensional flow to $\B^\g$ and its self-adjoint extensions.

Tree-level unitarity and renormalizability at all orders of \Eq{gravac} with the non-hermitian operator $(-\lst^2\B)^{\g-2}$ have been studied in \cite{Calcagni:2021aap,Calcagni:2022shb}. However, the approximations used in those works are insufficient to fully understand unitarity, which is anyway altered in the self-adjoint formulation \Eq{gravac}. The purpose of this paper is to study the unitarity of the theory \Eq{gravac} and of the general class of fractional scalar and gauge QFTs at all perturbative orders starting from the tree level. The resulting propagator is multi-valued and must be represented on a complex Riemann surface, from which one extracts a sheet where to build the physical theory. The level of mathematical sophistication we meet on such surface is nothing compared with the one found in multi-particle scattering amplitudes in ordinary QFT and in the Standard Model (an example is \cite{Dawid:2023jrj}). Still, it may be highly unfamiliar due to the fact that such a structure appears at the tree level and that it describes fundamental particle states.

Our main findings are the following. 
\begin{itemize}
\item Some Riemann sheets are excluded because they contain unphysical modes that cannot be removed, such as unpaired modes with complex masses.
\item On physical sheets, the tree-level exact Green's function of fractional QFT for $\g>0$ is made of various contributions: (i) for $\g>1$, a mode corresponding to the standard Green's function $1/(k^2+m^2)$ that dominates in the infrared (IR); (ii) a finite number of complex-conjugate poles; (iii) a continuum of complex-conjugate pairs coming from a discontinuity of the Green's function on the imaginary axis $\mathbb{I}=\{z\,:\,\Re\,z=0\}$. Of these, only (i) enters the physical spectrum, while the rest are purely virtual contributions just like in the purely fractional case.
\item Complex-conjugate modes do not affect unitarity provided one defines the Lorentzian amplitude with the fakeon or Anselmi--Piva prescription \cite{Anselmi:2017yux,Anselmi:2017lia,Anselmi:2021hab,Anselmi:2022toe,Anselmi:2025uzj,Anselmi:2025uda}. Unitarity at all orders in perturbation theory is established easily and for the first time, from the structure of the tree-level propagator applied to an arbitrary amplitude.
\item It is always possible to eliminate all extra poles from the theory.
\item These results are independent of the representation used to define the operator $(\B^2)^{\g/2}$ and of the way $\B$ and $(\B^2)^{\g/2}$ are combined in the kinetic term. The only things that can change are the presence of real extra poles (absent in our main definition) and the number and position of complex poles in any given sheet, while the range $0<\g<3$ for which one can eliminate all extra poles is common to all hermitian formulations. Together with the analysis of the problem of initial conditions in the classical theory \cite{Calcagni:2025wnn}, this points towards the fact that, although different realizations of a fractional QFT with the same scaling $\g$ may differ at a technical level in their purely virtual sector, the physics remains the same. Therefore, despite nonlocal QFTs be non-unique due to the possibility to define mathematically inequivalent form factors, one can establish universality classes based on physical properties such as the number of initial conditions, the type of classical solutions or the asymptotic particle spectrum.
\item If one uses instead the non-hermitian operator $(-\B)^\g$, the continuum of massive modes is standard and does not require the fakeon prescription but unitary is respected only for $\g$ in the ranges $0<\g<1$ and $2<\g<3$, in the Riemann sheets with no extra poles.
\end{itemize}  

The structure of this paper aims at a rigorous, step-by-step derivation of these results, confining technical details into theorems and giving prominence to the physical interpretation. In section~\ref{sec2} , we recall what we know so far about the renormalizability and unitarity of fractional QFT in $D=4$ spacetime dimensions, taking a general scalar-field toy model with derivative interactions as an example parametrizing the scalar, gauge and gravitational sectors of the theory. A study of the singularities (branch point and poles) of the full Green's function, including both the IR and the UV terms, is done in section~\ref{sec3} for the so-called HP model, where we also show how to eliminate all the complex poles. Soon in the analysis, we fix the sign and value of the coefficient $c_2$ in \Eq{gravac} to $c_2=1$. In section~\ref{sec4}, we calculate the contributions of complex poles (if present), of the branch point and of the discontinuity of the Green's function and analyze tree-level unitarity. In particular, in section~\ref{cocopo}, the problems carried by complex poles are discussed in some detail. Section~\ref{sec5} is devoted to the issue of how the physics changes if we modify our definition of the fractional kinetic term while maintaining its order $\g$. In section~\ref{othrep}, we consider the non-hermitian $(-\B)^\g$ and two alternative implementations of $(\B^2)^{\g/2}$ and show that the distribution of undesired or virtual degrees of freedom does change across these cases but not in a way altering the physical spectrum, which is the same as in the HP model. In section~\ref{compa}, we make a comparison between fractional QFT and nonlocal QFT with entire form factors. Perturbative unitarity at all orders is derived in section~\ref{sec6}. Section~\ref{sec7} deals with causality violations in the massless sector of this class of theories. Conclusions are in section~\ref{sec8}. 

Appendix~\ref{appA} contains a minor calculation of a residue. Appendices~\ref{appB} and \ref{appC} collect a sequence of theorems on the existence and distribution of poles in the Green's function for various hermitian and non-hermitian definitions of the multi-fractional kinetic term. These contents are very technical and are reported only for the sake of transparency. We recommend the user \emph{not} to read them, with the exception of Theorem~\ref{theo6-A} which is pivotal to the construction of a theory free of extra poles. 

A very condensed presentation of the theory, with also some considerations on black holes and gravitational waves, can be found in \cite{Briscese:2026jcf}.


\section{Fractional QFT: renormalizability and unitarity (so far)}\label{sec2} 

A toy model capturing the main renormalizability properties of scalar, gauge and gravitational QFTs is \cite{Calcagni:2022shb}
\ba
&&S=\!\int\!\rmd^{4}x\left[\frac{1}{2}\phi\left\{c_1(\B-m^2)-\frac{c_2}{\lst^2}[\lst^2(m^2-\B)]^{\g}\right\}\phi-\phi^{\frac{N}{2}}\left[\la_0+\la_{\g-\a}(-\lst^2\B)^{\g-\a}\right]\phi^{\frac{N}{2}}\right],\nn
&&\g>\a\geq 0\,,\label{eq:modelnh}
\ea
where $\phi$ is a real scalar with energy dimensionality $[\phi]=1$, the constant coefficients in the action have dimensionality $[c_1]=0=[c_2]$, $[m]=1$, $[\la_0]=[\la_{\g-\a}]=4-N$, $\mathbb{N}\ni N\geq 3$, the coefficient $c_1$ can take the values 0 (purely fractional model) or 1 (realistic model), $\g>0$ is a non-integer exponent, $\B=\p_\mu \p^\mu$ is the d'Alembertian (in momentum space, $\B\to -k^2=(k^0)^2-|\bm{k}|^2$) and $(-\B)^\g$ is a nonlocal operator called fractional d'Alembertian \cite{BGG,Mar91,Gia91,BGO,BG,doA92,Barci:1995ad,Barci:1996br,Barci:1996ny,Barci:1998wp} that admits different representations \cite{Calcagni:2025wnn}. Here we update the above toy model to one with self-adjoint derivative operators:
\ba
&&S=\!\int\!\rmd^{4}x\left(\frac{1}{2}\phi\left\{c_1(\B-m^2)-\frac{c_2}{\lst^2}[\lst^4(m^2-\B)^2]^{\frac{\g}{2}}\right\}\phi-\phi^{\frac{N}{2}}\left[\la_0+\la_{\g-\a}(\lst^4\B^2)^{\frac{\g-\a}{2}}\right]\phi^{\frac{N}{2}}\right),\nn
&&\g>\a\geq 0\,.\label{eq:model}
\ea

For the fractional operators to dominate in the UV, it must be $c_1\neq 0$ and $\g>1$, so that the kinetic term tends asymptotically to $\sim \phi(\B^2)^{\g/2}\phi$ at large momenta, while in the IR one recovers the canonical kinetic term $\sim \phi\B\phi$. If $c_1\neq 0$ and $\g<1$, then the UV and IR limits are interchanged. The free-level\footnote{Free-level means in the absence of interactions. For the propagator $\rmi\tilde G_*$ in momentum space, this is a synonym for tree-level but, in general, it is not. For example, the K\"all\'en--Lehmann representation of the time-ordered two-point function contains information of the interactions even at the tree level \cite{Briscese:2024tvc}.} Green's function in momentum space from \Eq{eq:model} is
\be\label{propkh}
\tilde G_*(-\bar k^2) =\frac{\lst^2}{c_1\bar k^2+c_2(\bar k^4)^{\frac{\g}{2}}}\,,\qquad \bar k^2\coloneqq \lst^2(k^2+m^2)\,.
\ee
Here we choose the convention $\B G_*=-\de$ for the Green's function.

The scalar models \Eq{eq:modelnh} and \Eq{eq:model} parametrize all the main sectors of fractional QFT with $N$ and $\a$. Apart from a scalar field with a nonlocal potential, they encode the following perturbative field content in a simplified way (i.e., with a trivial tensorial structure):
\begin{itemize}
\item $\a=\g$: scalar field with local $\phi^N$ potential, $N\geq 3$ \cite{Trinchero:2012zn,Trinchero:2017weu,Trinchero:2018gwe,Calcagni:2021ljs}.
\item $\a=1$, $N=4$, $m=0$: longitudinal mode of a non-Abelian gauge field \cite{Calcagni:2022shb}.
\item $\a=0$, $N=4$, $m=0$: graviton \cite{Calcagni:2022shb}.
\end{itemize}

In particular, the Green's function $\tilde G_{\mu\nu\s\t}$ for the graviton $h_{\mu\nu}=g_{\mu\nu}-\eta_{\mu\nu}$ can be found from the general expression valid for any (local or nonlocal) form factors $\cF_R(\B)$ and $\cF_2(\B)$ in $D$ dimensions \cite{Calcagni:2024xku}. From the generic action
\be\label{gravac2}
S=\frac{\Mpl^2}{2}\int\rmd^4x\,\sqrt{|g|}\,\left[R+ R\cF_R(\B)\,R+ G_{\mu\nu}\cF_2(\B)\,R^{\mu\nu}\right]\,,
\ee
in the case of \Eq{gravac} we have
\be
\cF_R(\B)=-c_0\lst^2 (\lst^4\B^2)^{\frac{\g}{2}-1}\,,\qquad \cF_2(\B)=-c_2\lst^2 (\lst^4\B^2)^{\frac{\g}{2}-1}\,,
\ee
and it is easy to see from \cite{Calcagni:2024xku} that the gauge-independent part of the Green's function in four dimensions is
\ba
\hspace{-1cm}\tilde G_{\mu\nu\s\t}(-k^2) &=& \frac{4}{\Mpl^{2}} \Bigg(\frac{{P}^{(2)}}{k^2 \left[1-k^2 \cF_2(-k^2) \right]}-\frac{{P}^{(0)}}{2k^2 \left\{1 + 2k^2 \left[3\cF_R(-k^2) + \cF_2(-k^2)\right]\right\}}\Bigg) \nn
&& + \textrm{gauge}\nn
&=&\frac{4\lst^2}{\Mpl^{2}} \left\{\frac{P^{(2)}_{\mu\nu\s\t}}{\lst^2k^2+c_2(\lst^4k^4)^{\frac{\g}{2}}}-\frac{P^{(0)}_{\mu\nu\s\t}}{2 \left[\lst^2k^2-2(3c_0+c_2)(\lst^4k^4)^{\frac{\g}{2}}\right]}\right\} + \textrm{gauge}\,,\label{gravprop}
\ea
where 
\ba
&& P^{(2)}_{\mu\nu\rho\s} \coloneqq \frac{1}{2} \left(\Theta_{\mu \rho} \Theta_{\nu \s} +  \Theta_{\mu\s} \Theta_{\nu\rho} \right)
- \frac{1}{3} \Theta_{\mu\nu} \Theta_{\rho\s} \,, \\
&&  P^{(0)}_{\mu\nu\rho\s} \coloneqq \frac{1}{3} \Theta_{\mu \nu} \, \Theta_{\rho \s}\,,\qquad \Theta_{\mu \nu} \coloneqq \eta_{\mu \nu} - \frac{k_{\mu} k_{\nu}}{k^2}\,,
\ea
are the Barnes--Rivers projectors in four dimensions in momentum space \cite{Riv64,Bar65,VanNieuwenhuizen:1973fi}.\footnote{In \cite{Calcagni:2021aap}, the propagator for $h_{\mu\nu}$ was found without decomposing it with the projectors.} The spin-2 part of \Eq{gravprop} has the same form as \Eq{propkh} with $c_1>0$. The spin-0 part has opposite sign but it can be gauged away and, therefore, does not constitute a problem in perturbation theory if the spin-2 term is not a ghost \cite{Calcagni:2024xku,Ordonez:1985kz,Ham09}.


\subsection{Renormalizability}\label{toyren} 

Let us recall that a QFT is non-renormalizable when it has divergences at all loop orders $L$ and their number grows with $L$ (Einstein gravity is an example); strictly renormalizable when it has divergences at all loop orders but their number does do not grow with $L$ and they can be reabsorbed by counter-terms of structurally the same form as terms present in the original action (Stelle gravity is an example); super-renormalizable when there are only a finite number of divergences up to some loop order $L_{\rm max}$, without counting divergent sub-diagrams; one-loop super-renormalizable when $L_{\rm max}=1$; finite if it has no divergences at any loop order. 

In the case of fractional QFT, an extension of Weinberg's theorem for local QFT \cite{Weinberg:1959nj} holds and the counter-terms to be added at one-loop level are local, finite in number, of dimensionality $\leq 4$ in four dimensions and identical to those for a standard local QFT \cite{Calcagni:2022shb}. In the scalar model \Eq{eq:model}, these counter-terms $\cO_{\b,i}$, $i=1,2,3,4$, are $\Lambda_{\rm cc}$ (cosmological constant), $\phi\B\phi$ (standard kinetic term), $m^2\phi^2$ (mass term) and $\phi^N$ (self-interaction). In the case of pure gravity, they are the curvature-squared terms $\cO_{\b,i}=R^2,\,R_{\mu\nu}R^{\mu\nu},\,R_{\mu\nu\s\t}R^{\mu\nu\s\t}$. Moreover, the Bogoliubov--Parasiuk--Hepp--Zimmermann (BPHZ) renormalization scheme \cite{Lowenstein:1975ug,Piguet:1995er} also holds \cite{Calcagni:2022shb}, which implies that all sub-divergences at any order in perturbation theory are accounted for by the above counter-terms found at one-loop order. Therefore, power counting is sufficient to establish when the theory is renormalizable as a function of $\g$. However, it is not sufficient to conclude when it is non-renormalizable, since there may be unforeseen cancellation mechanisms in action in Feynman diagrams. Hence, all the cases marked as ``non-renormalizable'' below should carry a question mark.

With this specification in mind, a power-counting analysis \cite{Calcagni:2021ljs,Calcagni:2022shb} shows that the theory is renormalizable when $\g\geq1$ and $\a\geq(N-2)/2$, when $1<\g<2$ and $(2-\g)(N-2)/2<\a<(N-2)/2$, or when $\g\geq 2$ and $0\leq \a<(N-2)/2$. More precisely:
\begin{itemize}
\item $\bm{\g\geq 1}$. The UV limit of the model is fractional for $\g>1$ and standard for $\g=1$ and we have:
\begin{itemize}
\item Fractional scalar theory ($\a=\g$) with $\phi^N$ potential for $N=3$ is
finite.
\item Fractional scalar theory ($\a=\g$) with $\phi^N$ potential for $N=4$ is
\bs\label{sreno}\ba
\hspace{-2.6cm}&&\text{strictly renormalizable:\hspace{2.55cm} $\g=1$}\,,\\
\hspace{-2.6cm}&&\text{super-renormalizable:\hspace{2.15cm} $1<\g\leq\frac{4}{3}$}\,,\\
\hspace{-2.6cm}&&\text{one-loop super-renormalizable:\hspace{.6cm} $\frac{4}{3}<\g\leq 2$}\,,\\
\hspace{-2.6cm}&&\text{finite:\hspace{4.9cm} $2<\g$}\,.
\ea\es	
\item Fractional scalar theory ($\a=\g$) with $\phi^N$ potential for $N\geq 5$ is
	\bs\label{sreno2}\ba
	\hspace{-1.4cm}&&\text{non-renormalizable:\hspace{4.55cm} $\g<\frac{2(N-2)}{N}$}\,,\\
	&&\text{strictly renormalizable:\hspace{4.cm} $\g=\frac{2(N-2)}{N}$}\,,\\
	&&\text{super-renormalizable:\hspace{2.15cm} $\frac{2(N-2)}{N}<\g\leq \frac{4(N-2)}{N+2}$}\,,\\
	&&\text{one-loop super-renormalizable:\hspace{.6cm} $\frac{4(N-2)}{N+2}<\g\leq N-2$}\,,\\
	\hspace{-2.6cm}&&\text{finite:\hspace{5.4cm} $N-2<\g$}\,.
	\ea\es	
\item Fractional gauge theory ($\a=1$, $N=4$, $m=0$) is
	\bs\label{greno}\ba
	\hspace{-2.7cm}&&\text{strictly renormalizable:\hspace{2.55cm} $\g=1$}\,,\\
	\hspace{-2.7cm}&&\text{super-renormalizable:\hspace{2.15cm} $1<\g\leq 2$}\,,\\
	\hspace{-2.7cm}&&\text{one-loop super-renormalizable:\hspace{.6cm} $2<\g$}\,.
	\ea\es
\item Fractional gravity ($\a=0$, $N=4$, $m=0$) is
	\bs\label{reno}\ba
	\hspace{-2.7cm}&&\text{non-renormalizable:\hspace{3.15cm} $\g<2$}\,,\\
	\hspace{-2.7cm}&&\text{strictly renormalizable:\hspace{2.6cm} $\g=2$}\,,\\
	\hspace{-2.7cm}&&\text{super-renormalizable:\hspace{2.15cm} $2<\g\leq 4$}\,,\\
	\hspace{-2.7cm}&&\text{one-loop super-renormalizable:\hspace{.6cm} $4<\g$}\,.\label{1lsu}
	\ea\es
\end{itemize}
\item $\bm{0<\g<1}$, $\bm{c_0=1}$. The UV limit of the model is standard and the renormalization properties are the above for $\g=1$. 
\item $\bm{0<\g<1}$, $\bm{c_1=0}$. The model is purely fractional (no term $\phi\B\phi$ in the actions \Eq{eq:modelnh} and \Eq{eq:model}) and we have:
\begin{itemize}
\item Fractional scalar theory ($\a=\g$) with $\phi^N$ potential for $N=3$ is
	\bs\label{sreno3}\ba
	\hspace{-2.6cm}&&\text{non-renormalizable:\hspace{3.2cm} $\g<\frac{2}{3}$}\,,\\
	\hspace{-2.6cm}&&\text{strictly renormalizable:\hspace{2.65cm} $\g=\frac{2}{3}$}\,,\\
	\hspace{-2.6cm}&&\text{super-renormalizable:\hspace{2.15cm} $\frac{2}{3}<\g\leq \frac{4}{5}$}\,,\\
	\hspace{-2.6cm}&&\text{one-loop super-renormalizable:\hspace{.6cm} $\frac{4}{5}<\g<1$}\,.
	\ea\es
\item Fractional scalar theory ($\a=\g$) with $\phi^N$ potential for $N\geq 4$ is non-renormal\-i\-za\-ble.
\item Fractional gauge theory ($\a=1$, $N=4$, $m=0$) is non-renormalizable.
\item Fractional gravity ($\a=0$, $N=4$, $m=0$) is non-renormalizable.
\end{itemize}
\end{itemize}
In all the renormalizable cases, there is a finite number of non-vanishing beta functions; for a single-field model, there are at most as many as the coefficients of the above counter-terms $\cO_{\b,i}$. One can make the theory finite by adding a finite set of operators $\cO_{{\rm K},i}$ called killers that do not change the propagator but make these beta functions vanish at all orders where divergences appear. These operators are local in quantum gravity with asymptotically polynomial form factors \cite{Modesto:2014lga} and nonlocal in FQG \cite{Calcagni:2022shb}. Each beta function is a power series in $\hbar$ and has to be compensated to zero at any loop order by a new operator, always of the same functional form $\cO_{{\rm K},i}$ but with a different frontal coefficient adjusting the beta function at the corresponding loop level. The only difference between the strictly renormalizable and the super-renormalizable case is in the number of terms in the series, respectively, infinite and finite.

All these results were obtained for the non-hermitian fractional QFT \Eq{eq:modelnh} but they hold also for the model \Eq{eq:model} (hence for the theory \Eq{gravac}), since they rely on a power counting not affected by the property of hermiticity.
 
 
\subsection{Unitarity (so far)}\label{toyuni} 
 
Any Green's function can be expressed as a Cauchy integral in the $(\Re\,z,\Im\,z)$ complex plane. We used this general observation to study the propagator of non-hermitian fractional QFT, defined in Euclidean signature and where all amplitudes are analytically continued \emph{\`a la} Efimov  to Lorentzian signature \cite{Pius:2016jsl,Briscese:2018oyx,Chin:2018puw,Efimov:1967dpd,Koshelev:2021orf,Buoninfante:2022krn}. For simplicity, a purely fractional Green's function was considered ($c_1=0$):
\be\label{propzuv}
\tilde G(z) =\frac{1}{c_2(-z)^\g}\,.
\ee
The Cauchy integral associated with this function is
\be\label{opt}
\tilde G(-\bar k^2)=\frac{1}{2\pi\rmi}\oint_{\G}\rmd z\,\frac{\tilde G(z)}{z+\bar k^2}\,,\qquad \tilde G(z)\coloneqq\frac{\tilde G_*(z)}{\lst^2}\,,
\ee
where $\G$ is any contour encircling only the pole of the integrand at $z=-\bar k^2$ and within and on top of which $\tilde G(z)$ is holomorphic. Calculated on the contour $\G_z$ shown in figure~\ref{fig1}, \Eq{opt} gives the spectral representation 
\be\label{spera1}
\frac{1}{c_2(\bar k^2)^\g}=\int_0^{+\infty}\rmd s\,\frac{\rho(s)}{s+\bar k^2}\,,\qquad\rho(s) = \frac{1}{c_2\pi}\frac{\sin(\pi\g)}{s^\g}\,.
\ee
It is possible to achieve tree-level unitarity but is technically complicated \cite{Calcagni:2022shb} and eventually requires to impose the fakeon prescription \cite{Anselmi:2021hab,Anselmi:2025uzj} projecting out ghost modes from the asymptotic spectrum order by order in perturbation theory.\footnote{The contributions of the circle at infinity and around the branch point $z=0$ vanish for $\g>0$ and $\g<1$, respectively. The interval $0<\g< 1$ guarantees that the tree-level spectral density $\rho(s)$ is positive-definite
for all $s$, provided $c_2>0$ as we assumed; the same condition giving a non-negative spectral density $\rho(s)\geq 0$ also holds at one loop \cite{Trinchero:2012zn,Trinchero:2017weu,Trinchero:2018gwe,Calcagni:2021ljs}. Therefore, one has to exclude the case $\g>1$, which would require $c_2<0$ to get $\rho(s)\geq 0$ but would make the contribution of the branch point divergent. However, for fractional gauge theories and gravity the interval $0<\g< 1$ corresponds either to a fractional operator dominating in the IR (in which case renormalizability is not improved with respect to standard QFT with $\g=1$) or to a purely fractional model (in which case only a $\phi^3$ scalar theory would be renormalizable for $\frac{2}{3}<\g<1$, while the other sectors would not), neither of which recovers Einstein gravity in the IR. For this reason, a mass regularization or splitting procedure for the Green's function \Eq{propzuv} was proposed in \cite{Calcagni:2022shb}, where the branch cut is split into two cuts lying on the imaginary axis $\mathbb{I}$, so that each cut is characterized by an exponent $0<\g_i<1$ such that $\g_1+\g_2=\g$. Then, the cuts are rotated back to the original one. This regularization extends the interval $0<\g<1$ for tree-level unitarity to $0<\g<2$,
which covers the interval $1\leq\g<2$ and all its sub-intervals for the renormalizability of a $\phi^N$ scalar theory and of gauge theory. However, one does not yet reach the renormalizability range $\g>2$ for gravity, which then requires another iteration of the same splitting procedure plus the fakeon prescription on the propagator, leading to tree-level unitarity for any $\g>0$ \cite{Calcagni:2022shb}.}
\begin{figure}[ht]
	\bc
	\includegraphics[width=9cm]{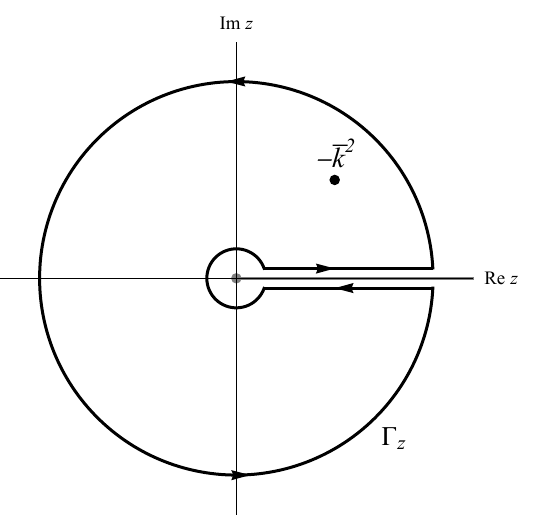}
	\ec
	\caption{\label{fig1} Contour of the Cauchy representation \Eq{opt} of the Green's function \Eq{propzuv}.}
\end{figure}  
\begin{figure}[ht]
	\bc
	\includegraphics[width=9cm]{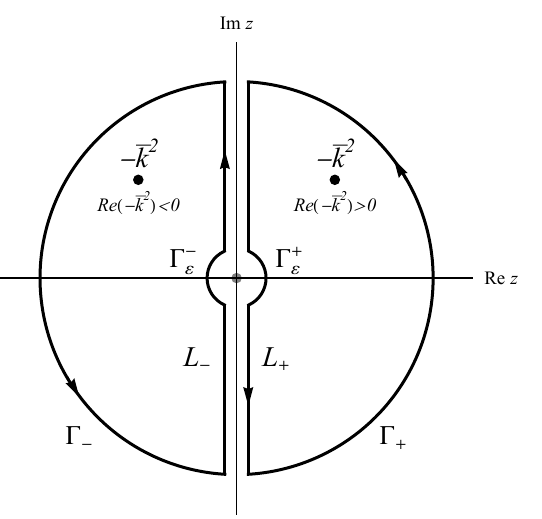}
	\ec
	\caption{\label{fig2} Contour of the Cauchy representation of the Green's functions \Eq{propzuv2} and \Eq{propz}. $\G_+$ and $\G_-$ are mutually disjoint pieces covering, respectively, the $\Re\,z>0$ and the $\Re\,z<0$ half-plane. The vertical lines $L_\pm$ run along the discontinuity at $\Re\,z=0$, while $\G_\ve^\pm$ circle around the branch point at $z=0$.}
\end{figure}  
\begin{figure}[ht]
	\bc
	\includegraphics[width=9cm]{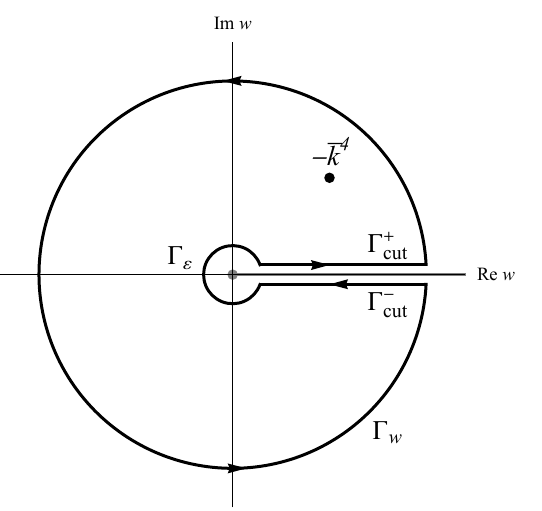}
	\ec
	\caption{\label{fig3} Contour of the Cauchy representation of the Green's functions \Eq{Gw} and \Eq{propw2} in the plane spanned by $w=-z^2$.}
\end{figure}  

The self-adjoint version of the theory is much more straightforward to deal with. Focusing again onto the UV limit of the propagator, 
\be\label{propzuv2}
\tilde G(z) =\frac{1}{c_2(z^2)^\frac{\g}{2}}\,,
\ee
it is not difficult to see that the propagator is represented by a non-analytic expression involving the Cauchy integral on two disjoint contours $\G_\pm$ \cite{Calcagni:2025wnn}:
\ba
\tilde G(-\bar k^2)&=& \frac{\Theta[\Re(-\bar k^2)]}{2\pi\rmi}\int_{\G_+}\rmd z\,\frac{\tilde G_+(z)}{z+\bar k^2}+\frac{\Theta[\Re(\bar k^2)]}{2\pi\rmi}\int_{\G_-}\rmd z\,\frac{\tilde G_-(z)}{z+\bar k^2}\nn
&\eqqcolon& \Theta[\Re(-\bar k^2)]\,\tilde G_+(-\bar k^2)+\Theta[\Re(\bar k^2)]\,\tilde G_-(-\bar k^2)\,,\label{Gz2}
\ea
where $\tilde G_\pm(z) =c_2^{-1}(\pm z)^{-\g}$, $\Theta$ is Heaviside step function defined with the convention
\be\label{heavi}
\Theta(x)=\left\{\begin{matrix}
1\qquad(x> 0)\,,\\   
0\qquad(x\leq 0)\,,
\end{matrix}\right.
\ee
and the contours $\G_\pm$ are shown in figure~\ref{fig2}. Taking \Eq{propzuv2}, the left-hand side of \Eq{Gz2} is $\tilde G(-\bar k^2)=1/[c_2(\bar k^4)^\frac{\g}{2}]$. The right-hand side can be calculated directly on the disconnected contour $\G_\pm$ in the figure and the final result is \cite{Calcagni:2025wnn}
\be\label{KLfin}
\frac{1}{c_2(\bar k^4)^\frac{\g}{2}}=\int_0^{+\infty}\rmd s\,\frac{\rho(s)}{s^2+\bar k^4}\,,\qquad \rho(s)=\frac{2}{c_2\pi}\frac{\sin\frac{\pi\g}{2}}{s^{\g-1}}\,.
\ee
An alternative and faster method to reach the same result, which we also employ in this paper, is to consider the Cauchy representation of $\tilde G(-k^2)$ in the complex plane of $w=-z^2$:
\be\label{Gw}
\tilde G(-\bar k^2)=\frac{1}{2\pi\rmi}\oint_{\G_w}\rmd w\,\frac{\tilde G(w)}{w+\bar k^4}\,,\qquad \tilde G(w)=\frac{1}{c_2(-w)^{\frac{\g}{2}}}\,,
\ee
where $\G_w$ is the same contour of figure~\ref{fig1} but in the $w$-plane and with pole $-\bar k^4$ instead of $-\bar k^2$ (figure~\ref{fig3}). The result is, of course, a trivial rewriting of \Eq{spera1}:
\ben\label{spera2}
\tilde G(-\bar k^2)=\int_0^{+\infty}\rmd t\,\frac{\rho(t)}{t+\bar k^4}\,,\qquad\rho(t) = \frac{2}{c_2\pi}\frac{\sin\frac{\pi\g}{2}}{t^\frac{\g}{2}}\,.
\een
With the change of variables $t=s^2$, we recover \Eq{KLfin}. 

The trick, illustrated in figure~\ref{fig4}, is performed by the multi-valued change of variable $w=-z^2$. When passing from the $w$-plane to the $z$-plane, angles are halved and the $w$-plane is squeezed into a half-plane and mirrored across the imaginary $z$-axis. The contour $\G_w$ is opened like a book and rotated by 90 degrees counter-clockwise (respectively, clockwise) and mapped into $\G_+$ (respectively, $\G_-$) in the $z$-plane. Thus, the branch point $w=0$ is mapped in the origin $z=0$; the infinitesimal circle $\G_\ve$ around it is mapped into the half-circle $\G_\ve^+$ on the right half of the $z$-plane and on the half-circle $\G_\ve^-$ on the left half; the branch cut on $\Im\,w=0$, $\Re\,w>0$ is split open into the discontinuity on $\Im\,z=0$; and the straight paths $\G_{\rm cut}^+\cup \G_{\rm cut}^-$ parallel to the cut are mapped into $L_+$ on one half and on $L_-$ on the other.
\begin{figure}[ht]
	\bc
	\includegraphics[width=13cm]{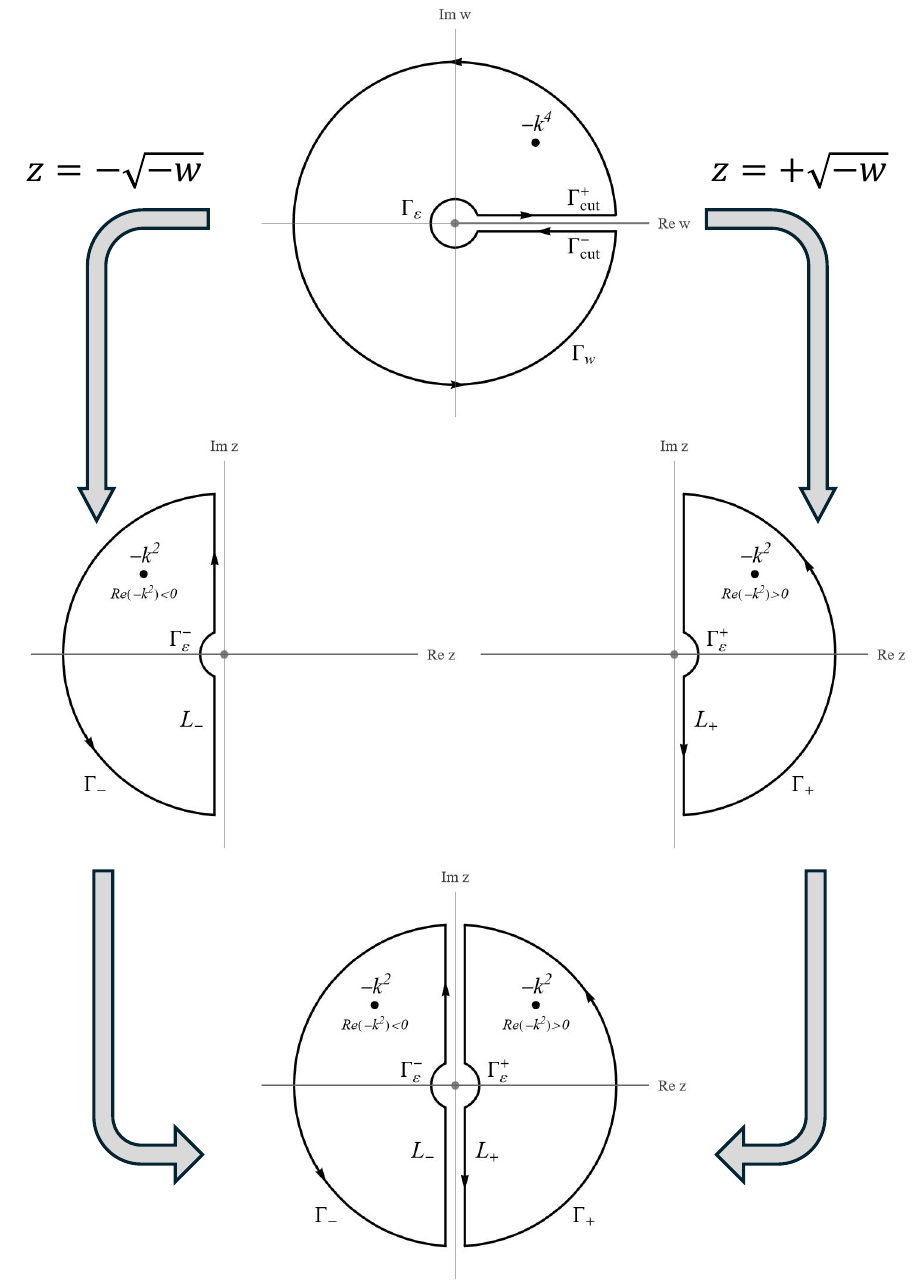}
	\ec
	\caption{\label{fig4} Mapping from the $w$-plane (figure~\ref{fig3}) to the $z$-planes with $z=\pm\sqrt{-w}$ and their combination into the disjoint contour of figure~\ref{fig2}. Bars on $k$ are omitted.}
\end{figure}  

Unlike \Eq{spera1} where the spectral density $\rho(s)$ weighs standard propagators $1/(s+\bar k^2)$ with mass $s$, \Eqq{KLfin} is an infinite superposition of complex-conjugate pairs with quartic propagator $1/(s+\bar k^4)$. These modes are projected out via the fakeon procedure, which, contrary to other prescriptions, simultaneously preserves the optical theorem and Lorentz invariance \cite{Anselmi:2025uzj}. For this reason, the sign of $c_2$ in \Eq{propzuv2} is irrelevant and we do not need to worry about positivity of $\rho(s)$ in this model.

Since the relation between unitarity, complex poles and Lorentz invariance is very important for the self-consistency of FQG, we delve on it more in detail. The problem was first pointed out in Lee--Wick theories \cite{Lee:1969fy,Lee:1970iw,Lee:1969zze} by Nakanishi \cite{Nakanishi:1971jj} but is present in any theory with complex poles adopting the Lee--Wick--Nakanishi prescription \cite{Anselmi:2025uzj}. Although a sector with kinetic term $\B^2+M^4$ and all field indices properly contracted is clearly Lorentz invariant at the tree level, and the theory itself is invariant at any order under infinitesimal Lorentz transformations \cite{Lee:1969zze}, Lorentz symmetry is broken under finite transformations already at the one-loop level \cite{Nakanishi:1971jj}. The main cause, when poles at finite distance from the real axis are present, is Lee--Wick's prescription for complex poles: integration of $k^0\in\mathbb{C}$ on a path $\G$ in the complex plane $\mathbb{C}$ and integration of spatial momenta $\bm{k}\in\mathbb{R}^3$ on real values,
\be\label{ultra}
\int_\mathbb{R}\rmd k^0\int_{\mathbb{R}^3}\rmd \bm{k}\to \int_\G\rmd k^0\int_{\mathbb{R}^3}\rmd \bm{k}\,.
\ee
This setting is the same as in Efimov analytic continuation, which is used in nonlocal theories. Equation~\Eq{ultra} also characterizes the definition of ultra-distributions \cite{Seb58,Has61,Morimoto:1975vi}, which is the space of functionals on which fractional differential operators act \cite{Barci:1996br,Barci:1998wp,Calcagni:2025wnn}. Therefore, the same problem of Lee--Wick theories arises also in fractional QFT and, in particular, fractional gravity. 
Indeed, there are no symmetry violations in theories with entire form factors \cite{Briscese:2018oyx} but, unfortunately, fractional form factors are not entire. To restore Lorentz invariance, one can modify the prescription \Eq{ultra} and integrate also spatial momenta on complex paths \cite{Cutkosky:1969fq}:
\be\label{ultra2}
\int_\mathbb{R}\rmd k^0\int_{\mathbb{R}^3}\rmd \bm{k}\to \prod_{\mu=0}^3\int_{\G^{(\mu)}}\rmd k^\mu\,.
\ee
An explicit Lorentz-invariant procedure of this type is the fakeon prescription in its first version \cite{Anselmi:2017lia,Anselmi:2025uzj}, which is equivalent to the much simpler but non-analytic operation of replacing the Euclidean external momentum $\pe^2$ in the Euclidean amplitude with its Lorentzian analytic continuations $p^2\pm\rmi\e$ and then taking the average \cite{Anselmi:2025uzj} (see also section~\ref{pvmo}). For amplitudes with a mirror symmetry $\cM(\pe^2)=\cM(-\pe^2)$, the operation is even simpler and becomes the replacement $\pe^2\to p^2$ \cite{Anselmi:2025uzj,Liu:2022gun}.

The approximations \Eq{propzuv} and \Eq{propzuv2} were done with the idea of capturing the new physics coming from the fractional part of the propagator. However, while the renormalizability analysis depends on the dimension $D$ of spacetime and can be performed in the UV limit of the operators in the action, the unitarity analysis depends neither on $D$ nor on any particular energy scale limit. If present, ghosts can circulate in loops and affect the physics at all energies. Indeed, the analyses based on \Eq{propzuv} \cite{Calcagni:2022shb} and \Eq{propzuv2} \cite{Calcagni:2025wnn} are only preliminary and the propagator with $c_1=1$ must be used in order to fully understand the unitarity properties of the theory as a function of $\g$. We will do so directly in the self-adjoint version \Eq{gravac} of FQG.


\section{Tree-level Green's function}\label{sec3}

In this section, we explore the spectrum of a generalization of the self-adjoint toy model \Eq{eq:model}. The free-level Green's function in momentum space from \Eq{eq:model} is \Eq{propkh}. However, it turns out that this operator of type ``$\B+(\B^2)^{\g/2}$'' leads to some limitations to the allowed values of $\g$ that preclude one-loop super-renormalizability (see section~\ref{sec5}). For this reason, in this section we consider a kinetic term ``$\B[1+(\B^2)^{(\g_1-1)/2}+(\B^2)^{(\g_2-1)/2}+\dots+(\B^2)^{(\g-1)/2}]$'' with some intermediate fractional operators of order $0<\g_1<\g_2<\dots<\g$, all positive because we want to avoid IR corrections to the dynamics. Despite the appearance, this kinetic term has a simpler structure than the other. Thus, we hereby study the singularities of the Green's function associated with this kinetic term, which we write in the complex plane with what we call the \emph{hermitian polynomial} (HP) definition:
\be\label{propz}
\boxd{\textrm{HP:}\qquad\tilde G(z) =\frac{1}{-z\big[1+c_2 (z^2)^{\frac{\om}{2}}\big]^\ups}\,,\qquad \g\equiv \ups\,\om+1\,,\quad \ups\in\mathbb{N}^+\,,}
\ee
where we recast the leading anomalous scaling $\g$ in terms of the integer power $\ups\geq 2$ (the $u=1$ case will be discussed in section~\ref{sec5}) of a fractional power $\om\notin \mathbb{N}$. The sign of $\om$ determines whether $\g\lessgtr 1$:
\bs\label{gom}\ba
&&-\frac12<-\frac{1}{\ups}<\om<0\,:\qquad 0<\g<1 \,,\label{gom1}\\
&&\hphantom{-\frac12<-\frac{1}{\ups}<}\,\om>0\,:\qquad\,\,\hphantom{0<}\g>1\,.\label{gom2}
\ea\es
The definition \Eq{propz} is called hermitian because such is the kinetic operator (which is actually not only hermitian but also self-adjoint) and polynomial because the denominator features a polynomial of $z^\om$. The singling out of $\om$ and $\ups$ in \Eq{propz} is arguably most natural, since fractional powers of operators are always defined as the composition of an integer power and a small exponent $|\om|<1$ \cite{Calcagni:2025wnn}.\footnote{Such composition can be of type $z^{-\om+\ups}$, $z^{\om+\ups}$ (see section~\ref{sec5} for a definition of $\tilde G(z)$ of this kind) or $z^{\om \ups}$ as in \Eq{propz} \cite{Calcagni:2025wnn}.}

The Green's function \Eq{propz} has a richer behaviour than its purely fractional counterpart \Eq{propzuv}. In fact, while the latter only has one branch point $z=0$ and a discontinuity on the whole imaginary axis $\mathbb{I}$, the former can have a certain number of poles with complex masses (infinitely many if $\g$ is irrational) distributed on a sequence of Riemann sheets. The rest of this section describes theses poles and sheets in a formal way and the trustful reader not interested in technicalities can skip in-text calculations and appendix~\ref{appB}.

To represent the Green's function \Eq{propz} directly on the $z$-plane of figure~\ref{fig2}, we first separate the $z=0$ singularity from the rest in order to avoid double dipping when combining the two halves. Writing
\be
\tilde G(z) =-\frac{1}{c_2z}+\left\{\frac{1}{c_2z}-\frac{1}{z\big[1+c_2 (z^2)^{\frac{\om}{2}}\big]^\ups}\right\}\eqqcolon -\frac{1}{c_2z}+\tilde\cG(z)\,,\label{propzt}
\ee
we define $\tilde\cG(z)$ in the $\Re\,z>0$ half-plane, take the branch $\sqrt{z^2}=+z$ and extend $\tilde\cG(z)$ to the $\Re\,z<0$ half-plane by picking the other branch $\sqrt{z^2}=-z$ therein: 
\bs\label{propz2split}\ba
\hspace{-1cm}\tilde G(z) &=& \Theta(\Re\,z)\,\tilde G_+(z)+\Theta(-\Re\,z)\,\tilde G_-(z)\nn
\hspace{-1cm}&=& -\frac{\Theta(\om)}{c_2z}+\Theta(\Re\,z)\,\cG_+(z)+\Theta(-\Re\,z)\,\cG_-(z)\,,\\
\hspace{-1cm}\tilde G_\pm(z) &=& -\frac{\Theta(\om)}{c_2z}+\cG_\pm(z)\,,\qquad
\cG_\pm(z)\coloneqq \pm\frac{\Theta(\om)}{c_2z}\mp\frac{1}{z\left[1+c_2(\pm z)^\om\right]^\ups}\,,\label{propz2split3}
\ea\es
where we used 
\be\label{heavi2}
\Theta(x)+\Theta(-x)=1\,,
\ee
and the fact that the $-1/z$ pole is absent for $\om<0$, as one can check. In terms of $\bar k^2$,
\ba
\tilde G(-\bar k^2) &=& \Theta[\Re(-\bar k^2)]\,\tilde G_+(-\bar k^2)+\Theta[\Re(\bar k^2)]\,\tilde G_-(-\bar k^2)\nn
&=& \frac{\Theta(\om)}{c_2\bar k^2}+\Theta[\Re(-\bar k^2)]\,\cG_+(-\bar k^2)+\Theta[\Re(\bar k^2)]\,\cG_-(-\bar k^2)\,.\label{TT1}
\ea
In the $w$-plane representation, which we use in section~\ref{brcu}, \Eqq{TT1} becomes
\be
\tilde G(-\bar k^2) =\frac{\Theta(\om)}{c_2\bar k^2}+\frac{1}{2\pi\rmi}\oint_{\G_w}\rmd w\,\frac{\cG(w)}{w+\bar k^4}\,,\label{propw1}
\ee
where
\be\label{propw2}
\boxd{\textrm{HP:}\qquad\cG(w) \coloneqq \frac{\Theta(\om)}{c_2(-w)^\frac12}-\frac{1}{(-w)^\frac12\big[1+ c_2(-w)^{\frac{\om}{2}}\big]^\ups}\,.}
\ee



\subsection{Branch point and discontinuity}

Contrary to the standard case (integer $\g$) where the point $z=0$ is a mass pole, for \Eq{propz} it is a branch point, as we prove in the next sub-section. 

Unless $\g$ is integer, this branch point and the associated branch cuts, which are the same as for the function $\tilde G(z)=(z^2)^{-\g/2}$ \Eq{propzuv2}, are always present in the theory. One branch cut runs on the positive imaginary semi-axis $\mathbb{I}^+$, while the other runs on the negative imaginary semi-axis $\mathbb{I}^-$. This creates a discontinuity between the $\Re\,z>0$ and the $\Re\,z<0$ half-planes which splits \Eq{propz} into two disconnected contributions shown in figure~\ref{fig2}. On each half, we assign one of the branches of $\sqrt{z^2}=\pm z$ \cite[Appendices A and B]{Calcagni:2025wnn}.


\subsection{Poles}\label{secpoles}

Recall that, for a complex number $z=x+\rmi y$ and $x\neq 0$, the principal value Arg of the argument is
\be\label{argxy}
{\rm Arg}(x + \rmi y) =
\left\{\begin{matrix}
	\arctan\left(\frac{y}{x}\right)\hphantom{+\pi}\qquad \text{if }\, x > 0\,,\hphantom{\, y \geq 0\,,} \\
	\arctan\left(\frac{y}{x}\right) + \pi\qquad \text{if }\, x < 0 \,,\, y \geq 0\,, \\
	\arctan\left(\frac{y}{x}\right) - \pi\qquad \text{if }\, x < 0 \,,\, y < 0\,.
\end{matrix}\right.
\ee
Writing $z=|z|\exp(\rmi\vp)$, the dispersion relation in the Green's function \Eq{propz2split} is
\ban
0&=& 1+c_2 (\pm z)^{\om}=1+c_2|z|^{\om}(\pm \rme^{\rmi\vp})^\om=1+{\rm sgn}(c_2)|c_2||z|^{\om}\rme^{\rmi\om(\vp+\de_\mp\pi)}\nn
&=& 1+|c_2||z|^{\om}\rme^{\rmi\om(\vp+\de_\mp\pi)+\rmi\frac{1-{\rm sgn}(c_2)}{2}\pi}\nn
&\eqqcolon& 1+|c_2||z|^{\om}\rme^{\rmi\psi}\,,
\ean
where we defined an angle $\psi$ and we chose to represent $\pm1=\exp(\rmi\pi\de_\mp)$ (this is a convention and any other choice would simply shift the position of the sheets along the Riemann surface). Then, to get a pole one requires $\psi=(2n+1)\pi$ and
\be
|z|=|c_2|^{-\frac{1}{\om}}\,,\qquad \vp=\vp_n^\pm\coloneqq\left[2n+1+\frac{{\rm sgn}(c_2)-1}{2}\right]\frac{\pi}{\om}-\de_\mp\pi\,,\qquad n\in\mathbb{Z}\,.
\ee
These singularities of the propagator, which we label as
\be\label{poles0}
z_n^\pm\coloneqq c_\om\,\rme^{\rmi\vp_n^\pm},\qquad c_\om\coloneqq |c_2|^{-\frac{1}{\om}}\,,\qquad n\in\mathbb{Z}\,,
\ee
are poles of order $\ups$, as one can see by parametrizing a small circle around each of them as $z=z_n^\pm+\ve\,\exp(\rmi\theta)$. Since $\pm z=[c_\om\cos(\vp_n^\pm+\de_\mp\pi)+\ve\cos(\theta+\de_\mp\pi)]+\rmi[c_\om\sin(\vp_n^\pm+\de_\mp\pi)+\ve\sin(\theta+\de_\mp\pi)]$, we have
\ben
f_n^\pm(\theta) \coloneqq [1+c_2(\pm z)^\om]^\ups= \left\{1+c_2[c_\om^2+2c_\om\ve\cos(\vp_n^\pm-\theta)+\ve^2]^{\frac{\om}{2}}\rme^{\rmi\om\Arg(\pm z)}\right\}^\ups,
\een
where
\ben
\Arg(\pm z)=\arctan\left[\frac{c_\om\sin(\vp_n^\pm+\de_\mp\pi)+\ve\sin(\theta+\de_\mp\pi)}{c_\om\cos(\vp_n^\pm+\de_\mp\pi)+\ve\cos(\theta+\de_\mp\pi)}\right]+\de_\theta\pi\,,\qquad \de_\theta\coloneqq 0,\pm 1\,.
\een
If $c_2=0=c_\om$, then $\Arg(\pm z)=\theta+(\de_\mp+\de_\theta)\pi$ and $f_n(0)\neq f_n(2\pi)$, so that circling around $z_n=0$ changes Riemann sheet and this is a branch point. If instead $c_\om\neq 0$, we have that $f_n(0)=f_n(2\pi)$ and \Eq{poles0} is a pole. Also, the residue is (appendix~\ref{appA})
\be
{\rm Res}\big[\cG_\pm(z),z_n^\pm\big] = \pm\frac{1}{\left[-c_2\om(\pm z_n^\pm)^{\om}\right]^{\ups}}=\pm\frac{1}{\om^\ups}\,,\label{resi2}
\ee
which is finite ($\om\neq 0$, $z_n^\pm\neq 0$) and non-vanishing, hence $z_n^\pm$ are poles of order $\ups$.

A change of sign of $c_2$ results in a rotation of the phases $\vp_n^\pm$, a different positioning of the poles in the complex plane and the appearance of real poles. In turn, this affects the choice of physical Riemann sheet, since some unphysical poles that must be avoided migrate from one Riemann sheet to another. All of this pertains to the mathematical presentation of the theory but leaves the physics untouched. Moreover, since the length $\lst$ is a free parameter hidden in $z$, the absolute value of $c_2$ does not matter. Therefore, from now on we fix
\be
c_2=+1=c_\om\,,
\ee
and
\be\label{c2pha}
\boxd{z_n^\pm=\rme^{\rmi\vp_n^\pm}\,,\qquad \vp_n^\pm=
\frac{2n+1}{\om}\pi-\de_\mp\pi\,,\qquad n\in\mathbb{Z}\,.}
\ee
Hence
\be
\vp_n^+= \frac{2n+1}{\om}\pi\,,\qquad\vp_n^-= \frac{2n+1}{\om}\pi-\pi\,.
\ee

Each contour $\G_\pm$ in figures~\ref{fig1} and \ref{fig2} runs in the Riemann sheet $\tilde S_l^\pm$ spanned by a coordinate $z=|z|\exp(\rmi\tilde\tau_l^\pm)$. Thus, each $\tilde G_\pm(-\bar k^2)$ is defined on
\ben
\tilde S_l^\pm\,:\qquad \tilde\t_l^\pm\in\,\Big](2l-\de_\pm)\pi,(2l+1+\de_\mp)\pi\Big[\,,\qquad l\in\mathbb{Z}\,.
\een
However, since $\tilde G(-\bar k^2)$ has a discontinuity $\Re\,z=0$ dividing the complex plane into two disjoint parts, the Green's function \Eq{propz} is actually defined on two sheets $S_{l,l'}\coloneqq S_l^+\cup S^-_{l'}$, each half-plane $S_l^\pm$ being the domain of $\tilde G_\pm(z)$ except the origin. Note that $l'$ may differ from $l$, since one can glue any of the two halves of the plane available on the Riemann surface. Each Riemann sheet $S_l^\pm$, $l\in\mathbb{Z}$, is a connected one-dimensional complex surface spanned by a complex coordinate with phase $\t_l^\pm$:
\be\label{Smpm}
S_l^\pm\,:\qquad \t_l^\pm\in\,\Big]\left(2l\mp\tfrac12\right)\pi,\left(2l+1\mp\tfrac12\right)\pi\Big[\,,\qquad l\in\mathbb{Z}\,,
\ee
where the $z=0$ point is removed. In particular,
\be
S_0^+\,:\quad \t_l^+\in\left]-\frac{\pi}{2},\frac{\pi}{2}\vphantom{\frac{3\pi}{2}}\right[\,,\qquad\qquad S_0^-\,:\quad \t_l^+\in\left]\frac{\pi}{2},\frac{3\pi}{2}\right[\,.
\ee
Each pair of Riemann sheets defines a different theory with a different spectrum. This might seem to introduce a brand new level of ambiguity not experienced in ordinary QFT but, as we will see, super-selection rules dictated by the absence of tachyons and of unpaired modes with complex masses greatly reduce the possibilities. Those that survive only differ in the purely virtual sector and, if one changes $c_2\to -1$, in the presence or not of an extra real pole in the physical spectrum.

If $\g=p/q$ is rational ($p\in\mathbb{Z}\ni q$) and irreducible ($p$ and $q$ have no common divisors except $\pm 1$), then 
\be\label{ppqp}
\om=\frac{p-q}{q\ups}\eqqcolon \frac{p'}{q'}\,,
\ee
and there are $p'$ poles of order $\ups$:
\be\label{pqpoles}
\textrm{$\g=\frac{p}{q}$\, rational:}\qquad {\rm card}\big\{z_n^+,z_n^-\,:\,n \in\mathbb{Z}\big\}=p'\,,
\ee
with phases
\be\label{phasepq}
\vp_n^\pm=\frac{(2n+1)q\ups}{p-q}\pi-\de_\mp\pi\,.
\ee
The finite number of poles \Eq{pqpoles} are distributed on a sequence of $q'$ contiguous Riemann sheets which we can conventionally make it start from $S_0$:
\be\label{gooseq}
\{S_0^\pm,S_1^\pm,\dots,S_{q'-1}^\pm\}\,.
\ee

If $\g\in\mathbb{R}\setminus\mathbb{Q}$ is irrational, then there are infinitely many non-repeating poles in the Riemann surface:
\be\label{irrpoles}
\hspace{-.7cm}\textrm{$\g$\, irrational:}\qquad {\rm card}\big\{z_n^\pm\,:\,n \in\mathbb{Z}\big\}=\aleph_0\,.
\ee
The finite sequence of Riemann sheets \Eq{gooseq} becomes infinite. In fact, since the set $\mathbb{Q}$ is dense in the reals ($\overline{\mathbb{Q}}=\mathbb{R}$, where the bar denotes closure), one can approximate $\g$ with an arbitrarily close rational number $p/q$, such that $p,q\to\infty$ but the ratio $p/q$ stays finite. This leads to an infinite sequence of sheets:
\be\label{gooseq2}
\{\dots,S_{-1}^\pm,S_0^\pm,S_1^\pm,\dots\}\,.
\ee

In general, for $\g>0$ the Riemann sheets can be empty or host isolated complex poles and pairs of complex-conjugate poles of the Green's function \Eq{propz}.
\begin{itemize}
\item \textbf{Isolated complex poles.} Sheets with unpaired complex poles without their conjugate are physically inadmissible. If a conjugate pair is split across two Riemann sheets, then each sheet has an unpaired complex pole leading to a mass spectrum with a non-zero imaginary part. 
\item \textbf{Complex-conjugate pairs.} If a conjugate pair is in the same sheet, then it does not affect unitarity by construction because its contribution is real-valued \cite{Lee:1969fy,Lee:1970iw,Lee:1969zze,Cutkosky:1969fq,Jansen:1993jj,Baulieu:2009ha,Modesto:2015ozb,Modesto:2016ofr,Mannheim:2018ljq,Mannheim:2020ryw} but Lorentz invariance is preserved only if one enforces the fakeon prescription \cite{Anselmi:2021hab,Liu:2022gun,Anselmi:2025uzj} (see section~\ref{cocopo} for more details). 
\end{itemize}
The fakeon prescription projects out the ghost modes carried by conjugate pairs but, as far as we know, it cannot fix the problems associated with 
 an isolated complex pole. Therefore, whenever a sheet contains 
  a complex unpaired pole it must be discarded, regardless of whether it also includes complex-conjugate pairs.

The existence and distribution of these poles on the Riemann surface is detailed in appendix~\ref{appB}. The most important result is Theorem~\ref{theo6-A}, stating that there always exists a Riemann sheet where the theory with $\g>0$ has no poles at all if
\be\label{omu2}
-\frac{1}{\ups}<\om<2\qquad\Longrightarrow\qquad 0<\g<2\ups+1\,.
\ee
This paves the way to a tremendous simplification of the theory. One always has the continuous spectrum of modes of the discontinuity, which is discussed below. But by defining the theory with an $\om$ in this interval, we remove all poles from $S_l^\pm$, healthy or unhealthy, and we do not have to worry about unwanted particle modes either spoiling the physical spectrum (coming from the $z=0$ point, which is outside $S_l^\pm$) or requiring the fakeon prescription. Since $\ups\geq 2$, the excellent news is that all of this happens at least for $0<\g<5$ and, practically, for an arbitrary range of $\g$ given that $\ups$ can be increased \emph{ad hoc}. This freedom is not needed, however, since \Eq{omu2} already encompasses the values in \Eq{sreno}--\Eq{reno} required for super-renormalizability. In particular, the minimal $\ups=2$ case is sufficient to remove real and complex poles from FQG while keeping one-loop super-renormalizability ($\g>4$ for gravity), according to \Eq{1lsu}.


\section{Tree-level spectrum and unitarity}\label{sec4}

In this section, we calculate the spectral representation of the Green's function \Eq{propz} on a generic Riemann sheet $S_{l}^+\cup S_{l'}^-$, where $l$ and $l'$ are not necessarily equal.

The contours on which one can integrate are those in figure~\ref{fig2} for, respectively, $\tilde G_-(z)$ and $\tilde G_+(z)$. Here $\G_\pm=\G_R^\pm\cup \G_\ve^\pm \cup L_\pm$, where $\G_R^\pm$ are the semi-circles at infinity, $\G_\ve^\pm$ are the semi-circles around the $z=0$ branch point and $L_\pm$ are the two pairs of straight integration lines parallel to $\mathbb{I}$. Of course, one can also calculate $\cG(w)$ on the contour $\G_w=\G_R\cup\G_\ve\cup\G_{\rm cut}^+\cup\G_{\rm cut}^-$ in the $w$-plane depicted in figure~\ref{fig3} and obtain the same result (in \cite{Briscese:2026jcf}, we take $w=+z^2$ and the branch cut is on the opposite semi-axis). This choice is actually preferred for the contribution of the discontinuity, since integration on $L_\pm$ in figure~\ref{fig2} does not lead to a clean spectral density $\rho(s)$ of the Green's function. The calculations are very similar to those in \cite{Calcagni:2022shb,Briscese:2024tvc,Calcagni:2025wnn}. There may also be infinitesimal circles $\G_n^\pm$ around poles (if present) inside these contours but they are not shown in the figures.


\subsection{Arcs at infinity}\label{sec42}

Parametrizing the coordinate along any of the $\G_R^\pm$ in figure~\ref{fig2} as $z=R\,\exp(\rmi\theta)$ with range $\theta_1\leq\theta<\theta_2$, we have
\ba
\frac{1}{2\pi\rmi}\int_{\G_R^\pm}\rmd z\,\frac{\tilde G_\pm(z)}{z+\bar k^2} &=& \frac{1}{2\pi}\int_{\theta_1}^{\theta_2}\rmd\theta\,\tilde G_\pm(R\,\rme^{\rmi\theta})+O\left(\frac{1}{R}\right)\nn
&=& \frac{1}{2\pi}\int_{\theta_1}^{\theta_2}\frac{\rmd\theta\,\rme^{-\rmi(\theta-\de_\pm\pi)}}{R[1+R^\om\rme^{\rmi\om(\theta+\de_\mp\pi)}]^\ups}+O\left(\frac{1}{R}\right)\nn
&\stackrel{R\to\infty}{\rightarrow}& 0\,,\label{gR2g2}
\ea
which is true for all $\om$ (hence for all $\g$) since $u>0$:
\be\label{gR2g2bis}
\forall\,\om\in\mathbb{R}\qquad\Longrightarrow\qquad \forall\,\g\in\mathbb{R}\,.
\ee
A calculation in the $w$-plane clearly gives the same result. Note the difference with respect to the purely fractional case $c_1=0$, for which $\g>0$ in order for this contribution to vanish.


\subsection{Complex poles}\label{cocopo}

Although we can always eliminate the poles, we calculate their contribution for completeness. On $S_l^\pm$,
\ba
\tilde G_n^\pm(-\bar k^2) &=& \lim_{\ve\to 0^+}\frac{1}{2\pi\rmi}\oint_{\G_n^\pm}\rmd z\,\frac{\tilde G_\pm(z)}{z+\bar k^2}\nn
&=& \lim_{\ve\to 0^+}\frac{\ve}{2\pi}\frac{1}{z_n^\pm+\bar k^2}\int_{2\pi}^0\rmd\theta \,\rme^{\rmi\theta}\tilde G_\pm(z_n^\pm+\ve\,\rme^{\rmi\theta})\nn
&=& -\frac{{\rm Res}[\tilde G_n^\pm(z),z_n^\pm]}{z_n^\pm+\bar k^2}\nn
&\stackrel{\textrm{\tiny \Eq{resi2}}}{=}& \frac{\mp\om^{-\ups}}{z_n^\pm+\bar k^2}\,,
\ea
where $\G_n^\pm$ is a circle of radius $\ve$ around $z_n^\pm=\exp(\rmi\bar\vp_n^\pm)$ parametrized as $z=z_n^\pm+\ve\exp(\rmi\theta)$. The penultimate step can be verified by an explicit calculation of the contour integral. The extra $-$ sign with respect to the residue is due to the fact that $\G_n^\pm$ is clock-wise.

For complex-conjugate pairs of poles $n=1,\dots, \bar n_\pm$, the contribution in each half-plane $S_l^\pm$ is
\be\label{Ipairs}
\tilde G_{\rm pairs}^\pm(-\bar k^2) =\sum_{n=1}^{\bar n_\pm}\left(\tilde G_n^\pm+\tilde G_n^{\pm*}\right)=\mp2\sum_{n=1}^{\bar n_\pm}\frac{\cos\vp_n^\pm+\bar k^2}{1+2\bar k^2\cos\vp_n^\pm+\bar k^4}\,,
\ee
which are then combined together as in \Eq{Gz2}.

Equation \Eq{Ipairs} is the Green's function of a complex-conjugate pair as a whole. By a field redefinition, one can see that the modes in a complex-conjugate pair always have opposite sign, which is the reason why they are usually described as complex ghosts \cite{Baulieu:2009ha,Asorey:2024mkb}. Since we cannot have such modes in the spectrum without violating unitarity, the issue is how to avoid them as asymptotic states. This is not granted \emph{a priori}. An isolated complex pole representing an on-shell particle would be physically observable as an imaginary component of the spectrum, which should not be the case because, by conservation of energy, incoming particles with real momenta cannot produce on-shell states with complex masses and \emph{vice versa}. This argument applies also to complex-conjugate pairs (these correspond to states with real total energy but only in some Lorentz frames, so that, in general, their energy may be complex-valued \cite{Jansen:1993jj}) but it does not prevent complex modes to enter as in-states.

Eventually, the degrees of freedom associated with complex-conjugate poles, called $\rmi$-particles in Yang--Mills theory and quantum chromodynamics \cite{Baulieu:2009ha}, do not appear in the physical spectrum as on-shell asymptotic in- or out-states and do not affect unitarity. Below, we recall how one reaches this conclusion. The off-shellness of $\rmi$-particles has been studied with a plethora of terminologies, arguments and techniques in past decades \cite{Lee:1969fy,Lee:1970iw,Cutkosky:1969fq,Jansen:1993jj,Modesto:2015ozb,Modesto:2016ofr,Anselmi:2017yux,Anselmi:2017lia,Liu:2022gun,Yamamoto:1970gw,Nakanishi:1972pt} but, because of this strong heterogeneity, it may be difficult to navigate through this sea of ideas and a compass has been provided only recently \cite{Anselmi:2022toe,Anselmi:2025uzj}. There are three ways in which $\rmi$-particles can be off-shell: as confined modes, as unstable modes, or as purely virtual modes.

\subsubsection{Confined modes}

One way to think of off-shellness is the interpretation of the pairing of complex-conjugate poles as confinement. For any given complex pair with phase $\vp_n^\pm=\vp$, calling $\rme^{\rmi\vp}=m^2+\rmi\, M^2$, where $m^2,M^2>0$, \Eq{Ipairs} can be written as (no $\pm\rmi\e$ prescription for the poles, since they are away from the real axis)
\be
\frac{1}{k^2+m^2+\rmi M^2}+\frac{1}{k^2+m^2-\rmi M^2}=\frac{2(k^2+m^2)}{(k^2+m^2)^2+M^4}\in\mathbb{R}\,.\label{ps2}
\ee
In gauge theories, the combined Green's function \Eq{ps2} is known as Gribov propagator and it describes a bound pair 
\cite{Baulieu:2009ha,Capri:2012hh,Dudal:2005na,Dudal:2007cw,Dudal:2008sp,Dudal:2011gd} (see \cite{Nakanishi:1975aq} for an early proposal). Circumstantial evidence of confinement is usually invoked by calculating the correlation function (bubble diagram) of the composite operator representing the pair and seeing that it has a pole with finite mass. Noting the striking similarity between the tree-level propagator for complex-conjugate pairs of ghosts in higher-derivative gravity and the non-perturbative propagator for gluons, in \cite{Liu:2022gun,Holdom:2015kbf,deBrito:2023pli} it was speculated that such pairs are confined states not appearing in the asymptotic spectrum. A similar conclusion was drawn in nonlocal systems in the presence of interactions with a gauge field \cite{Frasca:2022gdz}. Recently, in a toy model with a quartic potential it was demonstrated that these pairs are indeed bound states \cite{Asorey:2024mkb}. It is possible that the same could also happen in fractional theories, since the confining mechanism requires an interaction term for complex-conjugate modes, which is naturally present in gravity and can be introduced by hand in other cases. In ordinary QFT, a two-particle bound state is a pole at a mass between the single-particle state and the normal threshold $s=(2m)^2$, i.e., a state with a mass higher than $m$ but with not enough energy $2m$ to be torn apart. The same might hold in fractional QFT with conjugate pairs but, in the lack of a calculation along the lines of \cite{Capri:2012hh,Asorey:2024mkb}, here we remain agnostic about this point.

\subsubsection{Unstable modes}

Complex modes can also be regarded as off-shell by nature when interpreted as unstable particles. This topic is tightly related to how one defines scattering amplitudes and to several properties of the theory such as unitarity, the validity of the optical theorem, Lorentz invariance, hermiticity of the action and the Hamiltonian, the validity of power-counting renormalizability, the existence of a well-defined classical limit, and so on. In the presence of complex poles, there exist inequivalent definitions of Lorentzian scattering amplitudes, associated with inequivalent ways to project out complex modes from the asymptotic spectrum. Most procedures entail problems \cite{Anselmi:2022toe,Anselmi:2025uzj} such as a sign ambiguity in the amplitude (Cutkosky et al.'s (CLOP) prescription \cite{Cutkosky:1969fq,Anselmi:2017yux}), violation of the optical theorem (standard analytic continuation form Euclidean to Lorentzian amplitudes), of both power counting and the optical theorem (direct calculation on Minkowski spacetime \cite{Aglietti:2016pwz}) or of Lorentz invariance (Lee--Wick--Nakanishi prescription \cite{Lee:1969fy,Lee:1970iw,Nakanishi:1971jj}). In contrast, the fakeon prescription leads to a well-defined QFT with complex poles on Minkowski spacetime, respecting Lorentz invariance and the optical theorem \cite{Anselmi:2022toe,Anselmi:2025uzj}. (For the case of complex poles on curved backgrounds, see \cite{Koshelev:2020fok,Tokareva:2024sct}.)

The Lee--Wick--Nakanishi prescription is a particularly interesting illustration of how delicate is to keep $\rmi$-particles out of the spectrum while preserving vital properties of the theory. When quantized \emph{à la} Lee--Wick, complex modes are unstable and can thus be projected out via Veltman's procedure \cite{Veltman:1963th}. Unstable modes decay into other degrees of freedom, hence ideally they never appear asymptotically as in- or out-states in scattering amplitudes \cite{Veltman:1963th}.\footnote{In the literature of QFT with complex masses, these states have been characterized also as ``unphysical'' because they are not directly observable as single-particle asymptotic states in the physical Hilbert space, or even as a pair \cite{Jansen:1993jj,Modesto:2016ofr,Capri:2012hh}. We refrain from using the term ``unphysical'' because unstable particles can still leave an observable, indirect imprint in the way they affect scattering processes, or as resonances with real mass \cite{Jansen:1993jj}.} However, in realistic situations where the measurement of initial and final states is not performed in the infinite past and future, if the lifetime of such states is sufficiently large then they leave a physical imprint in the scattering process. This translates into a non-hermitian classical limit \cite{Anselmi:2022toe}. As an added problem, Lorentz symmetry is lost in this procedure \cite{Nakanishi:1971jj,Anselmi:2017lia,Anselmi:2025uzj}.

\subsubsection{Purely virtual modes}\label{pvmo}

Instead of relying on the instability or the confinement of modes with complex masses, another possibility with a very different physical meaning is to force them to be purely virtual. The fakeon prescription does precisely that. According to the optical theorem, the imaginary part of the amplitude can be written as a sum over intermediate states. If complex ghosts contributed to the imaginary part, then they would appear as intermediate states and would thus violate unitarity. The Lee--Wick--Nakanishi, CLOP and fakeon procedures guarantee that the net contribution of complex-conjugate pairs is manifestly real and does not enter the imaginary part of scattering amplitudes. Expression \Eq{ps2} (and its equivalent \Eq{Ipairs} in fractional theories) lacks the singularity points in the imaginary part that typically dictate the on-shell dispersion relation. Since the imaginary part of the contribution of complex-conjugate pairs to scattering amplitudes is zero in any of the above prescriptions, the sum of states in the optical theorem is only on intermediate states with real masses and unitarity on Minkowski spacetime is thus unaffected by the presence of the pairs. The fakeon prescription also guarantees Lorentz invariance at the same time. To understand intuitively how this procedure makes particles purely virtual, recall the Sokhotski--Plemelj formula for the standard Green's function with Feynman prescription:
\be
\frac{1}{x-\rmi\e}=\frac{1}{x-\rmi\e}\,\frac{x+\rmi\e}{x+\rmi\e}=\frac{x}{x^2+\e^2}+\rmi\,\frac{\e}{x^2+\e^2}\stackrel{\e\to 0^+}{=} {\rm PV}\!\left(\frac{1}{x}\right)+\rmi\pi\de(x)\,,\label{plemsok}
\ee
where PV is the principal value and $x=k^2+m^2$. The physical degrees of freedom $x=0$ are singled out by the imaginary part. At tree level, the fakeon prescription consists in taking only the PV part in \Eq{plemsok} or, in other words, the average continuation
\be\label{Gavg}
\tilde G_{\rm fake}(-k^2)=\frac{1}{2}\left[\frac{1}{k^2-\rmi\e+m^2}+\frac{1}{k^2+\rmi\e+m^2}\right]={\rm PV}\frac{1}{k^2+m^2}\,,
\ee
where we use $[\tilde G]=-2$ units (subscript $*$ omitted) and $\mp\rmi\e$ are, respectively the Feynman (causal) and Dyson (anti-causal) prescriptions.

Therefore, in fractional QFT complex pairs do not appear as asymptotic states either because we choose a physical Riemann sheet empty of them (which is always possible thanks to Theorem~\ref{theo6-A}) or because they are purely virtual (fakeon prescription). We pick the first avenue for complex poles but we will have to take the second one for the modes of the discontinuity, as we will see shortly.


\subsection{Branch point}\label{sec43}

In each of the contours in figure~\ref{fig2} on, respectively, the complex planes $S_l^\pm$, parametrizing $z=\ve\,\exp(\rmi\theta)$ we get the contribution
\ban
\cG_0^\pm(-\bar k^2) &=& \lim_{\ve\to 0^+}\frac{1}{2\pi\rmi}\int_{\G_\ve^\pm}\rmd z\,\frac{\cG_\pm(z)}{z+\bar k^2}\nn
&=& \lim_{\ve\to 0^+}\frac{\ve}{2\pi}\frac{1}{\bar k^2}\int_{\frac{\pi}{2}+\de_\mp\pi+2\pi l}^{\mp\frac{\pi}{2}+2\pi l}\rmd\theta \,\rme^{\rmi\theta}\cG_\pm(\ve\,\rme^{\rmi\theta})\nn
&\stackrel{\textrm{\tiny \Eq{propz}}}{=}& 
\lim_{\ve\to 0^+}\frac{\pm 1}{2\pi}\frac{1}{\bar k^2}\int_{\frac{\pi}{2}+\de_\mp\pi+2\pi l}^{\mp\frac{\pi}{2}+2\pi l}\rmd\theta\left\{\Theta(\om)-\frac{1}{\left[1+\ve^\om\rme^{\rmi\om(\theta+\de_\mp\pi)}\right]^\ups}\right\}.
\ean
The limit gives zero for any non-integer $\om$. From \Eq{TT1},
\be\label{Ive}
\tilde G_0(-\bar k^2) =\left\{\begin{matrix}
	0 &\qquad (\g<1)\,,\\
	\\
	\dfrac{1}{\bar k^2}&\qquad (\g>1)\,.
\end{matrix}\right.
\ee
This contribution 
is excellent news because it is the key piece to obtain a standard IR limit.

In the $w$-plane, there is no contribution from the branch point because it goes as $\sim\sqrt{\ve}$ for any $\om$. Indeed, \Eq{Ive} is fully contained in the separate term $\Theta(\om)/\bar k^2$ in \Eq{propw1}.


\subsection{Discontinuity}\label{brcu}

In this sub-section, we use the $w$-plane trick of figure~\ref{fig4} since the direct calculation on the $z$-plane is much more involved and suffers from technical obstructions such as the absence of a useful identity employed in the purely fractional case (unnumbered formula below eq.~(71) in \cite{{Calcagni:2025wnn}}). From \Eq{propw1}, it is easy to see that the contribution to the cut is
\ba
\tilde G_{\rm cut}(-\bar k^2) &=&\frac{1}{2\pi\rmi}\int_{\G_{\rm cut}^+\cup\G_{\rm cut}^-}\rmd w\,\frac{\cG(w)}{w+\bar k^4}\nn
&=&\lim_{\ve\to 0^+} \frac{1}{2\pi\rmi}\left[\int_0^{+\infty}\rmd t\,\frac{\cG(t+\rmi\ve)}{t+\rmi\ve+\bar k^4}+\int^0_{+\infty}\rmd t\,\frac{\cG(t-\rmi\ve)}{t-\rmi\ve+\bar k^4}\right]\nn
&=&\lim_{\ve\to 0^+} \frac{1}{2\pi\rmi}\int_0^{+\infty}\rmd t\,\frac{\cG(t+\rmi\ve)-\cG(t-\rmi\ve)}{t+\bar k^4}\,.\label{propwcut}
\ea
Noting that
\ben
\Arg(-t\mp\rmi\ve)=\arctan\frac{\mp\ve}{-t}\mp\pi\,\stackrel{\ve\to 0^+}{=}\, \mp\pi\,,
\een
from \Eq{propw2} we get
\ban
\cG(t\pm\rmi\ve)&=&\frac{\Theta(\om)}{(-t\mp\rmi\ve)^\frac12}-\frac{1}{(-t\mp\rmi\ve)^\frac12\big[1+ (-t\mp\rmi\ve)^{\frac{\om}{2}}\big]^\ups}\nn
&=&\mp\frac{\Theta(\om)}{\rmi\sqrt{t}}\pm\frac{1}{\rmi\sqrt{t}\left(1+t^{\frac{\om}{2}} \rme^{\mp\rmi\frac{\pi\om}{2}}\right)^\ups}\nn
&\stackrel{\ve\to 0^+}{=}&\mp\frac{\Theta(\om)}{\rmi\sqrt{t}}\pm\frac{\left(1+t^{\frac{\om}{2}} \rme^{\pm\rmi\frac{\pi\om}{2}}\right)^\ups}{\rmi\sqrt{t}\left(1+2t^{\frac{\om}{2}}\cos\frac{\pi\om}{2}+t^{\om}\right)^\ups}\,.\label{GG}
\ean
After renaming $t=s^2$ with $s>0$, we obtain
\ban
\tilde G_{\rm cut}(-\bar k^2) &=& \frac{1}{2\pi\rmi}\int_0^{+\infty}\rmd s\,2s\,\frac{\cG(s^2+\rmi\ve)-\cG(s^2-\rmi\ve)}{s^2+\bar k^4}\nn
&=& \frac{1}{\pi}\int_0^{+\infty}\rmd s\,\frac{1}{s^2+\bar k^4}\left[2\Theta(\om)-\frac{\left(1+s^\om\rme^{\rmi\frac{\pi\om}{2}}\right)^\ups
+\left(1+s^\om\rme^{-\rmi\frac{\pi\om}{2}}\right)^\ups}{\left(1+2s^\om\cos\frac{\pi\om}{2}+s^{2\om}\right)^\ups}\right],
\ean
and using the binomial series we reach the final expression 
\bs\label{profin}\ba
\tilde G_{\rm cut}(-\bar k^2) &=& \int_0^{+\infty}\rmd s\,\frac{\rho(s)}{s^2+\bar k^4}\,,\label{profinA}\\
\rho(s) &\coloneqq& \frac{2}{\pi}\left[\Theta(\om)-\frac{1}{\left(1+2s^\om\cos\frac{\pi\om}{2}+s^{2\om}\right)^\ups}\sum_{j=0}^\ups\binom{u}{j}s^{\om j}\cos\frac{\pi\om j}{2}\right].
\ea
\es
This is a superposition of complex-conjugate modes with a continuum of masses $s$. Since these pairs are purely virtual, unitarity holds at the tree level \emph{regardless of the sign of $\rho(s)$}. Note some interesting limits:
\begin{itemize}
\item For small $s$, the spectral density tends to zero: $\rho(s)\simeq -(2/\pi)s^{-\om\ups}\cos(\pi\om\ups/2)=-(2/\pi)s^{1-\g}\sin(\pi\g/2)$ for $\om<0$ and $\rho(s)\simeq (2/\pi)\ups s^\om\cos(\pi\om/2)$ for $\om>0$.
\item For large $s$, the spectral density tends to ${\rm sgn}(\om)(2/\pi)$: $\rho(s)\simeq (2/\pi)[\ups s^\om\cos(\pi\om/2)-1]$ for $\om<0$ and $\rho(s)\simeq-(2/\pi)[s^{1-\g}\sin(\pi\g/2)-1]$ for $\om>0$.
\item The purely fractional case is recovered by restoring $\Theta(\om)\to \Theta(\om)/c_2$ and $s^\om\to c_2s^\om$ in \Eq{profin} and then taking $c_2\gg (c_2s^\om)^\ups\gg 1$. This is indeed \Eq{KLfin} with $c_2=-1$ (compare \Eq{propzuv2} and \Eq{propz}).
\end{itemize}

Usually, the physical interpretation of a branch cut is that of a gas of particles with a continuous mass spectrum. In standard QFT, a typical branch cut starts at some multi-particle threshold $s=\left(\sum_im_i\right)^2$. This can be immediately understood as an interactive process occurring at high-enough energies to produce unbounded, intermediate, multi-particle states which are off-shell \cite{PeSc}. In our case, we have a discontinuity which is present already at the tree level in the absence of interactions, integration starts at $s=0$, i.e., at the mass $-k^2=m^2$ of a single particle, and it is difficult to assign the same physical meaning as for normal thresholds. Even if one could still interpret this as a continuum of off-shell multi-particle states not belonging to the physical spectrum, one must apply the fakeon prescription anyway because \Eq{profin} is the continuous superposition of complex pairs, which break Lorentz invariance in other prescriptions (sections~\ref{toyuni} and \ref{cocopo}).



\section{Uniqueness of fractional QFT}\label{sec5}

In general, \emph{any} fundamental perturbative QFT, even local, may be based upon a non-unique choice of kinetic term which is subsequently limited by requirements of self- and empirical consistency. It must: (i) reproduce observable phenomena, including symmetries; (ii) allow for a well-defined classical problem of initial conditions; (iii) preserve unitarity; (iv) be renormalizable. The same applies to nonlocal theories including FQG and, indeed, these requirements severely reduce the available choices. The reader can appreciate the strength of this combined super-selection by comparing the initial proposal \cite{Calcagni:2009kc} with the latest formulation with the self-adjoint extension of $\B^\g$ \cite{Calcagni:2025wnn}, passing through a number of alternative models \cite{Calcagni:2021ipd}.

However, contrary to local QFT, this super-selection is not enough. On one hand, for the HP definition we have seen that different Riemann sheets carry different poles and that one can always tune $\ups$ and $\g$ to get a specific pole distribution in any given sheet (this is nothing but the ambiguity in parametrizing the HP definition with generic $\ups\in\mathbb{N}^+$ and $\g>1$) or, conversely, a desired class of pole distributions (e.g., without complex poles) in more than one sheet. In ordinary QFT, at the tree level there is only one Riemann sheet, while at higher orders in the perturbative expansion one can single out what is called the ``physical sheet'' in the language of S-matrix theory, whereon the spectrum of the theory is well defined \cite{ELOP}. In fractional QFT, we can select a physical sheet once we require the absence of complex poles. However, we can get equally viable, different versions of the theory if we allow for complex-conjugate pairs.

On the other hand, the definition of the multi-fractional denominator in \Eq{propz} is not the only one and there are various ways to express the $\g$ power of $\B$ without altering the overall scaling $\g$ responsible for the renormalizability of the theory. This raises a question often asked in nonlocal QFT: If there is no unique way to select the form factor and one has many choices available, how can such theories be predictive?

The problem of having a different content of complex modes in different sheets is solved by the fakeon prescription, since the choice of sheet becomes a matter of convention. The problem of having different possible definitions of the kinetic term requires more and immediate attention. Below, we study other definitions which essentially cover all main possibilities. We show that the physical spectrum is the same as for \Eq{propz} and that the only differences emerge in the mass distribution of the fakeons. This only affects the purely virtual sector, which is not observable directly with the possible exception of some transient phenomena at the scale $\lst$ (see section~\ref{sec8}).



\subsection{Other versions of the fractional Green's function}\label{othrep}

We begin with the \emph{hermitian simple} (HS) definition
\be\label{hsim}
\boxd{\textrm{HS:}\qquad\tilde G(z) = \frac{1}{-z+c_2 (z^2)^\frac{\g}{2}} \,.}
\ee
This is called simple because it reproduces the intuitive toy model \Eq{eq:model} and its Green's function \Eq{propkh}. Equation~\Eq{hsim} is a special case of \Eq{propz} for $u=1$ and $c_2\to-c_2$.

Another operator has the \emph{hermitian asymmetric} (HA) definition
\be\label{hasy}
\boxd{\textrm{HA:}\qquad \tilde G(z) =\frac{1}{-z+c_2 (-z)^{\Upsilon}(z^2)^\frac{\g-\Upsilon}{2}}\,,\qquad \Upsilon=\lfloor\g\rfloor\,.}
\ee
It is called asymmetric because the exponent $\g$ has been redistributed between its integer floor $\Upsilon$ and the fractional or irrational part $\g-\Upsilon$. This arrangement is justified by the basic Balakrishnan--Komatsu representation of fractional powers of operators, which starts from a negative (here positive) exponent and then augments it with an integer power \cite{Calcagni:2025wnn}. Also, when $\Upsilon=1$, this is a special case of the HP definition \Eq{propz} with $\ups=1$, while when $\Upsilon$ is even one recovers the HS case.

Finally, we also consider the \emph{non-hermitian} (NH) operator that defined the theory in \cite{Calcagni:2021ljs,Calcagni:2021aap,Calcagni:2022shb}:
\be\label{nhrep}
\boxd{\textrm{NH:}\qquad\tilde G(z) =\frac{1}{-z+c_2 (-z)^\g}\,.}
\ee
The non-self-adjoint case is interesting because, perhaps surprisingly, it is as unitary as the others despite not having a well-defined classical limit.

\subsubsection{Arcs at infinity, branch point and discontinuity}

The result of section~\ref{sec42} can be adapted to the other models. The arcs at infinity give 0 for any $\g\in\mathbb{R}$. It is also not difficult to see that $z=0$ is a branch point for all these cases, for which we get \Eq{Ive}. 

HS and HA have a discontinuity on the whole imaginary axis $\mathbb{I}$ and they are represented by \Eqq{propz2split} with $\tilde G_\pm(z) =-\Theta(\g-1)/(c_2z)+\cG_\pm(z)$ given by
\ba
&&\textrm{HS:}\qquad \cG_\pm(z) =\pm\frac{\Theta(\g-1)}{c_2z}\pm\frac{1}{z\left[c_2(\pm z)^{\g-1}-1\right]}\,,\label{hsim2}\\
&&\textrm{HA:}\qquad \cG_\pm(z) =\pm\frac{\Theta(\g-1)}{c_2z}\pm\frac{1}{z\left[c_2(-1)^\Upsilon(\pm z)^{\g-1}-1\right]}\,.\label{hasy2}
\ea
Each of the $\tilde G_\pm(z)$ is defined on one of the half-planes $S_l^\pm$ given by \Eq{Smpm}.
In alternative, the Green's function can be written in the $w$-plane as \Eq{propw1} with
\ba
&&\textrm{HS:}\qquad \cG(w) =\frac{\Theta(\g-1)}{c_2(-w)^\frac12}+\frac{1}{(-w)^\frac12\left[c_2(-w)^{\frac{\g-1}{2}}-1\right]}\,,\label{hsim2w}\\
&&\textrm{HA:}\qquad\cG(w) =\frac{\Theta(\g-1)}{c_2(-w)^\frac12}+\frac{1}{(-w)^\frac12\left[c_2(-1)^\Upsilon(-w)^{\frac{\g-1}{2}}-1\right]}\,.\label{hasy2w}
\ea

In the NH case, there is only a branch cut on the positive real semi-axis. The Green's function $\tilde G(z)$ is defined on a connected one-dimensional complex surface $S_l$ spanned by a complex coordinate with phase $\t_l$:\footnote{If one chooses the convention $z\to -z$, then the branch cut is on the negative semi-axis and the Riemann sheets are parametrized by $\t_l\in\,]-\pi(1-2l),\pi(1+2l)]$. Then, the whole Riemann surface is symmetric with respect to the $\Im\,z=0$ axis.}
\be\label{Smnh}
\textrm{NH:}\qquad S_l\,:\qquad \t_l\in\,\big[2\pi l,2\pi(l+1)\big[\,,\qquad l\in\mathbb{Z}\,.
\ee
From now on, we set $c_2=1$ without loss of generality.

\subsubsection{Poles}\label{realpo}

Contrary to the poles $z_n=\exp(\rmi\vp_n)$ of the HP case, those of HS, HA and NH are simple. The calculation of the phases $\vp_n$ follows the same steps taken in section~\ref{secpoles} and yields
\ba
&&\textrm{HS:}\qquad \vp_n^\pm= \frac{2n+\de_\mp}{\g-1}\pi-\de_\mp\pi
\,,\label{c2pha1}\\
&&\textrm{HA:}\qquad \vp_n^\pm=
\frac{2n+1+(\Upsilon-1)(\de_\mp-1)}{\g-1}\pi-\de_\mp\pi\,,\label{c2pha2}\\
&&\textrm{NH:}\qquad \vp_n= \frac{2n+1}{\g-1}\pi+\pi\,.\label{c2pha3}
\ea
Note that the phase $\vp_n^-$ is identical in the HS and HA cases for any $\g$ and so is also the phase $\vp_n^+$ when $\Upsilon=2$ ($2<\g<3$). Therefore, the HS and HA analysis is the same in both half-planes when $2<\g<3$.

Just like in the HP case, one can always define the theory on sheets with no poles. There are infinitely many simple poles distributed on an infinite number of Riemann sheets as in \Eq{gooseq2} when $\g$ is irrational and $|p-q|$ simple poles on $q$ contiguous sheets when $\g=p/q$ is rational:
\be\label{pqpoles2}
\textrm{$\g=\frac{p}{q}$\, rational:}\qquad {\rm card}\big\{z_n^+,z_n^-\,:\,n \in\mathbb{Z}\big\}=|p-q|\,.
\ee
In the latter case, the main sequence \Eq{gooseq} is replaced by
\be
\textrm{HS, HA:}\qquad \{S_0^\pm,S_1^\pm,\dots,S_{q-1}^\pm\}\,,\label{gooseq'1}
\ee
where $S_l^\pm$ are given by \Eq{Smpm}, while for NH
\be
\textrm{NH:}\hphantom{, HA}\qquad \{S_0,S_1,\dots,S_{q-1}\}\,,\label{gooseq'2}
\ee
where $S_l$ is given by \Eq{Smnh}.

A notable difference with respect to the HP theory is that the HS, HA and NH dispersion relations also have real roots. This is because we choose $c_2=1$ in all cases but \Eq{propz} recovers the UV limit \Eq{propzuv2} with the opposite sign. This discrepancy was left on purpose to highlight the freedom one has to define fractional QFT, and the convergence to only one physics once the range of $\g$ is restricted.
\begin{itemize}
\item \textbf{Real pole with $\bm{M^2>0}$.} The heavy mode
\be
z=1\in S_l^+\,,\qquad\quad\! M^2=m^2+\mst^2\,,\label{pole+}
\ee
where $\mst\equiv\lst^{-1}$, always has a positive squared mass $M^2\gg m^2$ and can be considered as physical. If we choose a Riemann sheet with this pole, then the spectrum of the theory deviates at high energies with respect to the standard spectrum and, in principle, to new particle-physics phenomenology at scales $\sim M$. However, the possibility to test this is extremely remote given that $\mst$ is expected to be of order of the grand-unification scale or higher.
\item \textbf{Tachyon.} The pole
\be
z=-1\in S_l^-\,,\qquad M^2=m^2-\mst^2\,,\label{pole-}\\
\ee
is unphysical, since it corresponds to a particle with negative squared mass $M^2<0$. Indeed, for this mode to be standard (non-tachyonic) it must be $(m/\mst)^2>1$. A massive theory like this would have masses larger than the fundamental mass scale $\mst$. In a massless sector such as gravity, one would have $M^2=-\mst^2<0$. Since a complete QFT describing Nature would most likely have fundamental massless sectors (\Eqq{eq:model} with $m^2=0$), and since it is preferable to avoid tachyons and super-luminal modes, we regard \Eq{pole-} as unviable. 
\end{itemize}

The distribution on the Riemann sheets is detailed in appendices~\ref{appC1a}, \ref{appC2a} and \ref{appC3a} for real poles and in appendices~\ref{appC1b}, \ref{appC2b} and \ref{appC3b} for complex poles. However, we do not need to know these details thanks to Theorems~\ref{theo6-B1}, \ref{theo6-B2} and \ref{theo8-B3} (for, respectively, HS, HA and NH), stating that there are no poles for $0<\g<3$ once the Riemann sheet is chosen appropriately. Crucially, this interval covers the range $2<\g<3$ where fractional gauge theories are one-loop super-renormalizable and FQG is super-renormalizable with divergences also beyond one loop (section~\ref{toyren}). This range is also the one where HS and HA are identical. We conclude that the models HS$=$HA and NH can accommodate both renormalizability and unitarity, albeit with a restricted range of values $2<\g<3$ with respect to HP (the main model of this paper) and with a weaker renormalizability result in the case of FQG.

The contributions of the complex poles to the Cauchy representation of the Green's function can be calculated along the same lines as in section~\ref{cocopo}. The only other poles of interest, absent in the HP theory, are $z=\pm 1$. Although we can always remove these poles by an appropriate choice of the physical sheet, we briefly treat them for completeness. In the HS and HA models, the integration on the circle $\G_\ve$ parametrized by $z=\pm1+\ve\,\exp(\rmi\theta)$ gives
\ba
\tilde G_{*\ve} &=& \lim_{\ve\to 0^+}\frac{\lst^2}{2\pi\rmi}\int_{\G_\ve}\rmd z\,\frac{\tilde G_\pm(z)}{z+\bar k^2}\nn
&=& \lim_{\ve\to 0^+}\frac{\ve}{2\pi}\frac{\lst^2}{\bar k^2\pm 1}\int_{2\pi}^0\rmd\theta \,\rme^{\rmi\theta}\tilde G_\pm(\pm 1+\ve\,\rme^{\rmi\theta})\nn
&=& \frac{\lst^2}{\bar k^2\pm 1} = \frac{1}{k^2+m^2\pm\mst^2}\,.
\ea
In the NH theory, only $z=1$ is present.

\subsubsection{Continuum of modes}

The HP spectral representation \Eq{profin} can be recalculated for HS, HA and NH. Let us show the steps for HA, from which HS is obtained by setting $(-1)^\Upsilon=1$ everywhere. From \Eq{hasy2w} in the $w$-plane, we get
\ban
\cG(t\pm\rmi\ve)&=& \frac{\Theta(\g-1)}{(-t\mp\rmi\ve)^\frac12}+\frac{1}{(-t\mp\rmi\ve)^\frac12\left[(-1)^\Upsilon(-t\mp\rmi\ve)^\frac{\g-1}{2}-1\right]}\nn
&\stackrel{t=s^2}{=}&\mp\frac{\Theta(\g-1)}{\rmi s}\mp\frac{1}{\rmi s\left[(-1)^\Upsilon s^{\g-1}\rme^{\mp\rmi\frac{\pi(\g-1)}{2}}-1\right]}\nn
&=&\mp\frac{\Theta(\g-1)}{\rmi s}+\frac{\rmi (-1)^\Upsilon s^{\g-1}\rme^{\pm\rmi\frac{\pi\g}{2}}\pm1}{\rmi s\left[1-2(-1)^\Upsilon s^{\g-1} \sin\frac{\pi\g}{2}+s^{2(\g-1)}\right]}\,,\label{GGHA}
\ean
so that the spectral density in \Eq{profin} is replaced by
\ba
\rho_{\textsc{ha}}(s) &=& \frac{2}{\pi}\left[
\Theta(\g-1)+\frac{(-1)^\Upsilon s^{\g-1}\sin\frac{\pi\g}{2}-1}{1-2(-1)^\Upsilon s^{\g-1} \sin\frac{\pi\g}{2}+s^{2(\g-1)}}\right],\label{profinHA}\\
\rho_{\textsc{hs}}(s) &=&\frac{2}{\pi}\left[
\Theta(\g-1)+\frac{s^{\g-1}\sin\frac{\pi\g}{2}-1}{1-2 s^{\g-1} \sin\frac{\pi\g}{2}+s^{2(\g-1)}}\right].\label{profinHS}
\ea
We can reinstate $c_2$ factors with the replacements $\Theta(\g-1)\to \Theta(\g-1)/c_2$ and $s^{\g-1}\to c_2s^{\g-1}$. For $\g>1$, the spectral density tends to a constant (which vanishes if $c_2=1$) at $s=0$. The UV limit \Eq{KLfin} is recovered by taking the double limit $c_2\gg c_2s^{\g-1}\gg 1$. 

In the theory NH, instead of the discontinuity we have the contribution of the branch cut on the real positive semi-axis. Calling $\G_{\rm cut}^\pm$ the lines of the contour in figure~\ref{fig1} respectively above and below the cut, we have
\be
\tilde G_{\rm cut}(-\bar k^2) = \frac{1}{2\pi\rmi}\int_{\G_{\rm cut}^+\cup\G_{\rm cut}^-}\rmd z\,\frac{\tilde G(z)}{z+\bar k^2}=\lim_{\ve\to 0^+} \frac{1}{2\pi\rmi}\int_0^{+\infty}\rmd s\,\frac{\tilde G(s+\rmi\ve)-\tilde G(s-\rmi\ve)}{s+\bar k^2}\,,
\ee
where we have parametrized $z=s\pm\rmi\ve$ on $\G_{\rm cut}^\pm$. Since ($c_2=1$ from now on)
\ben
\tilde G(s\pm\rmi\ve) =\frac{1}{-s\mp\rmi\ve+(-s\mp\rmi\ve)^\g}\nn
\stackrel{\ve\to 0^+}{=} \frac{1}{-s+ s^\g\rme^{\mp\rmi\pi\g}},\label{GG-B3}
\een
we get
\be\label{Icut}
\tilde G_{\rm cut}(-\bar k^2) = \int_0^{+\infty}\rmd s\,\frac{\rho_{\textsc{nh}}(s)}{s+\bar k^2}\,,\qquad \rho_{\textsc{nh}}(s)\coloneqq \frac{1}{\pi}\frac{s^{\g-2}\sin(\pi\g)}{1-2 s^{\g-1}\cos(\pi\g)+s^{2(\g-1)}}\,.
\ee
For small $s$ and $\g<1$ or for large $s$ and $\g>1$ (UV limit), we recover the purely fractional case \Eq{spera1} with $c_2=1$. For small $s$ and $\g>1$ or for large $s$ and $\g<1$, $\rho_{\textsc{nh}}(s)\sim s^{\g-2}$ and the spectral density vanishes only in the super-renormalizable range $\g>2$.

While the same physical interpretation applies to HP, HS and HA, the one for NH is different because, once the Feynman prescription is restored, the imaginary part of \Eq{Icut} is non-zero, it contributes to scattering amplitudes with in- and out-states and can affect unitarity. Then, we have to make sure that the spectral density be positive semi-definite. The denominator is always positive, while the numerator is positive for non-integer $\g$ if, and only if, $\sin(\pi\g)>0$, i.e., when 
\be\label{unitytree}
\textrm{unitarity condition:}\qquad 2k<\g<2k+1\,,\qquad k\in\mathbb{N}\,.
\ee
Thus, in the range $0<\g<3$ where the Riemann sheets have no poles, the cut gives a unitary contribution at the tree level when
\be
0<\g<1\,,\qquad 2<\g<3\,.
\ee


\subsection{Comparison with nonlocal QFT with entire form factors}\label{compa}

FQG is not the only proposal for quantum gravity with nonlocal operators. Nonlocal QFTs where form factors are entire functions have been studied much longer and more extensively. In particular, the gravitational sector has received much attention and is known with the name of ``nonlocal quantum gravity'' (NLQG), more specifically, with asymptotically polynomial form factors. Reviews on several classical and quantum aspects of this class of theories are \cite{Modesto:2017sdr,Buoninfante:2022ild,BasiBeneito:2022wux,Koshelev:2023elc}. Omitting the matter sector, the action \Eq{gravac} is replaced by
\be\label{nlqg}
S_\textsc{nlqg}=\frac{\Mpl^2}{2}\int\rmd^4x\,\sqrt{|g|}\,\left[R-c_2 G_{\mu\nu}\frac{\rme^{\H_2(\B)}-1}{\B}\,R^{\mu\nu}-c_0R\frac{\rme^{\H_0(\B)}-1}{\B}\,R\right]+\dots,
\ee
where $c_2\neq 0$ and $\H_{0,2}(z)$ are entire functions. Popular choices that work for a scalar field are $\H_{0,2}(z)=-z$ (Wataghin form factor) or $\H_{0,2}(z)=z^2$ (Krasnikov form factor). For gauge theories and gravity, one must take something different, in particular, the complementary exponential integral $\H_{0,2}(z)={\rm Ein}(z)$; this is the form factor of NLQG.

Since the action of NLQG and FQG is formally the same up to the choice of form factors, the path integral can be treated with exactly the same tools \cite{Calcagni:2024xku}. However, compared with FQG, entire form factors make the study of perturbative unitarity at all orders more accessible than in the present case \cite{Pius:2016jsl,Briscese:2018oyx,Chin:2018puw}. NLQG also has a simpler treatment of the perturbative spectrum because the form factors $\exp\H_{0,2}(z)$ do not introduce any extra pole. In exchange, they have an essential singularity at infinity, which makes arcs at infinity in their Cauchy representation and in the associated propagator vanish only in a certain conical domain $\cC$ in the complex plane. In contrast, the fractional form factors are defined on the whole complex plane (see \Eqqs{gR2g2}--\Eq{gR2g2bis}),
\be\label{coni}
\cC=\mathbb{C}\quad \forall \g\in\mathbb{R}\,,
\ee
provided one gives a prescription on the imaginary axis; the essential singularity at infinity is traded for a continuum of complex-conjugate modes plus, if allowed, complex-conjugate poles at finite distance from the origin. Thus, neither in NLQG nor in FQG do the form factors introduce extra fundamental degrees of freedom\footnote{Of course, the quadratic curvature terms carry anyway more degrees of freedom than in Einstein gravity, none of them being ghosts \cite{Calcagni:2018gke}. We do not discussed them here.} but for different reasons. In NLQG, because they are entire; in FQG, because the complex-conjugate extra modes associated with the discontinuity of the Green's function are purely virtual and they do not contribute to the imaginary part of scattering amplitudes. Complex-conjugate pairs could manifest themselves as intermediate bound states (see section~\ref{sec8}), a feature not present in NLQG.

The extension of the convergence region \Eq{coni} of fractional form factors to the whole complex plane has an important technical consequence for the basic formulation of the theory. A typical problem of nonlocal QFT is that Wick rotation is not possible because $\cC$ does not cover the whole plane. This poses a problem at the time of establishing the spacetime signature where to start when defining the theory. In the case of NLQG, the path integral is defined on metrics in Lorentzian signature but Feynman diagrams and scattering amplitudes must be calculated in Euclidean signature and external momenta are later promoted to Lorentzian momenta via Efimov analytic continuation \cite{Pius:2016jsl,Briscese:2018oyx,Chin:2018puw,Efimov:1967dpd,Koshelev:2021orf,Buoninfante:2022krn}. On the other hand, thanks to \Eq{coni}, in FQG one can pass to and from the Euclidean formulation from and to the Lorentzian one and integrations in Feynman integrals on the real line are equivalent to integrations on the imaginary axis $\mathbb{I}$. This is a signal that the Osterwalder--Schrader conditions \cite{Osterwalder:1973dx,Osterwalder:1974tc} hold in fractional QFT. In particular, reflection positivity allows one to carry out unitarity calculations in alternative ways \cite{Trinchero:2017weu,Calcagni:2021ljs}. Therefore, Efimov's continuation is mandatory in NLQG but optional in FQG. 
For what concerns the prescription on the imaginary axis, one can impose that on the semi-axis $\Im\,z>0$, $\Re\,z=0$, the form factor is obtained by analytic continuation from the region $\Re\,z>0$, while on the semi-axis $\Im\,z<0$, $\Re\,z=0$ it is obtained by analytic continuation from the region $\Re\,z<0$. If one chooses this prescription, one can prove that Wick rotation is allowed, since one can consider a path consisting of the two semi-axes $\Im\,z=0$, $\Re\,z\geq0$ and $\Im\,z\geq0$, $\Re\,z=0$, plus the arc at infinity, such that the form factor is analytic in the whole first quadrant of the imaginary plane (same for the third quadrant).

Finally, NLQG is manifestly Lorentz invariant at all orders in perturbation theory \cite{Briscese:2018oyx} while, as discussed in section~\ref{brcu}, fractional QFT is such only if  one employs the fakeon prescription \cite{Anselmi:2025uzj} for the complex-conjugate modes (even if one considers the pole-free version of the theory).


\section{Perturbative unitarity}\label{sec6}

To be consistent, the theory should be unitarity at all orders in the perturbative expansion. This check was done at one-loop level for a purely fractional non-hermitian model in \cite{Calcagni:2021ljs}. From now on, $\g>1$ and $c_2=1$. 

We have been able to write the tree-level propagator in the HP, HS and HA theories in the same form:
\be
\rmi\tilde G_*(-\bar k^2) =\frac{\rmi\lst^2}{\bar k^2}+\rmi\lst^2\int_0^{+\infty}\rmd s\,\dfrac{\rho(s)}{s^2+\bar k^4}\,.\label{propkfin}
\ee
In this way, we can always isolate diagrams with the continuum of complex-conjugate pairs in their internal lines, which either violate the optical theorem if the associated amplitude is analytically continued as in standard textbooks (Euclidean external and internal momenta, analytic continuation $\pe^2\to p^2-\rmi\e$ after integration) or break Lorentz invariance if the amplitude is prescribed \emph{\`a la} Lee--Wick--Nakanishi (Lorentzian external and internal momenta, internal $k^0$ integrated over a certain complex path, real spatial $\bm{k}$) \cite{Anselmi:2025uzj}. These diagrams can be rendered harmless with the fakeon prescription (average analytic continuation to $\pe^2\to p^2\pm\rmi\e$), in which case these pairs never appear as asymptotic states in the optical theorem and Lorentz invariance is preserved \cite{Anselmi:2025uzj}, regardless of the sign of $\rho(s)$. The other contributions are all standard and they respect unitarity at any loop order.

We can see this with 
a couple of one-loop examples based on the HS model \Eq{eq:model}, first a scalar $\phi^3$ theory ($N=3$, $\a=\g$) and then a model with derivative vertices with $N=4$ simulating a gauge theory ($\a=1$) or gravity ($\a=0$). We start from a Euclidean amplitude with both internal and external momenta in Euclidean signature (respectively, $k^\mu=(k_4,\bm{k})$ and $p^\mu=(p_4,\bm{p})$). We set $\lst=1$ 
and consider $1<\g<3$ in order to avoid the contribution of complex pairs with fixed masses. The HP and HA cases are very similar. The NH case (model \Eq{eq:modelnh}) is even simpler and is studied below.


\subsection{Self-adjoint scalar models}

For $\la_{\g-\a}=0$ (no derivatives in vertices) and $N=3$, the Euclidean one-loop bubble in four dimensions is the same for all the self-adjoint theories:
\be\label{self1}
\feynmandiagram[layered layout, baseline=(i.base), horizontal=i to b]{
i -- [scalar, edge label=\(p\)] b -- [scalar, half left, edge label=\(k\)] 
c -- [scalar, edge label=\(p\prime\)]] f,
c -- [scalar, half left, edge label=\(k+p\)] b,
};
 = \rmi\de^D(p-p')\,\cM(p^2)\,,
\ee
where $\cM$ is the amplitude. The latter is first calculated in Euclidean signature:
\ba
\cM_\textsc{e}(\pe^2)&=&\frac{\la_0^2}{2}\int_{\mathbb{R}^4}\frac{\rmd^4\ke}{(2\pi)^4}\,\tilde G(-\bke^2)\,\tilde G[-(\bke+\bpe)^2]\nn
&=& \frac{\la_0^2}{2}\int_{\mathbb{R}^4}\frac{\rmd^4\ke}{(2\pi)^4}\,\frac{1}{\bke^2}\frac{1}{(\bke+\bpe)^2}\nn
&&+\frac{\la_0^2}{2}\int_0^{+\infty}\rmd s\,\rho(s)\int_0^{+\infty}\rmd s'\,\rho(s')\int_{\cD}\frac{\rmd^4\ke}{(2\pi)^4}\,\frac{1}{s^2+\bke^4}\frac{1}{s^{\prime 2}+(\bke+\bpe)^4}\nn
&&+\frac{\la_0^2}{2}\int_0^{+\infty}\rmd s\,\rho(s)\int_{\cD}\frac{\rmd^4\ke}{(2\pi)^4}\left[\frac{1}{\bke^2}\frac{1}{s^{2}+(\bke+\bpe)^4}+\frac{1}{s^2+\bke^4}\frac{1}{(\bke+\bpe)^2}\right]\nn
&\eqqcolon&\cM_{1\phi}^\textsc{e}(\pe^2)+\frac{\la_0^2}{2}\int_0^{+\infty}\rmd s\,\rho(s)\,\widetilde\cM_{\rm pairs}^\textsc{e}(s;\pe^2)\nn
&\eqqcolon&\cM_{1\phi}^\textsc{e}(\pe^2)+\cM_{\rm pairs}^\textsc{e}(\pe^2)\,,\label{M1Lfin}
\ea
where 
\be
\rho=\rho_{\textsc{hp}},\,\rho_{\textsc{hs}},\,\rho_{\textsc{ha}}\,,
\ee
we used \Eq{propkfin}, exchanged momentum and mass integrations, kept the integration domain $\cD$ unspecified and separated the standard Euclidean amplitude $\cM_{1\phi}^\textsc{e}(\pe^2)$ (coming from the first factor in \Eq{propkfin}) from the rest.

The $\ke$ integral in the third line of \Eq{M1Lfin} is the bubble amplitude of two complex-conjugate pairs with squared masses $\pm\rmi s$ and $\pm\rmi s'$. In turn, one can decompose this into four bubble diagrams with the particles $\cP$, $\cP'$ and their conjugate $\cP^*$ and $\cP^{\prime *}$ in the internal lines, through the combinations
$\cP$-$\cP'$, $\cP$-$\cP^{\prime *}$, $\cP^*$-$\cP'$ and $\cP^*$-$\cP^{\prime *}$, plus their symmetric counter-parts under $s\leftrightarrow s'$. These diagrams have been calculated for fixed masses in, e.g., \cite{Anselmi:2025uzj}.

Now we want to define the Lorentzian amplitude from the Euclidean one with the constraint that the optical theorem and Lorentz invariance be preserved. This requires to impose two different prescriptions on the two terms in \Eq{M1Lfin}. Since $\cM_{1\phi}^\textsc{e}(\pe^2)$ is the standard bubble diagram, we can analytically continue it as usual:
\be\label{Msta}
\cM_{1\phi}(p^2)\coloneqq\cM_{1\phi}^\textsc{e}(p^2-\rmi\e)\,,
\ee 
where we gave the Feynman prescription. However, we cannot do the same for $\cM^\textsc{e}_{\rm pairs}$ because the analytic continuation would violate the optical theorem, which is the reason why Lee and Wick originally introduced their prescription in the presence of complex poles \cite{Lee:1969fy,Lee:1970iw,Lee:1969zze}. This prescription breaks Lorentz symmetry for fixed masses $s$ and $s'$ \cite{Nakanishi:1971jj,Anselmi:2025uzj} and integration on such masses cannot, in general, restore Lorentz invariance miraculously.

For this reason, we have to resort to the fakeon prescription and define the Lorentzian amplitude as the average analytic continuation
\be\label{avgcont}
\cM_{\rm pairs}(p^2)\coloneqq \frac12[\cM^\textsc{e}_{\rm pairs}(p^2-\rmi\e)+\cM^\textsc{e}_{\rm pairs}(p^2+\rmi\e)]\,,
\ee
which is equivalent to a choice of domain $\cD$ in the complex hyperplane $(k^0,\bm{k})$ allowing for complex values of spatial momenta. At this point, since $\cM_{\rm pairs}(p^2)$ is real, it gives no contribution to the sum over asymptotic states in the optical theorem.
Overall, the Lorentzian amplitude is
\be
\cM(p^2)=\cM_{1\phi}^\textsc{e}(p^2-\rmi\e)+\cM_{\rm pairs}(p^2)\,,
\ee
which is made of a standard term plus a piece not contributing to the physical spectrum. It is an exercise to adapt the Cutkosky rules to this theory but it is already apparent that unitarity is preserved because the imaginary part of the amplitude does not contain the contribution of the virtual pairs. (The fakeon prescription is still needed to maintain Lorentz invariance.) 

The generalization to $L$ loops and a generic potential $\phi^N$ is straightforward and one can conclude that the fractional scalar field theory with a non-derivative monomial potential is perturbatively unitary. Generalization to the HP and HA models is immediate. 


\subsection{NH scalar model}

Curiously, NH is an example of a theory that maintains unitarity despite being non-hermitian, without the need of the fakeon prescription. Indeed,
\ba
\cM_\textsc{e}(\pe^2)&=&\frac{\la_0^2}{2}\int_{\mathbb{R}^4}\frac{\rmd^4\ke}{(2\pi)^4}\,\tilde G(-\bke^2)\,\tilde G[-(\bke+\bpe)^2]\nn
&=& \frac{\la_0^2}{2}\int_0^{+\infty}\rmd s\,\mu(s)\int_0^{+\infty}\rmd s'\,\mu(s')\int_{\cD}\frac{\rmd^4\ke}{(2\pi)^4}\,\frac{1}{s+\bke^2}\frac{1}{s'+(\bke+\bpe)^2}\nn
&\eqqcolon&\frac{\la_0^2}{2}\int_0^{+\infty}\rmd s\,\mu(s)\int_0^{+\infty}\rmd s'\,\mu(s')\int_{\mathbb{R}^4}\frac{\rmd^4\ke}{(2\pi)^4}\,\frac{1}{s+\bke^2}\frac{1}{s'+(\bke+\bpe)^2}\nn
&\eqqcolon&\frac{\la_0^2}{2}\int_0^{+\infty}\rmd s\,\mu(s)\int_0^{+\infty}\rmd s'\,\mu(s')\,\widetilde\cM_{1\phi}(s,s';\pe^2)\,,\label{M1LfinNH}
\ea
where
\be
\mu(s)\coloneqq \rho_{\textsc{nh}}(s)+\de(s)\,,
\ee
and
\ba
\widetilde\cM_{1\phi}(s,s';\pe^2)&=&\int_{\mathbb{R}^4}\frac{\rmd^4\ke}{(2\pi)^4}\,\frac{1}{s+\bke^2}\frac{1}{s'+(\bke+\bpe)^2}\nn
&=&\int_{\mathbb{R}^4} \frac{\rmd^4\ke}{(2\pi)^4}\int_0^1\,\frac{\rmd x}{\{(1-x)(\bke^2+s)+x[(\bke+\bpe)^2+s']\}^2}\nn
&=&\int_{\mathbb{R}^4} \frac{\rmd^4\ke}{(2\pi)^4}\int_0^1\,\frac{\rmd x}{[(\bke+x \bpe)^2+x(1-x)\bpe^2+(1-x)s+xs']^2}\nn
&\stackrel{\text{\tiny $k'=\bke+x\bpe$}}{=}&\int_{\mathbb{R}^4} \frac{\rmd^4 k'}{(2\pi)^4}\int_0^1\,\frac{\rmd x}{[k^{'2}+x(1-x)\bpe^2+(1-x)s+xs']^2}\nn
&=&\frac{1}{(2\pi)^4}\int_0^1\,\rmd x\int_{\mathbb{R}^4} \frac{\rmd^4 k'}{[k^{'2}+x(1-x)\bpe^2+(1-x)s+xs']^2}\,.\label{appate}
\ea
Given that $\bke,\bpe\in\mathbb{R}$ and $s,s'>0$, the conditions for the second equality are
\ben
(\bpe+\bke)^2-(\bke^2+s-s')>0\quad {\rm or}\quad -1<\frac{(\bpe+\bke)^2-(\bke^2+s-s')}{\bke^2+s}<0\,,
\een
but the second is equivalent to $(\bpe+\bke)^2-(\bke^2+s-s')<0$ and one can also analytically continue to the special case $(\bpe+\bke)^2-(\bke^2+s-s')=0$. This type of integral, which has no poles since $k'$ and $\bpe$ are real, is well known in dimensional regularization:
\ba
I_{r,w}& \coloneqq & \int \rmd^{D} k'\, \frac{(k^{'2})^r}{(k^{'2} +{\rm C})^w}= \frac{2\pi^{\frac{D}{2}}}{\G\left(\frac{D}{2}\right)} \int_0^{+\infty}\rmd k'\,  \frac{k^{'2r+D-1}}{(k^{'2} +{\rm C})^w}\nn
&=& \pi^{\frac{D}{2}} \frac{\Gamma\left(r+\frac{D}{2}\right) \Gamma\left(w-r-\frac{D}{2}\right)}{\Gamma\left(\frac{D}{2}\right)\G(s)}\,{\rm C}^{\frac{D}{2}+r-w}\nn
&\stackrel[\text{\tiny $w=2$}]{\text{\tiny $r=0$}}{=}& \pi^{\frac{D}{2}}\Gamma\left(2-\frac{D}{2}\right)\,{\rm C}^{\frac{D}{2}-2}\,\,\stackrel{\text{\tiny $D=4-2\ve$}}{=}\,\,\pi^2\Gamma\left(\ve\right)\,{\rm C}^{-\ve}\nn
&=&\pi^2\left(\frac{1}{\ve}-\g_\textsc{em}-\ln {\rm C}\right)+O(\ve)\,,
\ea
where $\g_\textsc{em}$ is the Euler--Mascheroni constant. Plugging ${\rm C}=x(1-x)\pe^2+(1-x)s+xs'$, \Eqq{appate} becomes
\be
\widetilde\cM_{1\phi}(s,s';\pe^2) = \lim_{\ve\to 0} \int_0^1\,\frac{\rmd x}{(4\pi)^2}\left\{\frac{1}{\ve}-\g_\textsc{em}-\ln\left[x(1-x)\bpe^2+(1-x)s+xs'\right]\right\}+O(\ve)
\,,\label{Mpp}
\ee
which is inserted into \Eq{M1LfinNH}. This is a superposition of ordinary modes that can be analytically continued as usual. The standard Cutkosky rule applies but now it acts as a convolution over the spectral densities. Unitarity of the Lorentzian amplitude \Eq{Msta} then follows suit, with a notable physical difference: the decay rate (or cross-section) of the incoming state is modified and, instead of a sharp turn-on at a specific mass threshold, the phase-space volume is smeared by the density function $\rho_{\textsc{nh}}(s)$.

Note that there are no new divergences arising due to this smearing. On one hand, the weight $\int_0^{+\infty}\rmd s\,\rho_{\textsc{nh}}(s)$ is finite and the $1/\ve$ pole keeps a finite residue. On the other hand, the remaining finite part contains an integral $\int_0^{+\infty}\rmd\ln s\,\rho_{\textsc{nh}}(s)$, which is also finite because $\rho_{\textsc{nh}}(s)$ drops off faster that $1/(s\ln s)$ as $s\to+\infty$.

The generalization to arbitrary loop orders follows the same scheme.


\subsection{HS gauge and gravity models}

When $\a\neq\g$ in \Eqq{eq:model}, the model mimics the action for a gauge theory ($m=0$, $N=4$, $\a=1$) or the graviton ($m=0$, $N=4$, $\a=0$) in fractional QFT. In this case, vertices carry extra factors $\sim(k^4)^{(\g-\a)/2}=|k^2|^{\g-\a}$, similarly to the non-hermitian case below. To see how they affect the above conclusions on unitarity, we can take the simple example of the $N=4$, $m=0$ one-loop diagrams \cite{Calcagni:2022shb}, generalizing the $\a=1$ results therein. The general four-point vertex in Lorentzian signature is
\ba
\tilde{\rm V}(k,p,q,r) &=&-4\rmi\la_0-4\rmi\la_{\g-\a}\left[|(p+k)^2|^{\g-\a}+|(p+r)^2|^{\g-\a}+|(p+q)^2|^{\g-\a}\right.\nn
&&\left.+|(k+r)^2|^{\g-\a}+|(k+q)^2|^{\g-\a}+|(r+q)^2|^{\g-\a}\right].\label{vertgen}
\ea
For the self-energy diagram
\be\label{self2}
\begin{tikzpicture}[baseline=(a1.base)]
  \begin{feynman}
    \vertex (a1) {};
    \vertex[right=1.9cm of a1] (a2);
    \vertex[right=1.9cm of a2] (a3) {};
    \vertex[above=1.5cm of a2] (b);

    \diagram* {
      (a1) --[scalar, edge label'=\(p\)] (a2) --[scalar, half left] (b) --[scalar, half left] (a2) --[scalar, edge label'=\(p\prime\)] (a3),
    };
  \end{feynman}
\end{tikzpicture}
 = \rmi\de^4(p-p')\,\cM(p^2)\,,
\ee
we can set $r=-k$, $q=-p'=-p$, so that
\be
\tilde{\rm V}(k,p,-p,-k) =-4\rmi\la_0-8\rmi\la_{\g-\a}\left[|(p+k)^2|^{\g-\a}+|(p-k)^2|^{\g-\a}\right].
\ee
Then, the Euclidean amplitude
\be
\cM_\textsc{e}(\pe^2)=\int_{\cD}\frac{\rmd^4\ke}{(2\pi)^4}\,\tilde{\rm V}(\ke,\pe,-\pe,-\ke)\,\rmi\tilde G(-\ke^2)\propto \la_{\g-\a} k_{\textsc{uv}}^{2(2-\a)}
\ee
is a power divergence, where $k_{\textsc{uv}}$ is a UV cut-off and we have dropped $O(\la_0)$ terms. Just like in the standard $\phi^4$ local theory, this contribution is removed by dimensional regularization and does not say much about unitarity. Note that the counter-term to be added is local, since $\a=0,1$ for gauge and gravity theories.

The one-loop four-point amplitude is more interesting:
\be
\parbox{4.1cm}{\includegraphics[width=4cm]{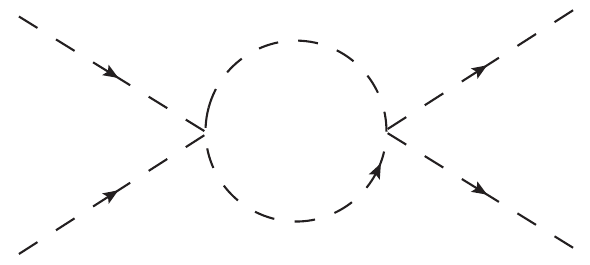}} + \parbox{4.1cm}{\includegraphics[height=3.2cm]{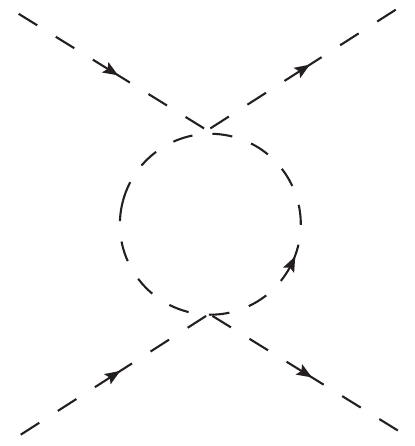}} + \parbox{4.1cm}{\includegraphics[height=3.2cm]{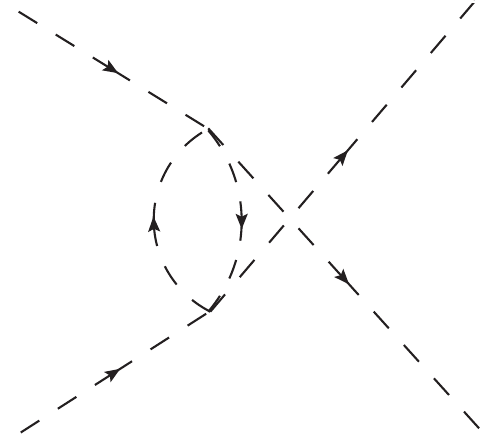}}.
\ee
Setting $p=0=q$ and $r=-k$, we have
\be\label{bonbon}
\tilde{\rm V}(k,0,0,-k) =-4\rmi\la_0-16\rmi\la_{\g-\a}|k^2|^{\g-\a}\,,
\ee
and the Euclidean amplitude for $\pe=0$ reads
\ba
\cM_\textsc{e}&=&\int_{\cD}\frac{\rmd^4\ke}{(2\pi)^4}\,[\tilde{\rm V}(\ke,0,0,-\ke)\,\rmi\tilde G(-\ke^2)]^2\nn
&=& \int_{\mathbb{R}^4} \frac{\rmd^4\ke}{(2\pi)^4}\frac{16[\la_0+4\la_{\g-\a}\ke^{2(\g-\a)}]^2}{\ke^2\ke^2}\nn
&&+\int_0^{+\infty}\rmd s\,\rho_{\textsc{hs}}(s)\int_0^{+\infty}\rmd s'\,\rho_{\textsc{hs}}(s')\int_{\cD} \frac{\rmd^4\ke}{(2\pi)^4}\frac{16[\la_0+4\la_{\g-\a}\ke^{2(\g-\a)}]^2}{(s^2+\ke^4)(s^{\prime 2}+\ke^4)}\nn
&&+\int_0^{+\infty}\rmd s\,\rho_{\textsc{hs}}(s)
\int_{\cD} \frac{\rmd^4\ke}{(2\pi)^4}\frac{32[\la_0+4\la_{\g-\a}\ke^{2(\g-\a)}]^2}{\ke^2(s^{2}+\ke^4)}\nn
&\eqqcolon& \cM^\textsc{e}_{1g}+\int_0^{+\infty}\rmd s\,\rho_{\textsc{hs}}(s)\,\widetilde\cM^\textsc{e}_{\rm pairs}(s)\,.
\ea
The contribution $\cM^\textsc{e}_{\rm pairs}$ of the complex-conjugate pairs can be fakeonized as before.


\subsection{NH gauge and gravity models}

The non-hermitian case for gauge and gravity fields is like the above one but without absolute values:
\ba
\tilde{\rm V}(k,p,q,r) &=&-4\rmi\la_0-4\rmi\la_{\g-\a}\left\{[(p+k)^2]^{\g-\a}+[(p+r)^2]^{\g-\a}+[(p+q)^2]^{\g-\a}\right.\nn
&&\left.+[(k+r)^2]^{\g-\a}+[(k+q)^2]^{\g-\a}+[(r+q)^2]^{\g-\a}\right\}.\label{vertgennh}
\ea

The one-loop four-point Euclidean amplitude with $\pe=0=q$ and $r=-\ke$
is
\ba
\cM_\textsc{e}&=&\int_{\cD}\frac{\rmd^4\ke}{(2\pi)^4}\,[\tilde{\rm V}(\ke,0,0,-\ke)\,\rmi\tilde G(-\ke^2)]^2\nn
&\eqqcolon&\int_0^{+\infty}\rmd s\,\mu(s)\int_0^{+\infty}\rmd s'\,\mu(s')\,\widetilde\cM_{1g}(s,s')\,,
\ea
where
\ba
\widetilde\cM_{1g}(s,s')&=&16\int_{\mathbb{R}^4} \frac{\rmd^4\ke}{(2\pi)^4}\frac{[\la_0+4\la_{\g-\a}\ke^{2(\g-\a)}]^2}{(\ke^2+s)(\ke^2+s')}\nn
&=& \lim_{\ve\to 0} \frac{\la_0^2}{\pi}\frac{1}{\sin\left(2\pi\frac{D}{4}\right)}\frac{{s'}^{\frac{D}{2}-1}-s^{\frac{D}{2}-1}}{s-s'}\Bigg|_{D=4-2\ve}\nn
&&+\frac{2\la_0\la_{\g-\a}}{\pi}\frac{1}{\sin\left[\pi\left(\g-\a\right)\right]}\frac{{s'}^{1+\g-\a}-s^{1+\g-\a}}{s-s'}\nn
&&++\frac{\la_{\g-\a}^2}{\pi}\frac{1}{\sin\left[2\pi\left(\g-\a\right)\right]}\frac{{s'}^{1+2(\g-\a)}-s^{1+2(\g-\a)}}{s-s'}\nn
&=& \lim_{\ve\to 0} \frac{\la_0^2}{\pi^2\ve}-\frac{\la_0^2}{\pi^2}\frac{s\ln s-s'\ln s'}{s-s'}\nn
&&+\frac{2\la_0\la_{\g-\a}}{\pi}\frac{1}{\sin\left[\pi\left(\g-\a\right)\right]}\frac{{s'}^{1+\g-\a}-s^{1+\g-\a}}{s-s'}\nn
&&++\frac{\la_{\g-\a}^2}{\pi}\frac{1}{\sin(2\pi\g)}\frac{{s'}^{1+2(\g-\a)}-s^{1+2(\g-\a)}}{s-s'}\,.\label{self4}
\ea
The first $O(\la_0^2)$ term is calculated in dimensional regularization and corresponds to the usual logarithmic divergence (local counter-term) plus a finite part. The other three contributions in \Eq{self4} must be integrated in $s$ and $s'$; we do not perform the rest of the calculation here.

To prove one-loop unitarity, we must verify that the Lorentzian amplitude
\be\label{Msta2}
\cM(p^2)\coloneqq\cM_\textsc{e}(p^2-\rmi\e)
\ee 
satisfies the optical theorem, specifically that its imaginary part is non-negative and corresponds to the physical phase space of the intermediate states. For fixed internal masses $s$ and $s'$, the imaginary part of the bubble amplitude $\widetilde{\mathcal{M}}_{1g}$ is found by replacing the propagators with on-shell delta functions, which is the ordinary Cutkosky rule:
\be
\frac{1}{k^2 + s} \to -2\pi\rmi\,\de(k^2 - s)\,\Theta(k^0)\,.
\ee
For the nonlocal vertex $\tilde{\rm V}$, restoring momentum dependence the Cutkosky rule yields
\be
\Im \left[ \widetilde{\mathcal{M}}_{1g}(s, s'; p^2) \right] = \frac{1}{16\pi} \frac{\sqrt{\lambda(s,s',p^2)}}{p^2} \left|\tilde{\rm V}(s,s') \right|^2 \Theta\left[p^2 - (\sqrt{s} + \sqrt{s'})^2\right]\,,
\ee
where $\lambda(a,b,c)=a^2+b^2+c^2-2ab-2ac-2bc$ is the Källén function. Integrating in $s$ and $s'$, since the spectral density is real we get
\be
\Im\,\cM(p^2) = \frac{1}{16\pi p^2} \int_0^{p^2} \rmd s \int_0^{(\sqrt{p^2}-\sqrt{s})^2} \rmd s' \, \mu(s) \mu(s') \left| \tilde{\rm V}(s,s') \right|^2 \sqrt{\lambda(p^2, s, s')}\,.
\ee
To satisfy unitarity, we require $\Im\,\cM(p^2) \geq 0$ for all physical $p^2$. The phase-space factor is $\sqrt{\lambda(s,s',p^2)} \geq 0$ within the integration limits defined by the $\Theta$-function. Since $\mu(s)=\rho(s)+\de(s)>0$ for the ranges \Eq{unitytree}, the integrand is a product of non-negative functions and the integral is non-negative.

This proves that the one-loop amplitude preserves unitarity. The imaginary part exactly sums the probabilities of producing all possible combinations of mass-continuum states allowed by the energy $p^2$. Generalization to higher orders is immediate as long as \Eq{unitytree} is enforced.


\section{Causality in fractional QFT}\label{sec7}

Before concluding, we examine the topic of causality in fractional QFT. First, we recall the difference between micro- and macro-causality \cite{Yamamoto:1969vb,Yamamoto:1970gw}. \emph{Micro-causality} is the property that any physical quantity or event at any given point $P$ must be determined by and determines only the quantities on an arbitrary space-like surface inside, respectively, its past and future light-cone. By definition, such property holds for arbitrarily small time and spatial intervals. \emph{Macro-causality} is such that any physical quantity or event at $P$ does not affect the outside of the future light-cone. This implies that the effects of $P$ are not necessarily manifest in all space-like surfaces in the future light-cone. In other words, macro-causality is limited to large times and distances as determined by the experiment.

It is well known that nonlocal operators, and in particular non-integer powers of the d'Alembertian \cite{BG,Belenchia:2014fda},  preserve macro-causality but smear the light cone at distances of order of the fundamental scale $\lst$ of the theory, thus violating micro-causality. At these distances, however, one is in a full quantum regime and such violations in the classical theory bear no harmful physical consequence. In the case of fractional QFT, we have decomposed the tree-level propagator into a genuine nonlocal part (the continuum of modes in the spectral density $\rho(s)$) and a local part made of the standard propagator and of complex-conjugate poles and one may wonder whether such decomposition is consistent with the above result on causality. This is indeed the case even when we keep the complex-conjugate poles, since a propagator with complex-conjugate masses respects macro-causality \cite{Yamamoto:1969vb,Yamamoto:1970gw}.

However, violations of causality may occur at arbitrarily large scales in the massless sector of fractional QFT due to the presence of fakeons. To understand why, we first review the case of fakeons with fixed mass and then move to the fractional case. In this section, we use $[\tilde G]=-2$ units omitting the subscript $*$.


\subsection{Fakeons with fixed mass}\label{sec7a}

Recall that causality violations happen in fakeon models at energies larger than the fakeon mass $m$ and events separated by spacetime intervals of order $\lesssim 1/m$ cannot be chronologically ordered \cite{Anselmi:2018tmf,Anselmi:2018bra,Anselmi:2019nie}. The origin of the breaking of causality lies in the advanced part in the fakeon prescription of the Green's function, which is the average of the retarded and advanced prescriptions (respectively, Feynman causal and Dyson anti-causal prescription). For a canonical particle with mass $m$, this is \Eqq{Gavg}. For $m$ large enough, one can rest assured that only micro-causality is broken and that these models do not exhibt experimentally problematic acausal phenomena. Lower bounds on the mass of fakeons from modifications of Newton's law in a two-derivative model (see \Eqq{gree3} below) are given in \cite{Anselmi:2019nie} and are $O(10^{-3}\,{\rm eV})$.

The same averaging procedure occurs in amplitudes with complex-conjugate pairs \cite{Anselmi:2025uzj,Anselmi:2025uda}. At the level of the propagator in momentum space, this is equivalent to a direct replacement of the momentum $\ke^2$ in the Euclidean Green's function of the whole pair with its Lorentzian version $k^2$ with no $\pm\rmi\e$ prescription \cite{Anselmi:2025uda}:
\be
\tilde G_{\rm fake}(-k^2)= \tilde G(-\ke^2)\Big|_{\ke^2\to k^2}\,.
\ee
Depending on whether one takes the sum or the difference of the two Green's functions of the pair, one has
\be\label{Gavg2}
\tilde G_{\rm fake}^{(\pm)}(-k^2)=\frac{1}{k^2+m^2-\rmi M^2}\pm\frac{1}{k^2+m^2+\rmi M^2}=\frac{A_\pm(k^2)}{(k^2+m^2)^2+M^4}\,,
\ee
where $A_+(k^2)=2(k^2+m^2)$ for the sum (Gribov function \Eq{ps2}) and $A_-(k^2)=2\rmi M^2$ for the difference.

Here there are two poles at $-k^2=m_\pm^2\coloneqq m^2\pm\rmi M^2$. What physically matters to have causality violations only at short scales is the presence of a mass gap away from $k^2=0$, so that $m$ and $M$ must not vanish simultaneously. When $m\neq 0$ and $M=0$, there is a real double pole at $-k^2=m^2$ which poses no causality issue. When $m=0$ and $M\neq 0$, there are two purely imaginary poles at $-k^2=\pm\rmi M^2$. This is sufficient to avoid violations of macro-causality because the position-space Green's function decays at scales $\sim \sqrt{2}/M$. To see all of this, we first find the modified Newton's potential, from which we can read the characteristic damping scale. Consider the $D$-dimensional Euclidean Green's function in position space of a massive particle:
\ban
G_{m,D}(\xe^2)&=&\int\frac{\rmd^D\ke}{(2\pi)^D}\,\frac{\rme^{\rmi \ke\cdot \xe}}{\ke^2+m^2}=\int\frac{\rmd\Om_D}{(2\pi)^D}\int_0^{+\infty}\rmd\ke\,\ke^{D-1}\,\frac{\rme^{\rmi \ke\cdot \xe}}{\ke^2+m^2}\nn
&=& \int_0^{+\infty}\rmd\ke\,\frac{\ke^{D-1}}{\ke^2+m^2}\int\frac{\rmd\Om_D}{(2\pi)^D}\,\rme^{\rmi \ke r\cos\theta}\,,
\ean
where $r=|\xe|$. Integrating over angular variables,
\ban
\int\frac{\rmd\Om_D}{(2\pi)^D}\,\rme^{\rmi \ke\cdot \xe} &=& \frac{2\pi^{\frac{D-1}{2}}}{(2\pi)^D\G\left(\frac{D-1}{2}\right)}\int_0^\pi\rmd\theta\,(\sin\theta)^{D-2}\rme^{\rmi \ke r\cos\theta}\nn
&=&\frac{1}{2^{D-1}\pi^{\frac{D+1}{2}}\G\left(\frac{D-1}{2}\right)}\,\pi^{\frac12}\G\left(\frac{D-1}{2}\right)\left(\frac{2}{\ke r}\right)^{\frac{D}{2}-1}J_{\frac{D}{2}-1}(\ke r)\nn
&=&\frac{1}{(2\pi)^{\frac{D}{2}}}\,\frac{1}{(\ke r)^{\frac{D}{2}-1}}\,J_{\frac{D}{2}-1}(\ke r)\,,
\ean
where $J$ is the Bessel function of the first kind, so that (we keep the argument in the left-hand side linear for simplicity)
\be\label{greeD}
G_{m,D}(r)=\frac{1}{(2\pi)^{\frac{D}{2}}}\,\frac{1}{r^{\frac{D}{2}-1}}\int_0^{+\infty}\rmd\ke\,\frac{\ke^{\frac{D}{2}}}{\ke^2+m^2}\,J_{\frac{D}{2}-1}(kr)\,.
\ee
In $D=3$ dimensions, $J_{1/2}(z)=\sqrt{2/(\pi z)}\sin z$ and \Eqq{greeD} becomes the Yukawa potential
\be\label{gree3}
G_{m,3}(r) = \frac{\rme^{-m r}}{4\pi r}\,,
\ee
which has the characteristic damping scale $1/m$ mentioned above. For $m=0$, one recovers the standard Newton's potential $\Phi_0(r)=1/(4\pi r)$. For the pair of complex-conjugate modes in \Eq{Gavg2}, we have
\ba
G_{m,M,3}^{(\pm)}(r)&\coloneqq &G_{m_-,3}(r)\pm G_{m_+,3}(r)=\frac{\rme^{-m_- r}\pm \rme^{-m_+ r}}{4\pi r}\nn
&=&\frac{\rme^{-\mu_+ r}}{2\pi r}\,\rme^{\rmi\de_\mp\frac{\pi}{2}}\cos\left(\mu_-r-\de_\mp\frac{\pi}{2}\right)\,,\label{Gavg3}
\ea
where
\be
m_\pm =\sqrt{m^2\pm\rmi M^2}=\mu_+\pm\rmi\mu_-\,,\qquad \mu_\pm\coloneqq \sqrt{\frac{\sqrt{m^4+M^4}\pm m^2}{2}}\,.
\ee
$G_{m,M,3}^{(-)}(r)$ was previously considered also in \cite{Anselmi:2025uda}. The modulation factor in \Eq{Gavg3} is $\cos(\mu_-r)$ for the $+$ sign and $\rmi\sin(\mu_-r)$ for the $-$ sign. The exponential factor damps the Green's function at large scales $\gg 1/\mu_+$. When $m\neq 0$ and $M=0$, $\mu_+=m$ and $\mu_-=0$, while when $m=0$ and $M\neq 0$ one has $\mu_+=\mu_-=M/\sqrt{2}$, as anticipated above. 

To see what happens to causality in the massless limit $m_\pm\to 0$, we take $D=4$ in \Eqq{greeD}, integrate and combine the complex-conjugate pair. Back to Lorentzian signature, the net result is \cite{Anselmi:2025uda,Anselmi:2019nie}
\ba
G_{m,M,4}^{(\pm)}(x^2)&=&G_{m_-,4}(x^2+\rmi\e)\pm G_{m_+,4}(x^2-\rmi\e)\nn
&=&\frac{1}{(2\pi)^2}\left[\frac{m_- K_1\left(m_-\sqrt{x^2+\rmi\e}\right)}{\sqrt{x^2+\rmi\e}}\pm\frac{m_+ K_1\left(m_+\sqrt{x^2-\rmi\e}\right)}{\sqrt{x^2-\rmi\e}}\right],\label{Gavg4}
\ea
where $K$ is the modified Bessel function of the second kind and we analytically continued to Lorentzian intervals $r\to \sqrt{x^2\pm\rmi\e}$ in, respectively, the $m_\mp$ part. Since $K_1(z)\simeq 1/z$ in the limit $z\to 0$, from the Sokhotski--Plemelj formula \Eq{plemsok} we have
\be\label{massless}
\lim_{m_\pm\to 0} G_{m,M,4}^{(+)}(x^2)=\frac{1}{2\pi^2}\,{\rm PV}\frac{1}{x^2}\,,\qquad \lim_{m_\pm\to 0} G_{m,M,4}^{(-)}(x^2)=-\frac{\rmi}{2\pi}\,\de(x^2)\,.
\ee
Note the absence of any time ordering (no $\Theta(\pm t)$ factors).  Therefore, in both cases the massless Green's function becomes arbitrary large near or, respectively, on the light cone $x^2=0$ and causality violations happen unsuppressed at all scales. 


\subsection{Fakeons with varying mass}\label{sec7b}

In hermitian fractional QFT, \Eqq{Gavg2} is replaced by the continuum of modes \Eq{profinA} with $M^2$ replaced by $M^2_s\coloneqq s\mst^2$ with $s\in[0,+\infty)$:
\be
\tilde G_{\rm cut}(-k^2)=\int_0^{+\infty}\rmd s\,\frac{\rho(s)}{2\rmi s}\,\tilde G_s(-k^2)\,,\qquad \tilde G_s(-k^2)\coloneqq\frac{2\rmi M_s^2}{(k^2+m^2)^2+M_s^4}\,.\label{Gcuts}
\ee
Comparing with \Eqq{Gavg2}, we see that \Eq{Gcuts} is of ``difference'' type, so that in $D=3$ dimensions the correction $\de\Phi(r)\coloneqq \Phi(r)-\Phi_0(r)$ to the standard Newton's potential $\Phi_0(r)$ is the integrated version of \Eqq{Gavg3} with the $-$ sign:
\bs\ba
\hspace{-1.5cm}\de\Phi(r)=G_{m,3}^{\rm cut}(r)&=&\int_0^{+\infty}\rmd s\,\frac{\rho(s)}{2\rmi s}\,G_{m,M_s,3}^{(-)}(r)=\int_0^{+\infty}\rmd s\,\frac{\rho(s)}{2s}\,\frac{\rme^{-\mu_+^s r}}{2\pi r}\,\sin(\mu_-^s r)\,,\\ 
\mu_\pm^s&\coloneqq& \sqrt{\frac{\sqrt{m^4+M_s^4}\pm m^2}{2}}=\sqrt{\frac{\sqrt{m^4+s\mst^4}\pm m^2}{2}}\,.\label{Gavg5}
\ea\es
Thus, there is a continuum of damping scales $1/\mu_+^s$ running from zero ($s\to+\infty$, $\mu_+^s\to+\infty$, total damping) to
\be
m\neq 0\,:\quad \frac{1}{\mu_+^s} = \frac{1}{m}\,,\qquad\qquad
m  =  0\,:\quad \frac{1}{\mu_+^s} = \frac{\sqrt{2}}{M_s}\,.
\ee
Therefore, for $m\neq0$ there always is a non-zero damping scale, while for $m=0$ there is no damping ($\mu_+^s\to 0$) when $s\to 0$.

In $D=4$ dimensions, from \Eqq{Gavg4} we have
\be
G_{m,4}^{\rm cut}(r)=\int_0^{+\infty}\rmd s\,\frac{\rho(s)}{2\rmi s}\,G_{m,M_s,4}^{(-)}(x^2)\,.
\ee
When $m=0$, for small $s$ the integrand tends to the delta in \Eqq{massless}, $G_{m,M_s,4}^{(-)}(x^2)\sim \de(x^2)$. In the HP model, the integration weight $\rho(s)/s\sim s^{\om-1}$ tends to zero for $\om>1$ but this does not prevent the damping scale $1/\mu_+^s$ to become very large at small but non-zero $s$. 

The conclusion is that fractional QFT has a continuum of fakeon modes with masses $M_s$ arbitrarily close to zero, which leads to a violation of macro-causality at distances $\gg \lst$ if $m=0$. 

This problem would fatally hinder the application of the program of fractional QFT in any massless sector, including gravity. A way out of it is to introduce a mass gap in the fractional operators definining the action. This new mass parameter could be any reasonably large energy scale respecting solar-system bounds but, for parsimony, we can simply take $\mst$. In the case of FQG, the action \Eq{gravac} is replaced by
\be\label{gravacM}
S=\frac{\Mpl^2}{2}\int\rmd^4x\,\sqrt{|g|}\,\left\{R-c_2\lst^2 G_{\mu\nu}[(\lst^2\B-1)^2]^{\frac{\g}{2}-1}\,R^{\mu\nu}-c_0\lst^2 R[(\lst^2\B-1)^2]^{\frac{\g}{2}-1}\,R\right\}\,,
\ee
plus matter and killer contributions. Neither the renormalizability nor the unitarity properties of the theory are affected by this reformulation and the parameter space is not enlarged. The only modification of interest is in the structure of the real poles \Eq{pole+} and \Eq{pole-} discussed in section~\ref{realpo}. The $z=1$ pole gets a mass $M^2=2\mst^2=O(m^2)$, while the $z=-1$ becomes massless, $M^2=0$. These poles can always be removed from the physical spectrum anyway, via a suitable choice of Riemann sheet. Note also that the graviton remains massless, since the introduction of a mass gap in \Eq{gravacM} is akin to a Stelle-like modification of gravity with the usual decomposition into a massless spin-2, massive ghost spin-2 and massive spin-0 fields, rather than to a massive-gravity alteration of the theory. The massive spin-2 and spin-0 fields are fakeonized as above and the leading term in the total Newton's potential is the standard one $\Phi_0(r)$. The bounds $m_g\leq 2.42\!-\!3.62\times 10^{-23}\,{\rm eV}$ on the mass of the graviton \cite{Bernus:2020szc,LIGOScientific:2021sio} are therefore respected.


\section{Conclusions}\label{sec8}

In this paper, we have studied the physical spectrum of asymptotic states of QFTs with fractional kinetic terms. Choosing the Riemann sheet whereupon to define the theory, we have removed a massive real extra pole, a tachyonic mode, unpaired modes with complex masses and pairs of complex-conjugate modes. Using the fakeon presciption, we have virtualized a continuum of complex-conjugate modes carrying ghosts while preserving Lorentz invariance. We mainly worked with the hermitian polynomial definition \Eq{propz} (HP) but we have duplicated all the results also for the hermitian simple definition \Eq{hsim} (HS), the hermitian asymmetric definition \Eq{hasy} (HA) and the non-hermitian definition \Eq{nhrep} (NH). The only difference among all these definitions is in the distribution of the poles at any given Riemann sheet and in the range of $\g$ for which all poles are removed; this range is $2<\g<3$ for NH and HS$=$HA (the two collapse into the same theory for these $\g$'s) and is unrestricted only in the HP theory.

In all the theories considered here, for $\g<1$ the vacuum is the only asymptotic state; this generalizes and is in agreement with the $\g=1/2$ case treated in \cite{Barci:1996ny} with the Hamiltonian formalism. For $\g>1$, the spectrum is constituted by a particle that can go on-shell with dispersion relation $k^2+m^2=0$. The continuum-mass part of the propagator corresponds to physical degrees of freedom that never go on-shell. At all perturbative orders, neither the continuum of particles nor the conjugate pairs (if present) disturb unitarity, provided that one prescribes the amplitudes \emph{à la} Anselmi--Piva for the self-adjoint theories or, for the non-hermitian one, $\g$ takes values specified by the unitarity condition \Eq{unitytree}.

In section~\ref{sec5}, we have mentioned that differences in the spectral density $\rho(s)$ do not imply that these theories are physically different because the fakeon prescription confines all the complex-conjugate modes into a virtual sector never feeding the pool of asymptotic states in scattering amplitudes. However, there is the possibility that fakeons be observable indirectly as resonance-like double bumps in the real part of the dressed propagator \cite{Anselmi:2023wjx}. If this were the case, then it would be possible, at least in principle, to discriminate empirically between different choices of self-adjoint kinetic term. Moreover, scenarios with the extra real pole \Eq{pole+} could also be considered as alternative to the present ones where such pole is either absent (HP theory) or removed by a suitable choice of Riemann sheet (HS, HA and NH theories). Having an extra heavy mode in the spectrum can have interesting phenomenological consequences, although not in present particle accelerators where one cannot reach the near-Planckian or near-grand-unification energies $\mst=\lst^{-1}$. Therefore, strictly speaking, the issue of uniqueness of fractional QFTs is mitigated but not completely closed and it will deserve further attention in the future.

These results put fractional QFTs in general, and FQG in particular, on a much firmer footing. They drastically simplify the theoretical construction of the propagator and remove the necessity of the bulky regularization procedure of \cite{Calcagni:2022shb}, which is an artifact of purely fractional QFTs with kinetic term $\sim\B^\g$. They also redefine the use of the fakeon prescription in fractional QFT, adopted in \cite{Calcagni:2022shb} to remove branch-cut modes from the spectrum but limited here to a sector of complex-conjugate particles. Although the prescription itself has a solid motivation, it adds a complication for the analysis of the classical dynamics, since the classical equations of motion in that case are not those from the bare action but must be derived through a non-trivial procedure called classicization \cite{Anselmi:2019rxg,Anselmi:2018bra,Anselmi:2019nie,Anselmi:2025uda}. The equations of motion for the action \Eq{gravac} with $c_0=0$ and $c_0\neq 0$ have been derived in, respectively, \cite{Calcagni:2021aap} and \cite{Salvador-Garcia:2026akl}; the cosmological dynamics on a homogeneous and isotropic background is currently under study \cite{Salvador-Garcia:2026akl}. However, these equations of motion should be classicalized or ``classicized,'' since the continuum of complex-conjugate modes of the propagator discontinuity have been fakeonized, i.e., made purely virtual. The classicized equations of motion can be derived only approximately, since they are perturbative in the interaction coupling \cite{Anselmi:2025uda}. This is not a problem but, rather, a complication, due to the highly nonlinear structure of the gravitational sector. The derivation of the classicized equations is also work in progress.

\bigskip

\noindent{\bf Note added.} During the completion of the manuscript, we became aware of a work by D.~Anselmi focused on the application of the fakeon prescription to fractional models \cite{Anselmi:2026hvr}. The reader can find a valuable insight on this topic therein.


\section*{Acknowledgments}

We thank D.~Anselmi and L.~Rachwa\l\ for useful discussions. The authors (G.C.\ as PI) are supported by grant PID2023-149018NB-C41 funded by the Spanish Ministry of Science, Innovation and Universities MCIN/AEI/10.13039/501100011\-033. The work of G.C.\ was made possible also through the support of the WOST, \href{https://withoutspacetime.org}{WithOut SpaceTime project}, supported by Grant ID 63683 from the John Templeton Foundation. The opinions expressed in this work are those of the authors and do not necessarily reflect the views of the John Templeton Foundation.


\appendix
\addtocontents{toc}{\protect\setcounter{tocdepth}{1}}


\section{Calculation of the residue \Eq{resi2}}\label{appA}

Let $z_n^\pm\neq 0$. We prove \Eq{resi2} in two ways. One is through l'H\^opital rule:
\ban
{\rm Res}[\tilde G_\pm(z),z_n^\pm] &=& \lim_{z\to z_n^\pm}\frac{1}{(\ups-1)!}\frac{\rmd^{\ups-1}}{\rmd z^{\ups-1}}[(z-z_n^\pm)^\ups \bar G_\pm(z)]\nn
&=&\lim_{z\to z_n^\pm}\frac{1}{(\ups-1)!}\frac{\rmd^{\ups-1}}{\rmd z^{\ups-1}}\frac{\mp(z-z_n^\pm)^\ups}{z[1+c_2(\pm z)^\om]^\ups}\nn
&=& \lim_{z\to z_n^\pm}\frac{1}{(\ups-1)!}\frac{\rmd^{\ups-1}}{\rmd z^{\ups-1}}\left\{\frac{\mp 1}{z[\pm c_2\om(\pm z_n^\pm)^{\om-1}]^{\ups}}+O(z-z_n^\pm)\right\}\nn
&=&\lim_{z\to z_n^\pm}\frac{\mp 1}{(\ups-1)!}\frac{(-1)^\ups(\ups-1)!}{z^\ups[\pm c_2\om(\pm z_n^\pm)^{\om-1}]^{\ups}}\nn
&=&\frac{\mp 1}{[-c_2\om(\pm z_n^\pm)^{\om}]^{\ups}}\,.
\ean
The other is through the series expansions
\ben
\frac{1}{z}=\frac{1}{z_n^\pm+(z-z_n^\pm)}=\frac{1}{z_n^\pm}\sum_{k=0}^{\infty} \left(\frac{z-z_n^\pm}{z_n^\pm}\right)^k (-1)^k\,,
\een
and
\ban
[1+c_2(\pm z)^\om]^\ups &=& \left[\sum_{h=1}^{\infty}\frac{\rmd^h c_2(\pm Z)^\om}{\rmd Z^h}\bigg|_{Z=z_n^\pm} \frac{\left(z-z_n^\pm\right)^h}{h!}\right]^\ups\\
&=&\left(z-z_n^\pm\right)^\ups\left[\sum_{h=1}^{\infty}\frac{\rmd^h c_2(\pm Z)^\om}{\rmd Z^h}\bigg|_{ Z=z_n^\pm} \frac{\left(z-z_n^\pm\right)^{h-1}}{h!}\right]^\ups,
\ean
which yield
\ban
\frac{-1\pm 1}{z}\mp\frac{1}{z[1+ c_2\om(\pm z)^{\om}]^{\ups}}&=&
\frac{1}{z_n^\pm}\sum_{k=0}^{\infty} \left(-\frac{1}{z_n^\pm}\right)^k \left(z-z_n^\pm\right)^{k}\Bigg\{-1\pm 1\\
&&\mp\left(z-z_n^\pm\right)^{-\ups}\left[\sum_{h=1}^{\infty}\frac{\rmd^h c_2(\pm Z)^\om}{\rmd Z^h}\bigg|_{ Z=z_n^\pm} \frac{\left(z-z_n^\pm\right)^{h-1}}{h!}\right]^{-\ups}\Bigg\}.
\ean
The residue is the coefficient of the term $\left(z-z_n^\pm\right)^{-1}$, obtained for $k=u-1$ and $h=1$:
\ben
\pm\left(-\frac{1}{z_n^\pm}\right)^u \left[\frac{\rmd c_2(\pm Z)^\om}{\rmd Z}\bigg|_{ Z=z_n^\pm} \right]^{-\ups}
= \frac{\pm1}{(-z_n^\pm)^{u} \left[ \pm c_2 \om (\pm z_n^\pm)^{\om-1} \right]^{\ups}} 
= \frac{\pm1}{ \left[-c_2 \omega (\pm z_n^\pm)^{\om} \right]^{\ups}}\,.
\een


\section{Poles distribution for the HP operator}\label{appB}

In this appendix, we prove several theorems about the distribution of the 
 complex poles of the Green's function \Eq{propz} on the Riemann surface.

Poles with a non-vanishing imaginary part create an elaborate pattern on the Riemann surface. Let us focus first on pairs $(z_n^\pm,z_{n_*}^\pm)\coloneqq (z_n^\pm,z_n^{\pm*})$ of complex-conjugate poles. For any $\g$, the phases $\vp_n$ and $\vp_{n_*}$ of the complex-conjugate poles $z_n$ and $z_{n_*}$ are related to each other by
\be\label{constrong1}
\vp_n^\pm=-\vp_{n_*}^\pm+2\pi k\,,\qquad k\in\mathbb{Z}\,.
\ee
To see where these conjugate pairs lie, we write down the conditions for the poles $z_n^\pm$ and $z_{n_*}^\pm$ to belong to, respectively, sheet $S_l^\pm$ and $S_{l_*}^\pm$:
\ba
&& 2l\mp\frac12 < \frac{\vp_n^\pm}{\pi} < 2l\mp\frac12+1\,,\label{inter1}\\
&& 2l_*\mp\frac12 < \frac{\vp_{n_*}^\pm}{\pi} < 2l_*\mp\frac12+1\,.\label{inter2}
\ea
Plugging \Eq{constrong1} into \Eq{inter1}, we get
\be\label{inter3}
2(k-l)\pm\frac12-1 < \frac{\vp_{n_*}^\pm}{\pi} < 2(k-l)\pm\frac12\,.
\ee
Comparing \Eq{inter3} with \Eq{inter2}, we reach the condition
\be\label{mmst}
l+l_*=k-\de_\mp\,.
\ee
Then:
\begin{itemize}
\item When the pole $z_n=z_n^+$ is in the first or fourth quadrant, we take the $-$ sign in \Eq{mmst}, so that $l+l_*=k$. If one pole is in a certain sheet, its conjugate is in a different one unless $l_*=l=k/2$:
\be\label{beo1}
z_n^+\in S_l^+\,,\qquad z_{n_*}^+\in S_{k-l}^+\,.
\ee
\item When the pole $z_n=z_n^-$ is in the second or third quadrant, we take the $+$ sign in \Eq{mmst}, so that $l+l_*=k-1$:
\be\label{beo2}
z_n^-\in S_l^-\,,\qquad z_{n_*}^-\in S_{k-l-1}^-\,.
\ee
The conjugate pair $(z_n^-,z_{n*}^-)$ is in the same sheet only if $l_*=l=(k-1)/2$.
\end{itemize}

The following theorems single out the Riemann sheets where complex poles are absent or appear in conjugate pairs.
\begin{theo}[{\bf HP: complex-conjugate pairs ---real $\bm{0<\g<2\ups+1}$}]\label{theo3-A}
When
\be\label{gulim}
0<\g<2\ups+1\,,
\ee
there are no complex-conjugate pairs $(z_n^\pm,z_n^{\pm*})$ in any Riemann sheet $S_{l}^\pm$.
\end{theo}
\begin{proof}
In order to see how many complex poles $z_n^\pm$ lie in any given sheet $S_l^\pm$, we set the conditions on $n$ such that \Eq{inter1} holds. From \Eq{c2pha} and \Eq{inter1}, $2l-1/2<(2n+1)/\om<2l+1/2$ and $2l+1/2<(2n+1)/\om-1<2l+3/2$, hence
\ba
\hspace{-1.2cm}&& S_l^+:\qquad \left\{\begin{matrix} 
-\frac12<\om<0:\quad \hphantom{-}\frac{\om}{4}+\om l-\frac{1}{2}<n<-\frac{\om}{4}+\om l-\frac{1}{2}\\
\hspace{.8cm}\om>0:\quad -\frac{\om}{4}+\om l-\frac{1}{2}<n<\frac{\om}{4}+\om l-\frac{1}{2}\end{matrix}\right.\,,\label{condia}\\
\hspace{-1.2cm}&& S_l^-:\qquad \left\{\begin{matrix} 
-\frac12<\om<0:\quad \hphantom{-}\frac{\om}{4}+\om (l'+1)-\frac12<n<-\frac{\om}{4}+\om (l'+1)-\frac12\\
\hspace{.6cm}\om>0:\quad -\frac{\om}{4}+\om (l+1)-\frac12<n<\frac{\om}{4}+\om (l+1)-\frac12
\end{matrix}\right.\,,\label{condib}
\ea
where we split the intervals in $\om$ according to \Eq{gom}. The inequalities $n>a$ and $n<b$ are equivalent to $n\geq \lfloor a\rfloor+1$ and $n\leq \lceil b\rceil-1$. Thus, the admissible integers in the interval $a<n<b$ are $\lfloor a\rfloor+1,\,\lfloor a\rfloor+2,\,\dots,\,\lceil b\rceil-1$ and their number in any given Riemann sheet $S_l^\pm$ is
\be\label{flocei}
N_l^\pm=\lfloor b\rfloor-\lceil a\rceil+1\,.
\ee
In all the cases \Eq{condia} and \Eq{condib}, the length of the interval is $b-a=|\om|/2$, implying that $N_l^\pm=0$ or $N_l^\pm=1$ for
\be\label{gulim2}
-\frac1\ups<\om<2
\ee
(hence $0<\g<2u+1$). 
 Therefore, in this range of $\om$ values there are not enough poles to create complex-conjugate pairs on any Riemann sheet. For $2<\om<4$, $n=-1,0$ and $N_l^\pm=2$. For $\om>4$, $N_l$ increases further.
\end{proof}
Now we present a powerful result, compatible with all the above theorems, stating that it is always possible to find Riemann sheets with no real or complex pole therein.
\begin{theo}[{\bf HP: empty sheets ---real $\bm{0<\g<2\ups+1}$}]\label{theo6-A}
There always exists a completely empty Riemann sheet $S_{l,l'}=S_l^+\cup S_{l'^-}$ for any non-integer $\om$ in the interval \Eq{gulim2}. 
\end{theo}
\begin{proof}
Consider first $S_l^+$. The second interval in \Eq{condia} is centered at $\om l-1/2$ with half-length $\om/4<1/2$, which implies that there is no integer in it ($N_l^+=0$) if, and only if, the distance between $\om l-1/2$ and the closest integer $n$ is greater than or equal to $\om/4$. Defining ${\rm dist}(x)\coloneqq \min(\{x\},1-\{x\})$, where $\{x\}=x-\lfloor x\rfloor$ the fractional part of $x$, and using the property ${\rm dist}(x-n)={\rm dist}(x)$ for any integer $n$, we have
\be\label{distcon+}
S_l^+:\qquad 
{\rm dist}\!\left(\om l-\frac{1}{2}\right)\geq \frac{\om}{4}\,,
\ee
Similarly, for $S_{l'}^-$ the second interval in \Eq{condib} is centered at $\om(l'+1)-1/2$ with half-length $\om/4$, which implies that $N_{l'}^-=0$ if, and only if,
\be\label{distcon-}
S_{l'}^-:\qquad 
{\rm dist}\left[\om(l'+1)-\frac12\right]\geq \frac{\om}{4}\,.
\ee
For $\om>2$, one has $N_l^\pm\geq 1$ because the interval is $>1$ and contains at least one non-trivial integer $n$. For $-1/\ups<\om<0$, both \Eq{distcon+} and \Eq{distcon-} are trivially satisfied since the distance is a non-negative number. For $0<\om<2$ irrational, there are infinitely many $l$ such that $N_l^\pm=0$. Since we can arbitrarily glue together two different half-planes $S_l^+$ and $S_{l'}^-$ into a Riemann sheet $S_{l,l'}=S_l^+\cup S_{l'^-}$, we can always find two values $(l,l')$ such that $S_{l,l'}$ is empty of poles for any given $0<\om<2$. For irrational $\om$, one way to show this is to note that the sequences of the rational part $\{\om l\}$ (as well its counterpart with $l\to l'+1$) is equidistributed modulo 1 (Weyl theorem), hence the number $\om l-1/2$ is dense in $[0,1[$. Therefore, values arbitrarily close to any fractional part, or arbitrarily far from integers, occur such that \Eq{distcon+} and \Eq{distcon-} hold. 

The extension to rational $\om=p'/q'$ is straightforward. Take $S_l^+$, substitute $\om=p'/q'$ in \Eq{distcon+} and multiply by $2q'$ and find $|2p'l-q'|\geq p'/2$, where $p'l=0,1,\dots,q-1$ modulo $q$. Therefore, $|2p'l-q'|=|2m-q'|$, $m=0,1,\dots,q-1$, and we need to check whether and when the following condition holds:
\be\label{fail}
N_l^+=0\,:\qquad |2m-q'|\geq\frac{p'}{2}\,,\qquad m=0,1,\dots,q'-1\,.
\ee
The smallest possible value of $|2m-q'|$ is 0 if $q'$ is even ($m=q'/2$) and 1 if $q'$ is odd ($m=(q'\pm 1)/2$), while the largest value is $q'$ ($m=0$). Thus, inequality \Eq{fail} holds for all $m$ only if the maximum respects it: $q'> p'/2$, i.e., $\om<2$.
\end{proof}
Fixing $\om$, one can find the pairs $(l,l')$ (which always exists) obeying \Eq{distcon+} and \Eq{distcon-}. Conversely, fixing $(l,l')$ one can find the allowed range of $\om$ in the interval $]0,2[$. For example, for $l=0=l'$ it must be
\be\label{00}
0<\om\leq\frac25\,\cup\,\frac23\leq\om\leq\frac65\,.
\ee

If one fixes $\ups$ such that $\g$ falls outside the interval \Eq{gulim}, then there may be conjugate pairs in some Riemann sheets. 
\begin{theo}[{\bf HP: complex-conjugate poles ---irrational $\bm{\g>2\ups+1}$}]\label{theo4-A}
For $\g>2\ups+1$ ($\om>2$) irrational, the only Riemann sheets that can host 
conjugate pairs are $S_0^+$ and $S_{-1}^-$.
\end{theo}
\begin{proof}
Equations \Eq{c2pha} and \Eq{constrong1} imply that
\be\label{constrong2}
n_*= -n-1+\om(k+\de_\mp)\,.
\ee
Since $\om$ is irrational, $n_*$ is an integer only if $k+\de_\mp=0$, i.e., if
\ba
&& S_l^+:\qquad \!\! k=0\,,\qquad\,\,\,\, n_*=-n-1\,,\label{stro1b}\\
&& S_l^-:\qquad \!\! k=-1\,,\qquad n_*=-n-1\,.\label{stro3b}
\ea
Equation \Eq{mmst} is in force, so that 
\ba
&& z_n^+\in S_l^+\,,\qquad z_{n_*}^+\in S_{-l}^+\,,\\
&& z_n^-\in S_l^-\,,\qquad z_{n_*}^-\in S_{-l-2}^-\,.
\ea
In particular, the conjugate poles $z_{n_*}^\pm$ belong to the same sheet of $z_n^\pm$ only if $l=-\de_\mp$ ($l=0$ for $z_{n_*}^+$ and $l=-1$ for $z_{n_*}^-$). 
\end{proof}

\begin{theo}[{\bf HP: complex-conjugate poles ---rational $\bm{\g=p/q>2\ups+1}$}]\label{theo5-A}
Let $\g=p/q>2\ups+1$ ($\om=p'/q'$ given in \Eq{ppqp}) be irreducible. The only Riemann sheets in the sequence \Eq{gooseq} that can host a complex-conjugate pair $(z_n,z_{n_*})$ are $S_0^+$ for any $q'$ and $S_{{q'}/{2}}^+$ and $S_{{q'}/{2}-1}^-$ if $q'$ is even and $p'$ is odd. In particular, to find pairs in $S_0^-$, one must have $q'=2$. When $p'=1$, there are no complex poles anywhere.
\end{theo}
\begin{proof}
For rational $\g$, \Eq{constrong2} is an integer if, and only if, $k+\de_\mp$ is an integer multiple of $q'$. This fixes all admissible values for $k$ in the sequence \Eq{gooseq}:
\bs\ba
&& S_l^+:\qquad \!\! k=0\,,\qquad\,\,\,\, n_*=-n-1\,,\label{stro1}\\
&&\hphantom{S_l^+:}\qquad k=q'\,,\qquad\,\,\, n_*=-n-1+p'-q'\,,\label{stro2}\\
&& S_l^-:\qquad \!\!  k=-1\,,\qquad n_*=-n-1\,,\label{stro3}\\
&& \hphantom{S_l^-:}\qquad \!  k=q'-1\,,\quad n_*=-n-1+p'-q'\,.\label{stro4}
\ea\es
From \Eq{mmst}, the two complex-conjugate poles $(z_n^+,z_{n*}^+)$ lie in the same sheet only when $l_*=l=k/2$, so either in $S_0^+$ (\Eqq{stro1}) or in $S_{q'/2}^+$ when $q'$ is even (\Eqq{stro2}; hence $p'$ is odd). The conjugate pair $(z_n^-,z_{n*}^-)$ is in the same sheet only if $l_*=l=(k-1)/2$, hence either $l=-1$ (\Eqq{stro3}, i.e., outside \Eq{gooseq}) or $l=q'/2-1$ (\Eqq{stro4}); the latter case requires $q'$ to be even (hence $p'$ odd). Thus, these pairs can be found on the sheet $S_{q'/2-1}^-$ and, in particular, in $S_0^-$ only if $q'=2$ ($k=1$). Finally, since $p'$ counts the number of poles of the equation $1+(\pm z)^{p'/q'}=0$, there is only one real pole (of order $u$) when $p'=1$.
\end{proof}

Theorems \ref{theo3-A}, \ref{theo4-A} and \ref{theo5-A} individuate the sheets where one can find whole complex-conjugate pairs but they do not count them as a function of $\g$ precisely. In particular, they do not determine whether such pairs are actually accompanied by unpaired complex poles (or, even worse, are absent and only unpaired poles are present), nor whether sheets allowing for unpaired poles are populated or actually empty. For example, when $\g$ is rational $S_0^+$ can contain unpaired complex poles when setting $l=0$ and $l_*=k$, but for which values of $\g$? Are there really pairs in $S_{-1}^-$ for $\g$ irrational and does this sheet also or only host unpaired poles?

In general, for $\g>2\ups+1$ there are unpaired complex poles on $S_l^\pm$ depending on the values of $\g$ and $l$. For example, for $\om=3/5$ and $l=0$ one has an unpaired pole $z_0^-=\rme^{\rmi2\pi/3}$ in $S_0^-$. We do not formalize the full classification of isolated poles on the Riemann surface here. 


\section{Poles distribution for the HS, HA and NH operators}\label{appC}

In this appendix, we redo the work of appendix~\ref{appB} but for the Green's functions \Eq{hsim}--\Eq{nhrep}.


\subsection{HS operator}\label{appC1}

This section is devoted to the pole distribution of the Green's function \Eq{hsim}.

\subsubsection{Real poles}\label{appC1a}

\begin{theo}[{\bf HS: real poles ---irrational $\bm{\g>0}$}]\label{theo1-B1}
For $\g>0$ irrational, the only real pole of the Green's function \Eq{hsim} is $z=1$ in the sheet $S_0^+$.
\end{theo}
\begin{proof} The phases in \Eq{c2pha1} can never be integer multiples of $\pi$ unless we are in $S_l^+$ and $n=0=c$. From
\be
S_l^+:\qquad \frac{2n}{\g-1} = 2l\,,\label{Smcond1irrB1}
\ee
we conclude that $l=0$. All other sheets $S_l^-$ and $S_{l\neq 0}^+$ are empty of real poles.
\end{proof}

\begin{theo}[{\bf HS: real poles ---rational $\bm{\g=p/q>0}$}]\label{theo2-B1}
When $\g=p/q$ is an irreducible rational, the Green's function \Eq{hsim} has a physical real pole $z=1$ living in $S_0^+$. The tachyonic real pole $z=-1$ is present on all Riemann sheets $S_{l}^-$ such that $l=(2n+1)q/[2(p-q)]-1$. In particular, in the main sequence \Eq{gooseq'1} there is no $z=-1$ pole in $S_{q-1}^-$, while there is a tachyonic pole in $S_0^-$ only if
\be\label{nogam}
\g= n+\frac32=\frac12,\,\frac32,\,\frac52,\,\frac72\,,\dots\,,\qquad n\in\mathbb{Z}\,.
\ee
\end{theo}
\begin{proof}
The $z=-1$ pole is present on the sheets $l=(2n+1)q/[2(p-q)]-1$, where $l$ is integer only if $(2n+1)q$ is a multiple of $2(p-q)$, i.e., if $q$ is even. In particular, there is no real pole in $S_{-1}^-$ and $S_{q-1}^-$; $S_0^-$ has the tachyonic pole only if $\g=p/q = (2n+3)/2=1/2,3/2,\dots$; $S_1^-$ hosts it only if $\g=(2n+5)/4=1/4,3/4,\dots$; and so on. Avoiding these values guarantees the absence of this mode in any given sheet.

The $z=1$ pole is found on $S_0^+$ (corresponding to $n=0$) and on all Riemann sheets determined by \Eq{Smcond1irrB1}, i.e., when $n=l(p-q)/q$ is integer, hence when $l$ is an integer multiple of $q$. The only sheet in this sequence $\{S_0^+,S_{\pm q}^+,S_{\pm 2q}^+,\dots\}$ included in the minimal set \Eq{gooseq'1} is $S_0^+$.
\end{proof}
In particular, we have no other real mode except $z=0$ only in two cases: (I) $\g>0$ irrational and the theory is not defined on $S_0^+$, which is always possible; (IIa) $\g$ rational and not taking any of the values in \Eq{nogam} (if the theory is defined on $S_0^-$; the case on other sheets $S_l^-$ can be inferred from the formul\ae\ above), with (II) the theory not defined on $S_0^+$.

\subsubsection{Complex poles}\label{appC1b}

\begin{theo}[{\bf HS: complex-conjugate pairs ---real $\bm{0<\g<3}$}]\label{theo3-B1}
When
\be\label{gulim-B1}
0<\g<3\,,
\ee
there are no complex-conjugate pairs $(z_n^\pm,z_n^{\pm*})$ in any Riemann sheet $S_{l}^\pm$.
\end{theo}
\begin{proof}
From \Eq{c2pha1} and \Eq{inter1}, $2l-1/2<2n/(\g-1)<2l+1/2$ and $2l+1/2<(2n+2-\g)/(\g-1)<2l+3/2$, hence
\ba
\hspace{-.9cm}&& S_l^+:\quad \left\{\begin{matrix} 0<\g<1:\quad \hphantom{-}\frac{\g-1}{4}+(\g-1)l<n<-\frac{\g-1}{4}+(\g-1)l\\
\hspace{.4cm}\g>1:\quad  -\frac{\g-1}{4}+(\g-1)l<n<\frac{\g-1}{4}+(\g-1)l\end{matrix}\right.\,,\label{condia-B1}\\
\hspace{-.9cm}&& S_l^-:\quad \left\{\begin{matrix} 0<\g<1:\quad \hphantom{-}\frac{\g-1}{4}-\frac12+(\g-1)(l+1)<n<-\frac{\g-1}{4}-\frac12+(\g-1)(l+1)\\
\hspace{.4cm}\g>1:\quad -\frac{\g-1}{4}-\frac12+(\g-1)(l+1)<n<\frac{\g-1}{4}-\frac12+(\g-1)(l+1)
\end{matrix}\right..\label{condib-B1}
\ea
In all the cases \Eq{condia-B1} and \Eq{condib-B1}, the length of the interval is $b-a=|\g-1|/2$, implying that $N_l^\pm=0,1$ for $0<\g<3$, $N_l^\pm=1,2$ for $3<\g<5$, $N_l^\pm=2,3$ for $5<\g<7$, and so on.
\end{proof}
This theorem is insensitive to the eventual presence of unpaired complex poles. To fill this gap, we now ask which sheets $S_l^+$ and $S_l^-$ are empty of real and complex poles. The analogue of Theorem~\ref{theo6-A} is
\begin{theo}[{\bf HS: empty sheets ---real $\bm{0<\g<3}$}]\label{theo6-B1}
There always exists a completely empty Riemann sheet $S_{l,l'}=S_l^+\cup S_{l'^-}$ for any non-integer $\g$ in the interval \Eq{gulim-B1},
except when
\be\label{exep}
2<\g=\frac{p}{q}=3-\frac{1}{q}\,,\qquad p=3q-1\,.
\ee
For the values \Eq{exep}, there always is an unpaired complex pole in at least one half of the sheet $S_{l,l'}$, except on the sheet $S_0^+\cup S_{q-1}^-$, which only has the $z=1$ pole \Eq{pole+}.
\end{theo}
\begin{proof}
Consider first $S_l^+$. The second interval in \Eq{condia-B1} is centered at $(\g-1)l$ with half-length $(\g-1)/4<1/2$, which implies that there is no integer in it ($N_l^+=0$) if, and only if, the distance between $(\g-1)l$ and the closest integer $n$ is greater than or equal to $(\g-1)/4$:
\be\label{distcon+B1}
S_l^+:\qquad 
{\rm dist}(\g l)\geq \frac{\g-1}{4}\,,
\ee
where ${\rm dist}(x)\coloneqq \min(\{x\},1-\{x\})$, $\{x\}=x-\lfloor x\rfloor$ is the fractional part of $x$ and we used the property ${\rm dist}(x-k)={\rm dist}(x)$ for any integer $k$. Similarly, for $S_{l'}^-$ the second interval in \Eq{condib-B1} is centered at $(\g-1)(l'+1)-1/2$ with half-length $(\g-1)/4$, which implies that $N_{l'}^-=0$ if, and only if,
\be\label{distcon-B1}
S_{l'}^-:\qquad 
{\rm dist}\left[\g(l'+1)-\frac12\right]\geq \frac{\g-1}{4}\,
\ee
For $\g>3$, one has $N_l^\pm\geq 1$ because $b-a>1$ and, therefore, the interval contains at least one non-trivial integer $n$. For $0<\g<1$, the right-hand side of \Eq{distcon+B1} and \Eq{distcon-B1} is negative and the inequalities are trivial. For $1<\g<3$ ($0<b-a<1$), there are infinitely many $l$ such that $N_l^\pm=0$. The condition $1<\g<3$ ensures the above intervals are short enough, while the necessary and sufficient conditions \Eq{distcon+B1} and \Eq{distcon-B1} ensure such intervals are positioned so as not to capture any integer $n$. Since we can arbitrarily glue together two different half-planes $S_l^+$ and $S_{l'}^-$ into a Riemann sheet $S_{l,l'}=S_l^+\cup S_{l'^-}$, we can always find two values $(l,l')$ such that $S_{l,l'}$ is empty of poles for any given $1<\g<3$. Indeed, the sequences of the rational parts $\{\g l\}$ and $\{\g (l'+1)\}$ are dense in $[0,1]$ when $\g$ is irrational, so that values arbitrarily close to any fractional part, or arbitrarily far from integers, occur such that \Eq{distcon+B1} and \Eq{distcon-B1} hold. In the case of rational $\g$, take $S_l^+$ and note that the fractional part $\{\g l\}=j/q$, where $j=0,1,\dots,q-1$ runs on the main sequence \Eq{gooseq'1}. Therefore, ${\rm dist}(\g l)=\min(j/q,1-j/q)$. The largest possible distance ${\rm dist}_{\rm max}$ occurs at $j=\lfloor q/2\rfloor$, so that ${\rm dist}_{\rm max}=1/2$ if $q$ is even and ${\rm dist}_{\rm max}=(q-1)/(2q)$ if $q$ is odd. Applying \Eq{gooseq'1}, we get ${\rm dist}_{\rm max}\geq(p-q)/(4q)$. Thus, we have $p\leq 3q$ for $q$ even and $p\leq3q-2$ for $q$ odd. The former inequality is automatically satisfied inside the interval $q<p<3q$, while the latter fails when $p=3q-1$ (in both cases $p=3q$ is excluded since $\g$ is non-integer). These are the exceptional values \Eq{exep}. The proof on $S_l^-$ gives the same result but this time $q$ is odd. Replacing $p=3q-1$ in \Eq{c2pha1} we get $\vp_n^+=2nq\pi/(2q-1)$ and $\vp_n^-=[(2n-1)q+1]\pi/(2q-1)$, which are not integer multiples of $\pi$. 
\end{proof}
Examples of empty sheets are given in table~\ref{tab1}. Examples of values of $\g$ such that we cannot find any empty sheet are $\g=8/3,14/5,20/7$.
\begin{table}[h!]
	\begin{center}
		\begin{tabular}{ccccc}
$\bm{\g}$ & $\frac{5}{3}$ & $\frac{11}{5}$ & $\frac{5}{2}$ & $\frac{11}{4}$ \\\hline
$\bm{l}$  & 1 & 2 & 1 & 2 \\
$\bm{l'}$ & 0 & 0 & 1 & 3
		\end{tabular}
	\end{center}
\caption{Examples of empty sheets $S_{l,l'}=S_l^+\cup S_{l'^-}$ for some rational values of $1<\g<3$.\label{tab1}}
\end{table}	

The next theorems cover the range $\g>3$.
\begin{theo}[{\bf HS: complex-conjugate poles ---irrational $\bm{\g}>3$}]\label{theo4-B1}
For $\g>3$ irrational, the only Riemann sheets that can host complex-conjugate pairs $(z_n,z_{n_*})$ are $S_0^+$ and $S_{-1}^-$.
\end{theo}
\begin{proof} The proof is identical to the one in Theorem~\ref{theo4-A} except that here, from \Eq{c2pha1} and \Eq{constrong1},
\be\label{constrong2-B1}
n_*=-n-\de_\mp+(\g-1)(k+\de_\mp)\,,
\ee
which is an integer only if
\be
S_l^+:\qquad \!\! k=0\,,\qquad\,\,\,\, n_*=-n\,,\label{stro1-B1}
\ee
and \Eq{stro3b} hold. The rest follows as in Theorem~\ref{theo4-A}.
\end{proof}

\begin{theo}[{\bf HS: complex-conjugate poles ---rational $\bm{\g=p/q}>3$}]\label{theo5-B1}
Let $\g=p/q>3$ be irreducible. The only Riemann sheets in the sequence \Eq{gooseq'1} that can host a complex-conjugate pair $(z_n,z_{n_*})$ are $S_0^+$ for any $q$ and $S_{{q}/{2}}^+$ and $S_{{q}/{2}-1}^-$ if $q$ is even and $p$ is odd. In particular, to find pairs in $S_0^-$, one must have $q=2$ and $p\geq 7$ odd.
\end{theo}
\begin{proof}
For rational $\g$, \Eq{constrong2-B1} is an integer if, and only if, $k+\de_\mp$ is an integer multiple of $q$. This fixes all admissible values for $k$ in the sequence \Eq{gooseq'1}, through \Eq{stro1-B1} and
\bs\ba
&&S_l^+:\qquad k=q\,,\qquad\,\,\,\, n_*=-n+p-q\,,\label{stro2-B1}\\
&& S_l^-:\qquad \!\!  k=-1\,,\qquad n_*=-n-1\,,\label{stro3-B1}\\
&& \hphantom{S_l^-:}\qquad \!\!  k=q-1\,,\quad n_*=-n-1+p-q\,.\label{stro4-B1}
\ea\es
From \Eq{mmst}, the two complex-conjugate poles $(z_n^+,z_{n*}^+)$ lie in the same sheet only when $l_*=l=k/2$, so either in $S_0^+$ (\Eqq{stro1-B1}) or in $S_{q/2}^+$ when $q$ is even (\Eqq{stro2-B1}; hence $p$ is odd). The conjugate pair $(z_n^-,z_{n*}^-)$ is in the same sheet only if $l_*=l=(k-1)/2$, hence either $l=-1$ (\Eqq{stro3-B1}, i.e., outside \Eq{gooseq'1}) or $l=q/2-1$  (\Eqq{stro4-B1}); the latter case requires $q$ to be even (hence $p$ odd). Thus, these pairs can be found on the sheet $S_{q/2-1}^-$ and, in particular, in $S_0^-$ only if $q=2$ ($k=1$), in which case $n_*=-n-3+p$ and $p\geq 7$ odd (for $p\geq 5$, $\g<3$ and there are no pairs).
\end{proof}
In particular, from \Eq{nogam}, the $q=2$ case not only has a symmetric distribution of the poles in $S^-_0$ around the real axis but also the $z=-1$ real pole. Therefore, this case should be excluded on physical grounds. 

The above theorems do not specify the distribution of unpaired complex poles for $\g>3$. 
In general, there are two possible configurations when $\g$ is rational (figure~\ref{fig5}). (a) Complex poles are symmetric with respect to the real axis, so that they are all organized into complex-conjugate pairs plus the tachyonic real pole.\footnote{From Theorems~\ref{theo2-B1} and \ref{theo5-B1}, for rational $\g$ it follows that whenever there are pairs in $S_{l}^-$ (i.e., for $l=q/2-1$ with $q$ even), there is also the $z=-1$ pole. Indeed, the equality $q/2-1=l=(2n+1)q/[2(p-q)]-1$ always holds.} (b) Complex poles are distributed asymmetrically with respect to the real axis and there are neither conjugate pairs nor the tachyon. The counting of the poles on $S_l^-$ is given by \Eq{condib-B1} and \Eq{flocei}.
\begin{figure}[ht]
	\bc
	\includegraphics[width=7.5cm]{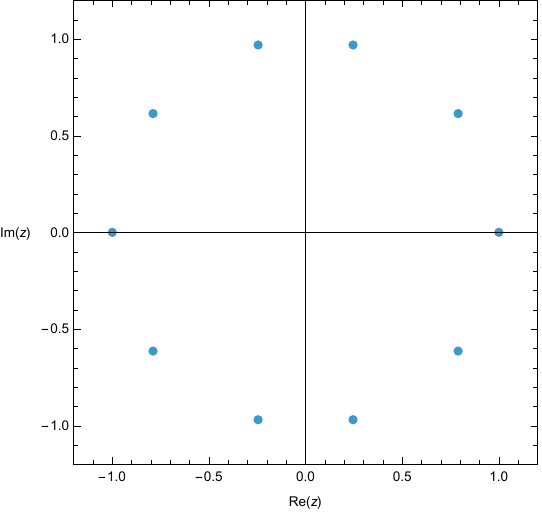}\includegraphics[width=7.5cm]{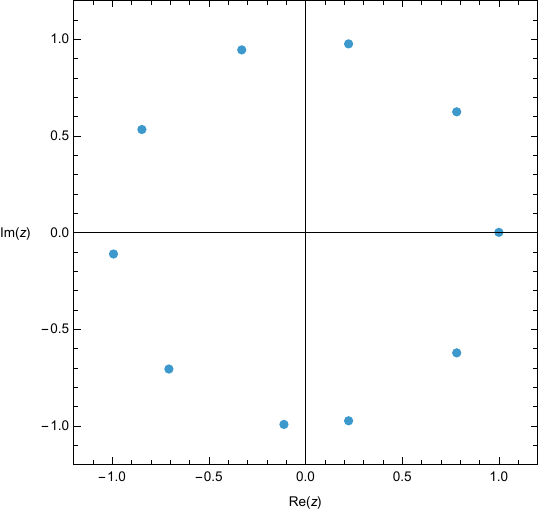}
	\ec
	\caption{\label{fig5} Examples of pole distribution in $S_0=S_0^+\cup S_0^-$ symmetric in $S_0^-$ (left, $\g=21/2$) and asymmetric (right, $\g=31/3$) for the theory with HS operator.}
\end{figure}


\subsection{HA operator}\label{appC2}

This section is devoted to the pole distribution of the Green's function \Eq{hasy} for odd $\Upsilon$. (The HA theory is identical to the HS one when $\Upsilon$ is even.)

\subsubsection{Real poles}\label{appC2a}

\begin{theo}[{\bf HA: real poles ---irrational $\bm{\g>0}$}]\label{theo1-B2}
For $\g>0$ irrational and $\Upsilon$ odd, there are no real poles anywhere on the Riemann surface.
\end{theo}
\begin{proof} Real poles correspond to phases $\vp_n=c\pi$, where $c\in\mathbb{Z}$ is related to $l$ by
\ba
&&S_l^+:\qquad \frac{2(n+1)-\Upsilon}{\g-1} =c=2l\,,\label{Smcond1irr1B2}\\
&&S_l^-:\qquad \frac{2(n+1)-\g}{\g-1} =c=2l+1\,.\label{Smcond1irr2B2}
\ea
When $\g>0$ is irrational and $\Upsilon$ is odd, these phases can never be integer multiples of $\pi$.
\end{proof}

\begin{theo}[{\bf HA: real poles ---rational $\bm{\g=p/q>0}$}]\label{theo2-B2}
When $\g=p/q$ is an irreducible rational and $\Upsilon$ is odd, the Green's function \Eq{hasy} does not have the real pole $z=1$. The tachyonic real pole $z=-1$ is present on all Riemann sheets $S_{l}^-$ such that $l=(2n+1)q/[2(p-q)]-1$. In particular, in the main sequence \Eq{gooseq'1} there is no $z=-1$ pole in $S_{q-1}^-$, while there is a tachyonic pole in $S_0^-$ only if $\g$ takes the values \Eq{nogam}.
\end{theo}
\begin{proof}
The analysis on $S_l^-$ for the $z=-1$ pole is the same as in the HS case, Theorem~\ref{theo2-B1}. The $z=1$ pole is found on all Riemann sheets determined by \Eq{Smcond1irr2B2}, i.e., when $n=l(p-q)/q+\Upsilon/2-1$ is integer, which is not the case if $\Upsilon$ is odd.
\end{proof}
The conclusion is similar to the HS case but it has more leeway in the admitted sheets when $\Upsilon$ is odd if we want to avoid the $z=\pm 1$ poles.

\subsubsection{Complex poles}\label{appC2b}

\begin{theo}[{\bf HA: complex-conjugate pairs ---real $\bm{0<\g<3}$}]\label{theo3-B2}
When \Eq{gulim-B1} holds, there are no complex-conjugate pairs $(z_n^\pm,z_n^{\pm*})$ in any Riemann sheet $S_{l}^\pm$.
\end{theo}
\begin{proof}
The proof is exactly the same as in Theorem~\ref{theo3-B1}, since the length of the intervals
\ba
\hspace{-1.2cm}&& S_l^+:\qquad \left\{\begin{matrix} 0<\g<1:\quad \frac{\g-1}{4}(4l+1)+\frac{\Upsilon}{2}-1<n<\frac{\g-1}{4}(4l-1)+\frac{\Upsilon}{2}-1\\
\hspace{.7cm}\g>1:\quad  \frac{\g-1}{4}(4l-1)+\frac{\Upsilon}{2}-1<n<\frac{\g-1}{4}(4l+1)+\frac{\Upsilon}{2}-1\end{matrix}\right.\,,\label{condia-B2}
\ea
and \Eq{condib-B1} is still $b-a=|\g-1|/2$.
\end{proof}

\begin{theo}[{\bf HA: empty sheets ---real $\bm{0<\g<3}$}]\label{theo6-B2}
There always exists a completely empty Riemann sheet $S_{l,l'}=S_l^+\cup S_{l'^-}$ for any non-integer $0<\g<3$ except when \Eq{exep} holds. For $\g=3-1/q$, there always is an unpaired complex pole in at least one half of the sheet $S_{l,l'}$, except on the sheet $S_0^+\cup S_{q-1}^-$, which only has the $z=1$ pole \Eq{pole+}.
\end{theo}
\begin{proof}
The proof is the same as in Theorem~\ref{theo6-B1} under the replacement
\ben
(\g-1)l\to (\g-1)l+\frac{\Upsilon}{2}-1
\een
everywhere, where $\Upsilon=1,2$. The HS case is recovered for $\Upsilon=2$, while for $\Upsilon=1$ the sheets $l$ and $l'+1$ are exchanged, so that $(l_{\rm HA},l_{\rm HA}')=(l_{\rm HS}'+1,l_{\rm HS}-1)$ ${\rm mod}\,q$. Conditions \Eq{distcon+B1} and \Eq{distcon-B1} are therefore interchanged but the exceptional values \Eq{exep} remain the same, since $2<\g<3$ in this case ($\Upsilon=2$). 
\end{proof}
In particular, the examples in table~\ref{tab1} are exactly the same.

\begin{theo}[{\bf HA: complex-conjugate poles ---irrational $\bm{\g>3}$}]\label{theo4-B2}
For $\g>0$ irrational, the only Riemann sheets that can host complex-conjugate pairs $(z_n,z_{n_*})$ are $S_0^+$ and $S_{-1}^-$.
\end{theo}
\begin{proof} The proof is identical to the one in Theorem~\ref{theo4-A} except that here, from \Eq{c2pha2} and \Eq{constrong1},
\be\label{constrong2-B2}
n_* = -n-1+(\Upsilon-1)(1-\de_\mp)+(\g-1)(k+\de_\mp)\,,
\ee
which is an integer only if
\be
S_l^+:\qquad \!\! k=0\,,\qquad\,\,\,\, n_*=-n-2+\Upsilon\,,\label{stro1-B2}
\ee
and \Eq{stro3b} hold. The rest follows as in Theorem~\ref{theo4-A}.
\end{proof}

\begin{theo}[{\bf HA: complex-conjugate poles ---rational $\bm{\g=p/q>3}$}]\label{theo5-B2}
Let $\g=p/q>0$ be irreducible. The only Riemann sheets in the sequence \Eq{gooseq'1} that can host a complex-conjugate pair $(z_n,z_{n_*})$ are $S_0^+$ for any $q$ and $S_{{q}/{2}}^+$ and $S_{{q}/{2}-1}^-$ if $q$ is even and $p$ is odd. In particular, to find pairs in $S_0^-$, one must have $q=2$ and $p\neq 3$ odd.
\end{theo}
\begin{proof}
The proof is identical to that of Theorem~\ref{theo4-B1} except that \Eqqs{stro1-B1} and \Eq{stro2-B2} are replaced by, respectively, \Eq{stro1-B2} and
\be
S_l^+\qquad k=q\,,\qquad\,\,\,\, n_*=-n-2+\Upsilon+p-q\,.\label{stro2-B2}
\ee
\end{proof}


\subsection{NH operator}\label{appC3}

This section is devoted to the pole distribution of the Green's function \Eq{nhrep}.

\subsubsection{Real poles}\label{appC3a}

\begin{theo}[{\bf NH: real poles ---irrational $\bm{\g>0}$}]\label{theo1-B3}
For $\g>0$ irrational, there are no real poles $z=\pm1$ in any Riemann sheet.
\end{theo}
\begin{proof} Trivial, since the phases in \Eq{c2pha3} can never be integer multiples of $\pi$.
\end{proof}

\begin{theo}[{\bf NH: real poles ---rational $\bm{\g=p/q>0}$}]\label{theo2-B3}
When $\g=p/q$ is an irreducible rational, the Green's function \Eq{nhrep} has no pole $z=1$ in any Riemann sheet. The tachyonic pole $z=-1$ is present on all Riemann sheets $S_{l}$ such that $l=(2n+1)q/[2(p-q)]$. In particular, in the main sequence \Eq{gooseq'2} there is no $z=-1$ pole in $S_0$ for any $q$ nor on any other sheet when $q$ is odd.
\end{theo}
\begin{proof}
The phase of a real pole in any given Riemann sheet $S_m$ is $\vp_n=c\pi$, $c\in\mathbb{Z}$, where $\vp_n$ is given in \Eq{c2pha3}. The phases $\vp_n=0,\pi,2\pi$ are excluded if $\g$ is non-integer, since they correspond, respectively, to $c=0$ ($p/q=-2n$), $c=1$ ($q=0$) and $c=2$ ($p/q=2n+2$). In all the other cases,
\ben
\frac{p}{q}=\frac{2n+c}{c-1}\,,\qquad c\neq 0, 1,2\,.
\een
Since $c\neq 1$, there is no real pole $z=\pm 1$ on $S_0$. The sheets where this pole is hosted can be found from $2l\pi\leq c\pi < 2(l+1)\pi$, giving
\be\label{SmcondB3}
l\leq\frac{p+2nq}{2(p-q)}=\frac{c}{2}< l+1\,.
\ee
We exclude the case of even $c=2l$, since the point $z=1$ lies in the branch cut on the positive semi-axis. When $c=2l+1$ is odd, we have a real pole \Eq{pole-} on the negative real semi-axis (i.e., away from the branch cut at $\Re\,z\geq 0$) and, since $p/q=(2n+2l+1)/(2l)$, $p$ is odd and $q$ is even. Fixing $n$, $p$ and $q$, one finds $l=(2n+1)q/[2(p-q)]$. 
\end{proof}

\subsubsection{Complex poles}\label{appC3b}

Contrary to all the other cases, in the NH theory only $S_0$ can host complex-conjugate pairs. To see this, we have to modify the discussion between \Eq{constrong1} and \Eq{mmst}. The poles $z_n$ and $z_{n_*}$ belong to, respectively, sheet $S_l$ and $S_{l_*}$ when
\ba
&& 2\pi l < \vp_n < 2\pi(l+1)\,,\label{inter1-B3}\\
&& 2\pi l_* < \vp_{n_*} < 2\pi(l_*+1)\,,\label{inter2-B3}
\ea
where we take strict inequalities because we exclude the self-conjugate case of real poles. Plugging relation \Eq{constrong1} into \Eq{inter1-B3}, we get
\be\label{inter3-B3}
2\pi (k-l-1) < \vp_{n_*} < 2\pi(k-l)\,.
\ee
Comparing \Eq{inter3-B3} with \Eq{inter2-B3}, we obtain
\be\label{mmst-B3}
l+l_*=k-1\,,
\ee
which replaces \Eq{mmst}. 

\begin{theo}[{\bf NH: complex-conjugate pairs ---real $\bm{0<\g<3}$}]\label{theo3-B}
When $\g$ lies in the interval \Eq{gulim-B1}, there are no complex-conjugate pairs $(z_n,z_n^*)$ in any Riemann sheet $S_{l}$.
\end{theo}
\begin{proof}
We follow the same steps as in Theorem~\ref{theo3-A}. From \Eq{c2pha3} and \Eq{inter1-B3}, $2l<(2n+1)/(\g-1)+1<2l+2$, hence
\ba
&&0<\g<1:\quad \frac{\g-1}{2}(2l+1)-\frac{1}{2}<n<\frac{\g-1}{2}(2l-1)-\frac{1}{2}\,,\\
&&\hspace{.7cm}\g>1:\quad  \frac{\g-1}{2}(2l-1)-\frac{1}{2}<n<\frac{\g-1}{2}(2l+1)-\frac{1}{2}\,.\label{condi-B3}
\ea
The length of the interval is $|\g-1|/2$, hence $N_l=0$ or $N_l=1$ when \Eq{gulim-B1} holds.
\end{proof}

Finally, the equivalent of Theorems~\ref{theo6-B1} and \ref{theo6-B2} is the following.
\begin{theo}[{\bf HN: empty sheets ---real $\bm{0<\g<3}$}]\label{theo8-B3}
There always exists a completely empty Riemann sheets $S_l$ for any non-integer $0<\g<3$ except when $q$ is odd and \Eq{exep} holds. For these values, there always is an unpaired complex pole in $S_{l\neq 0}$, while the sheet $S_0$ only has a complex-conjugate pair $(z_n,z_n^*)$.
\end{theo}
\begin{proof}
As in the HS and HA theories, the case $0<\g<1$ is trivial. Taking $\g>0$ from now on, the interval \Eq{condi-B3} 
 is centered at $(\g-1)l-1/2$ with half-length $0<(\g-1)/4<1/2$. Thus, when $l\neq 0$ there is no integer in it ($N_l=0$) if, and only if, the distance between $(\g-1)l-1/2$ and the closest integer $n$ is greater than or equal to $(\g-1)/4$:
\be\label{distcon0}
{\rm dist}\left(\g l-\frac12\right)\geq \frac{\g-1}{4}\,.
\ee
This inequality has the same form as \Eq{distcon-B1} and the same analysis follows through for $l\neq 0$, leading to \Eq{exep} for $q$ odd. When $l=0$, \Eq{distcon0} holds regardless of the value of $\g$ and a separate treatment is required, for instance as in Theorem~\ref{theo7-B3}. It turns out that $N_0=2$ and $S_0$ hosts a complex-conjugate pair.
\end{proof}

\begin{theo}[{\bf NH: complex-conjugate poles ---irrational $\bm{\g>0}$}]\label{theo4-B3}
For $\g>0$ irrational, the only Riemann sheet that can host complex poles is $S_0$ and such poles are always coupled into conjugate pairs $(z_n,z_{n_*})$.
\end{theo}
\begin{proof} From \Eq{c2pha3} and \Eq{constrong1},
\be\label{constrong2-B3}
n_*=-n-1+(\g-1)(k-1)\,,
\ee
which is an integer if, and only if,
\be
k=1\,:\qquad n_*=-n-1\,.\label{stro1-B3}
\ee
The condition \Eq{stro1-B3} describes two complex-conjugate poles with phases $\vp_n$ and $\vp_{-n-1}=-\vp_n+2\pi$, hence necessarily lying in the same Riemann sheet. This sheet is $S_0$, since \Eq{mmst-B3} yields $l_*=-l$ which is equal to $l$ if, and only if, $l=0$.
\end{proof}
For any rational or irrational $\g>0$, we can be more specific about complex-conjugate pairs of poles on the sheet $S_0$, which we denote as $(\bar z_n,\bar z_n^*)$ to distinguish them from generic pairs $(z_n,z_n^*)$ in other sheets:
\be
(\bar z_{n+1},\bar z_{n+1}^*)\coloneqq (z_n,z_{-n-1})=(z_n,z_n^*)\in S_0\,,\qquad n=0,\,\dots,\,\bar n-1\,,
\ee
where now $n$ is restricted to semi-positive integers and $\bar n$ is the number of pairs in $S_0$.
\begin{theo}[{\bf NH: complex poles on $\bm{S_0}$ ---real $\bm{\g>0}$}]\label{theo5-B3}
	For real $\g>0$:
	\begin{itemize}
	\item[a)] When $0<\g<2$, there are no complex poles on the Riemann sheet $S_0$.
	\item[b)] When $\g>2$, all complex poles on the Riemann sheet $S_0$ are organized into $\bar n\geq 1$ conjugate pairs such that
	\be\label{gk}
	2\bar n<\g<2\bar n+2\,,\qquad \bar n=1,\,2,\,\dots\,.
	\ee
	\end{itemize}
\end{theo}
\begin{proof}
In order to see how many complex poles $z_n$ lie in $S_0$, we set the conditions on $n$ such that $0<\vp_n<2\pi$. From \Eq{c2pha3},
\ba
&& 0<\g<1\,:\qquad \frac{\g}{2}-1<n<-\frac{\g}{2}\,,\label{condia-B3}\\
&& \hspace{.75cm}\g>1:\qquad \!\!\!\hphantom{-1}-\frac{\g}{2}<n<\frac{\g}{2}-1\,,\label{condib-B3}
\ea
where we took strict inequalities since we are ignoring real poles.

{\bf a)} Condition \Eq{condia-B3} never holds because the interval is a subinterval of $(-1,0)$ centered around $-1/2$. The same is true for condition \Eq{condib-B3} when $1<\g<2$. 

{\bf b)} When $\g>2$, condition \Eq{condib-B3} is satisfied for any $n$ such that
\be\label{ccpn}
-\bar n\leq n\leq \bar n-1\,,\qquad \bar n\in\mathbb{N}^+\,,
\ee
where $\bar n$ is the largest natural number smaller than $\g/2$ and is thus determined by \Eq{gk}. For any given $\bar n$, there are $2\bar n$ integer values of $n$ falling in the interval \Eq{ccpn}. However, according to \Eqq{stro1-B3}, also $n_*$ satisfies \Eq{ccpn} and the whole conjugate pair $(z_n,z_{n_*})=(z_n,z_n^*)$ is in $S_0$. Therefore, half of the values of $n$ inside \Eq{ccpn} correspond to a conjugate pole and there are in total $2\bar n/2=\bar n$ conjugate pairs in $S_0$:
\be
{\rm card}\big\{(\bar z_n,\bar z_n^*)\in S_0\big\}=\bar n\,.
\ee
This concludes the proof.
\end{proof}
Thus, if $2<\g<4$ ($\bar n=1$), then only $n=-1,0$ are allowed and $S_0$ hosts one conjugate pair. If $4<\g<6$ ($\bar n=2$), then $n=-2,-1,0,1$ and $S_0$ hosts two conjugate pairs. If $6<\g<8$ ($\bar n=3$), then $n=-3,-2,-1,0,1,2$ and $S_0$ hosts three conjugate pairs. And so on:
\ban
\hspace{-1.cm}&&2<\g<4\,:\qquad (\bar z_1,\bar z_1^*)=(z_0,z_{-1})\,.\\
\hspace{-1.cm}&&4<\g<6\,:\qquad (\bar z_1,\bar z_1^*)=(z_0,z_{-1})\,,\quad (\bar z_2,\bar z_2^*)=(z_1,z_{-2})\,.\\
\hspace{-1.cm}&&6<\g<8\,:\qquad (\bar z_1,\bar z_1^*)=(z_0,z_{-1})\,,\quad (\bar z_2,\bar z_2^*)=(z_1,z_{-2})\,,\quad (\bar z_3,\bar z_3^*)=(z_2,z_{-3})\,.\\
\hspace{-1.cm}&&\dots
\ean

\begin{theo}[{\bf NH: complex-conjugate poles ---rational $\bm{\g=p/q>0}$}]\label{theo6-B3}
For $\g=p/q>0$ rational and irreducible, complex-conjugate poles $z_n$ and $z_{n_*}$ can lie in the same Riemann sheet $S_l$ in the sequence \Eq{gooseq} only for $l=0$ and, if $q$ is even, also for $l=q/2$. Otherwise, $z_n$ and $z_{n_*}$ lie in different sheets $S_l$ and $S_{l_*}$ such that
\be\label{maincon-B3}
l+l_*=q\,.
\ee
\end{theo}
\begin{proof}
When $\g=p/q$ is rational, \Eq{constrong2-B3} which is an integer if, and only if, either \Eq{stro1-B3} or
\be
k=q+1\,:\qquad n_*=p-q-n-1\,,\label{stro2-B3}
\ee
holds. We already know from Theorem~\ref{theo4-B3} that \Eqq{stro1-B3} corresponds to a conjugate pair belonging to the sheet $S_0$, while \Eqq{stro2-B3} together with \Eq{mmst-B3} yield \Eq{maincon-B3}. Thus, to a pole in $S_1$ there would correspond its conjugate in $S_{q-1}$, the conjugate of a pole in $S_2$ would be in $S_{q-2}$, and so on. If $q$ is odd, then $q-1$ is even and the pairs are always split across the sheets of the sequence \Eq{gooseq'2} complex in this way. If $q$ is even, then \Eq{maincon-B3} is well defined also for $l=l_*$, in which case there may also be a complex-conjugate pair in the sheet $S_{q/2}$.
\end{proof}
Theorems \ref{theo1-B3}--\ref{theo6-B3} are instrumental to prove the following general rules, essentially stating that massless gauge and gravitational theories are well defined and well behaved only in $S_0$ when $\g$ is rational.
\begin{theo}[{\bf NH: distribution of poles ---rational $\bm{\g=p/q>0}$}]\label{theo7-B3}
For $\g=p/q>0$ rational and irreducible, inside the sequence \Eq{gooseq'2}:
\begin{itemize}
\item[a)] When $0<\g<2$ ($|p-q|<q$), $q-|p-q|$ of the $q$ sheets of the sequence have no poles, while the other sheets have either one tachyonic pole for a massless action (\Eq{eq:model} with $m^2=0$) or an isolated complex pole. $S_0$ is always empty.
\item[b)] When $\g=1$, all sheets are identical and there are no extra poles apart from the $z=0$ pole. When $\g=2$ ($|p-q|=q$), all sheets are identical and there is one extra pole apart from the $z=0$ pole.
\item[c)] When $\g>2$ ($|p-q|>q$), some sheets are populated by more than one pole. $S_0$ always hosts at least one pair of complex-conjugate poles and never contains real or unpaired complex poles. The other sheets in the sequence do.
\end{itemize} 
\end{theo}
\begin{proof}
{\bf a)} Let $0<\g<2$. From \Eqq{pqpoles2}, it is clear that the $|p-q|$ poles cannot populate all the $q$ sheets of a sequence. That the sheet $S_0$ is always empty of poles when $0<\g<2$ can be seen from \Eqq{pqpoles2}. The requirement $0\leq \vp_n<2\pi$ is never satisfied for $0<p/q<2$. In fact, if $0<p/q<1$, then $\vp_n\propto (p-q)^{-1}<0$ for $n\geq 0$, while for $n<0$ one has $(2|n|q-p)/(q-p)>2$. If $1<p/q<2$, then for $n\geq 0$ we get $(p+2nq)/(p-q)>2$, while for $n<0$ we have $\vp_n\propto p-2|n|q<0$.

If the poles are even, then there are $|p-q|/2$ complex-conjugate couples $(z_n,z_n^*)$. However, since they spread thin over the $q$ sheets, there is at most one pole in a sheet, leading to an unphysical theory. To see this, we note that we have an even number $|p-q|$ of poles only when both $p$ and $q$ are odd (examples: $\g=1/3,17/9$); they cannot be both even because $p/q$ is irreducible. These singularities are organized into $|p-q|/2$ pairs of complex-conjugate poles 
\ben
\tilde z_l\in \{z_n\}\setminus\{\bar z_n\}\,,\qquad l=\bar n+1,\bar n+2,\dots\,,
\een
in the sheets $S_1$, $\dots$, $S_{q-1}$, which are even in number. According to Theorem \ref{theo6-B3} and condition \Eq{maincon-B3}, the poles in each pair $(\tilde z_l,\tilde z_l^*)$, if present, are split in the sheets
\be\label{SSsch}
(S_1,S_{q-1}),\, (S_2,S_{q-2}),\, \dots,\, (S_{|p-q|/2},S_{|p-q|/2+1})\,.
\ee
If $p/q<1$, we have $q-(q-p)=p$ empty sheets in the sequence: $S_0$ and $p-1$ (even) sheets in the sub-sequence $\{S_1,\dots,S_{q-1}\}$, which are paired in the same conjugate way as above. If $1<p/q<2$, we have $q-(p-q)=2q-p$ empty sheets, $S_0$ and $2q-p-1$ (even) sheets in the sub-sequence. Thus, depending on the values of $p$ and $q$, we have 0, 2 ($S_{|p-q|/2}$ and $S_{|p-q|/2+1}$, the central sheets in the sub-sequence), 4 ($S_{|p-q|/2-1}$, $S_{|p-q|/2}$, $S_{|p-q|/2+1}$, $S_{|p-q|/2+2}$), $\dots$ regular (i.e., empty of poles) sheets in the sub-sequence. The singular sheets have complex poles not conjugated in a pair. 

If the poles are odd, then there are $(|p-q|-1)/2$ pairs of complex-conjugate poles plus a real pole we denote as $\tilde z_0=-1$. When $p$ is even and $q$ is odd (examples: $\g=2/5$, $2/3$, $4/5$, $4/3$, $16/9$), these pairs are distributed in the sub-sequence $\{S_1,\dots,S_{q-1}\}$ according to the above scheme $(S_1,S_{q-1})$, $(S_2,S_{q-2})$, $\dots$. The real pole $\tilde z_0$ is in the $l$-th sheet found from \Eq{SmcondB3}. The rest of $q-|p-q|-1$ sheets, odd in number, are empty. When $p$ is odd and $q$ is even (examples: $\g=1/4,1/2,3/4,3/2$), the same pattern occurs but there is the extra sheet $S_{q/2}$ in the sequence and the real pole $\tilde z_0$ is therein.

Since sheets with a real-valued pole have a tachyon and sheets with isolated complex poles lead to a non-real spectrum, a viable theory can be hosted only in regular sheets such as $S_0$.

{\bf b)} The integer cases $\g=1,2$ are easy to work out and we leave them to the reader.

{\bf c)} Let now $\g>2$. The condition $0\leq \vp_n<2\pi$ translates into
\be\label{pqondi}
\frac{p}{q}\geq -2n\qquad (n<0)\,,\qquad \frac{p}{q}>2(n+1)\,,
\ee
which are always guaranteed for $n=0,-1$. By the condition \Eq{stro1-B3}, the constraints \Eq{pqondi} are automatically valid also for $n'=-n-1$. Therefore, the pole $\bar z_1=z_0$ and its complex-conjugate $\bar z_1^*=z_{-1}$ always lie in the Riemann sheet $S_0$ and, if $p/q\gg 1$, there may be more conjugate pairs therein but no unpaired complex pole. Also, there is no real pole on $S_0$ unless $\g$ is integer, as shown in Theorems~\ref{theo1-B3} and \ref{theo2-B3}.

In all the other sheets in the sequence, there is either an odd or an even number of poles. Both cases are unphysical, the first for the reasons explained in Theorem~\ref{theo2-B3} (there always is a real tachyonic pole) and the second because complex-conjugate pairs are divided through different sheets according to the pattern established in Theorem~\ref{theo6-B3}. The real pole $\tilde z_0=-1$ is present only when $p$ is odd and is in one of these sheets together with complex poles; this sheet is $S_{q/2}$, if $q$ is even. Complex-conjugate poles are always paired in $S_0$ (where there is at least one pair, or more if $p\gg q$; see Theorem~\ref{theo5-B3}), while in the sub-sequence $\{S_1,\dots,S_{q-1}\}$ they are either split across different sheets or, when paired together, accompanied by the real pole $\tilde z_0$. If $p>2q$, then $p-q>q$ and there are more poles than sheets in the sequence. When $p$ is odd and $q$ is even (examples: $\g=5/2,7/2,9/2,11/2,17/2$), there is one real pole in $S_{q/2}$. When $p$ is even and $q$ is odd (example: $\g=10/3$) or when $p$ is odd and $q$ is odd (example: $\g=7/3$), there is no real pole.
The complex poles $\tilde z_j$ in the sub-sequence are distributed according to the scheme \Eq{SSsch}. On one hand, when $p/q$ is not much greater than 2, then $p-q= O(q)$ and in any sheet of the sub-sequence there can be a set of poles $\tilde z_{j}$, $\tilde z_{j+1}^*$, $\tilde z_{j+2}$, $\dots$ but conjugate pairs $(\tilde z_{j},\tilde z_{j}^*)$, $(\tilde z_{j+1},\tilde z_{j+1}^*)$, $\dots$ never combine in the same sheet, thus leading to a non-real spectrum. This is a consequence of \Eqq{stro2-B3}, which says that $\vp_n$ and $\vp_{n_*}$ are separated by $2\pi(q+1)>2\pi$. On the other hand, when $p/q\gg 2$, then the $p-q\gg q$ poles over-populate the sequence of sheets and tend to pair together on the same sheet. However, this always happens in the sheet with the real pole, thus leading to a spectrum with a tachyon.
\end{proof}


\end{document}